\documentclass[twoside,11pt]{article}
\usepackage{blindtext}

\RequirePackage{amsmath}
\RequirePackage{amssymb}
\RequirePackage{amsthm}
\RequirePackage{bm}
\RequirePackage{url}
\usepackage{multirow}
\usepackage{graphicx}
\usepackage{subfigure}
\usepackage{makecell}
\usepackage{booktabs}
\usepackage{array}
\usepackage{url}
\usepackage{algorithm}
\usepackage{algorithmic}
\usepackage{dsfont}
\usepackage{comment}
\usepackage{enumitem}
\usepackage{enumerate}
\usepackage[OT1]{fontenc}
\usepackage{natbib}
\usepackage[usenames]{color}

\long\def\comment#1{}

\newcommand{\red}{\color{red}}
\newcommand{\blue}{\color{blue}}

\let\emptyset\varnothing

\def\wh{\widehat}

\def\##1\#{\begin{align}#1\end{align}}
\def\$#1\${\begin{align*}#1\end{align*}}

\let\cite\citet

\let\hat\widehat
\let\tilde\widetilde
\newcommand{\bbeta}{\bm{\beta}}
\newcommand{\cP}{\mathcal{P}}

\newcommand{\bE}{\mathbb{E}}

\newcommand{\tb}{\boldsymbol{\widetilde{\beta}}}
\newcommand{\tx}{\boldsymbol{\widetilde{x}}}
\newcommand{\tX}{\boldsymbol{\widetilde{X}}}
\newcommand{\bb}{\boldsymbol{\beta}}
\newcommand{\bx}{\boldsymbol{x}}
\newcommand{\ba}{\boldsymbol{\alpha}}
\newcommand{\LG}{\boldsymbol{L}_{\mathcal{G}}}
\newcommand{\bz}{\boldsymbol{z}}
\newcommand{\bO}{\boldsymbol{O}}

\newcommand{\bX}{\boldsymbol{X}}
\newcommand{\cL}{\mathcal{L}}
\newcommand{\tC}{\widetilde{C}}
\newcommand{\ocL}{\overline{\mathcal{L}}}
\newcommand{\ob}{\overline{\bb}}
\newcommand{\oba}{\overline{\ba}}
\newcommand{\oa}{\overline{\alpha}}
\newcommand{\bc}{\overline{\boldsymbol{c}}}

\newcommand{\eb}{\mathbf{e}}

\newcommand{\xb}{\mathbf{x}}

\newcommand{\be}{\bm{e}}

\newcommand{\bw}{\bm{w}}

\newcommand{\Xb}{\mathbf{X}}

\newcommand{\cE}{\mathcal{E}}
\newcommand{\cG}{\mathcal{G}}

\newcommand{\cM}{\mathcal{M}}
\newcommand{\cO}{\mathcal{O}}

\newcommand{\cT}{{\mathcal{T}}}

\newcommand{\PP}{\mathbb{P}}

\newcommand{\RR}{\mathbb{R}}

\newcommand{\balpha}{\bm{\alpha}}

\newcommand{\bOmega}{\bm{\Omega}}

\newcommand{\argmin}{\mathop{\mathrm{argmin}}}

\newcommand{\sign}{\mathop{\mathrm{sign}}}
 \usepackage[abbrvbib, preprint]{jmlr2e}

\usepackage{lastpage}
\ShortHeadings{Uncertainty Quantification for Ranking with Covariates}{Fan, Hou and Yu}
\firstpageno{1}

\begin{document}

\title{Uncertainty Quantification of MLE for Entity Ranking
with Covariates}

\author{\name Jianqing Fan \email jqfan@princeton.edu \\
\name Jikai Hou \email jikaih@princeton.edu \\
       \addr Department of Operations Research and Financial Engineering\\
       Princeton University\\
    Princeton, NJ, United States\\
\name Mengxin Yu \email mengxiny@wharton.upenn.edu \\
       \addr Department of Statistics and Data Science, the Wharton School\\
      University of Pennsylvania\\
    Philadelphia, PA, United States
}

\editor{My editor}

\maketitle

\begin{abstract}%   <- trailing '%' for backward compatibility of .sty file
	We study statistical estimation and inference for the ranking problems based on pairwise comparisons with additional covariate information.  %Despite extensive studies, few prior literatures investigate this problem when covariate information exists. 
In specific, in this paper, we study a Covariate-Assisted Ranking Estimation (CARE) model in a systematic way, that extends the well-known Bradley-Terry-Luce (BTL) model by incorporating the covariate information. We impose natural identifiability conditions, derive the statistical rates for the MLE under a sparse comparison graph, and obtain its asymptotic distribution. Moreover, we validate our theoretical results through large-scale numerical studies.
\end{abstract}

\begin{keywords}
  %BTL model, Entity ranking, Ranking with covariates, Uncertainty quantification, Maximum likelihood estimator
  {High-Dimensional Inference}, {Entity ranking}, {Ranking with covariates}, {Uncertainty quantification}, {Maximum likelihood estimator}.
\end{keywords}

\section{Introduction}

Ranking plays an essential role in many real-world applications. For example, it is crucial in individual choice \citep{luce2012individual}, psychology \citep{thurstone1927method, thurstone2017law}, recommendation systems \citep{baltrunas2010group, li2019estimating}, and many others. The ranked items such as sports teams \citep{massey1997statistical, turner2012bradley}, scientific journals \citep{stigler1994citation}, web pages  \citep{dwork2001rank}, election candidates \citep{plackett1975analysis}, or even movies  \citep{harper2015movielens} will not only illustrate their qualities but also affect people's future choices.  Thus, the ranking problem  has been extensively studied in statistics, machine learning, operations research, etc.; see, for example, \citep{hunter2004mm,richardson2006beyond,jang2018top,chen2019spectral,chen2022partial,chen2022optimal,liu2022lagrangian} for more details. 

Among various  models for the ranking problem, the most well-known one is the Bradley-Terry-Luce (BTL) model \citep{bradley1952rank,luce2012individual}, which assumes the existence of scores $\{\theta_i^*\}_{i=1}^{n}$ of $n$ compared items such that the preference between item $i$ and item $j$ is given by 
$$
	\mathbb{P}(\textrm{item $j$ is preferred over $i$}) = e^{\theta_j^*} / (e^{\theta_i^*} + e^{\theta_j^*}),  \qquad \mbox{for } (i,j)\in[n]\times[n].
$$
The underlying assumption of this BTL model is the scores of compared items are fixed and do not explicitly use their attributes. However, in many real-world applications, covariate information often exists and this heterogeneity needs to be incorporated.
For example, US News and Times Higher Education consider many characteristics of universities, such as international research reputation, teaching quality, the ratio between students and professors, and citations to conduct global university rankings. In addition, in NBA basketball competitions, the final rank of a team is also affected by its underlying attributes, such as the ability to defend,  make a three-point shot, etc. 
%However, existing prior arts that study the BTL model mostly focus on doing ranking without incorporating the features of compared items. 

%{\blue There exist some previous literature that consider incorporating covariate information into the ranking problems. To name a few, \cite{turner2012bradley,guo2018experimental,schafer2018dyad,zhao2022learning,li2022bayesian}.    }

Thus, a crucial question still remains open:
    \begin{quote}
      ``\emph{Can one design a provably efficient mechanism for ranking by incorporating features of compared items and conduct associated high-dimensional statistical inference?}"
\end{quote}

 To this end, we follow the idea from related literature \citep{turner2012bradley,li2022bayesian}, by incorporating feature information of items into the BTL model, and call the model the Covariate-Assisted Ranking Estimation (CARE) model. Specifically,  we address covariate heterogeneity by assuming the underlying score (ability) of the $i$-th item is given by $\alpha_i^*+\bx_i^\top \bbeta^*$, where $\bx_i^\top \bbeta^*$ captures the covariate effect and $\alpha_i^*$ is the intrinsic score that cannot be explained by the covariate.
In this case, the outcome of pairwise comparison is modeled as  
 $$
 	\mathbb{P}(\textrm{item $j$ is preferred over $i$}) = \frac{e^{\alpha_j^*+\bx_j^\top\bbeta^*} }{ e^{\alpha_i^*+\bx_i^\top\bbeta^*} + e^{\alpha_j^*+\bx_j^\top\bbeta^*}}.
$$ 

We do not assume that all pairs are compared.  Rather, each pair is selected at random for comparison. In specific, we let
the underlying comparison graph be a realization of the Erdős-Rényi random graph with edge probability $p$. 
In addition, once a pair is selected, they are compared $L$ times. This can be two teams matching $L$ times or two universities ranked $L$ times by experts. In this work, we consider the fixed design in the sense that the randomness only comes from results of comparisons.

%{\blue Although, there exists some other literature that also consider incorporating covariate information into the ranking problems, \cite{guo2018experimental,schafer2018dyad,zhao2022learning,chau2023spectral,finch2022introduction}, their motivations, model settings and methodology or theoretical contributions are different from us. In specific, most of them assume the underlying scores of all compared items are fully explained by covariates without studying the effects of individual intrinsic scores. In addition, they do not assume the comparisons are realized from random graphs.}  %Moreover, they only provide an $\ell_2$ estimation bound without considering statistical inference for the unknown parameter $\beta^*$. This prevent one to infer rank information of any specific item.  }

There are several challenges in studying statistical inference for our CARE model. First, our model incorporates feature information into the original BTL model, not only the underlying scores $\{\alpha_i^*\}_{i=1}^{n}$ but also $\bbeta^*$ shall be estimated and analyzed in a novel way. This also gives rise to the issue of identifiability.  Second, given consistent estimators, it remains open to quantifying these key components' uncertainty. 
Most existing work focuses more on deriving statistical rates of convergence for those underlying scores via various estimators in the BTL model to achieve specific rank recoveries such as top-K and partial recovery \citep{chen2019spectral,chen2022partial}. There are few results established for the inference  of the BTL model \citep{simons1999asymptotics,han2020asymptotic,gao2021uncertainty,liu2022lagrangian}, letting alone the uncertainty quantification for the more general BTL model with covariates (CARE model).

In our work, we resolve the first challenge by designing a novel constrained maximum likelihood estimator (MLE) $( \hat \ba_M, \hat \bbeta_M)$
which efficiently estimates the underlying scores $\{\alpha_i^*\}_{i=1}^{n}$ and $\bbeta^*$.  With some proper initialization, the MLE can be solved by simply running the projected gradient descent algorithm.  By leveraging the `leave-one-out' technique \citep{chen2019spectral}, we prove that the statistical rate of convergence of the MLE of the intrinsic scores $\{\alpha_i^*\}_{i=1}^{n}$ and overall scores $\{\alpha_i ^* +\mathbf{x}_i^\top\bbeta^*\}_{i=1}^{n}$ in $\ell_{\infty}$-norm are of order $\mathcal{O}(\sqrt{\log n/npL}),$ and $\mathcal{O}(\sqrt{(d+1)\log n/npL})$, respectively, in which $d$ is the dimension of the observed covariates. These statistical rates reduce to the standard minimax rates for estimating the BTL model when no covariate exists,  \citep{chen2019spectral}. 

To take on the second challenge, namely, depicting the asymptotic distribution of the MLE,  we first approximate the MLE by the minimizer of the quadratic approximation of our joint likelihood function, whose uncertainty is easier to depict. The critical difficulty lies in quantifying this approximation error. To tackle this issue, we then utilize the `leave-one-out' technique and derive \emph{novel proofs}, which is valid under the minimal sample complexity up to logarithm terms.   In  a more specific BTL model, the seminal works by \cite{gao2021uncertainty} and \cite{fan2022ranking} (when considering pairwise comparison) leverage the minimizers of the more restricted diagonal quadratic approximations of their marginal likelihoods to approximate the MLE. They capture the approximation errors based on a `leave-two-out' technique. %However, this asymptotic independence property does not hold in our setting when covariate information is incorporated. 
In contrast, in this work, we utilize the minimizer of the quadratic approximation of the joint likelihood to approximate the MLE, and achieve a tighter approximation error than \cite{gao2021uncertainty,fan2022ranking}.

Finally, we conduct numerical experiments to corroborate our theory.  The performance of the model is also convincingly illustrated by an analysis of the data on Pokemon competitions. From the perspective of stock selection and return prediction, our proposed covariate-assisted BTL model (CARE) outperforms the original BTL model in many aspects.

To summarize, the contributions of this work are  of multiple folds. First,  we study a Covariate-Assisted Ranking Estimation (CARE) model in a systematic way that extends the well-known Bradley-Terry-Luce (BTL) model by incorporating the covariate information.
Specifically, we derive $\ell_{\infty}$- and $\ell_2$- statistical rates for the MLE of $\{\alpha_i^*\}_{i=1}^{n}$ and $\bbeta^*$, respectively.  Moreover, we also conduct uncertainty quantification for our MLE, where we improve the approximation errors given in existing works and derive more general asymptotic results. Furthermore,  our results hold even on the sparse comparison graph, i.e. the probability of pairwise comparison $p \asymp  1/n$ up to logarithm terms, with minimal sample complexity. Finally, we illustrate our methods via  large-scale numerical studies on  synthetic and real data. Numerical results lend further support to our proposed CARE model over the original BTL model.

\subsection{Prior Arts}

Ranking problems based on pairwise comparison for parametric and  non-parametric models have received much attention. For  the BTL model,   \cite{hunter2004mm} studies its variants and establishes theoretical properties using a minorization-maximization algorithm. \cite{chen2015spectral} use a two-step method to study the BTL model, which is provably optimal in terms of sample complexity. \cite{jang2016top} leverage  the spectral method to recover the top-K items with only minimal samples.  In addition, \cite{negahban2012iterative} propose an iterative rank aggregation algorithm named Rank Centrality to recover the underlying scores of the BTL model in optimal $\ell_2$- statistical rate. In the sequel, \cite{chen2019spectral} derive both $\ell_2$- and $\ell_{\infty}$- optimal statistical rates of those underlying scores and prove that  the regularized MLE and spectral method are both optimal for recovering top-K items when the conditional number is a constant. Furthermore, \cite{chen2022partial} proves that for partial recovery, MLE is optimal, but the spectral method is not when we have a general conditional number. It is worth noting that the aforementioned prior arts mainly focus on studying the parametric BTL model. There is also a series of works that studies specific non-parametric variants of the BTL model. For instance, \cite{shah2017simple} develop a counting-based algorithm to recover top-$K$ ranked items under the nonparametric stochastically transitive model. For more details on the non-parametric comparison models, see \cite{shah2016stochastically,shah2017simple,chen2017competitive,pananjady2017worst} and the references therein.

Going beyond the pairwise comparison, there also exist other works that study ranking problems using $M$-way comparisons $(M\ge 2)$. The first well-known model is the Plackett-Luce model and its variants \citep{plackett1975analysis,guiver2009bayesian,cheng2010label,hajek2014minimax,maystre2015fast,szorenyi2015online,jang2018top,fan2022ranking}. For instance, a closely related work is \cite{jang2018top}, who study the Plackett-Luce model under a uniform hyper-graph. They divide $M$-way compared data into pairs and utilize the spectral method to derive the $\ell_{\infty}$- statistical rate of underlying scores. They further provide a lower bound for sample complexity to recover top-K items in the Plackett-Luce model.  Another well-known model is the Thurstone model \citep{thurstone1927method}, which admits the Plackett-Luce model as a particular case; see \cite{thurstone1927method,guiver2009bayesian,hajek2014minimax,vojnovic2016parameter,jin2020rank} for more details. 

The aforementioned literature mainly focuses on non-asymptotic statistical consistency results for the underlying scores of compared items through various ranking frameworks. However, the limiting distributional results for ranking models still remain highly under-explored. There are several results on the asymptotic distributions for the ranking scores in the BTL model. For instance, \cite{simons1999asymptotics} derive the asymptotic normality of the MLE of the BTL model in the scenario where all pairs of comparison are fully conducted (i.e., $p=1$).   \cite{han2020asymptotic} further extend the results  to the regime  where the comparison graph (Erdős-Rényi random graph) is dense but not fully connected, i.e., $p \gtrsim n^{-1/10}$. In addition, recently, \cite{liu2022lagrangian} propose a Lagrangian debiasing method to derive asymptotic distribution for ranking scores, where they allow sparse comparison graph $p \asymp  1/n$ but require comparison times $L$ to be larger than $n^2.$
 Moreover, \cite{gao2021uncertainty} utilize a `leave-two-out' trick to derive asymptotic distributions for ranking scores with optimal sample complexity in the regime where the comparison graph is sparse, i.e., $p \asymp  1/n$ up to logarithm terms. 
 
All aforementioned models and methods mainly study the estimation and uncertainty quantification for ranking models without considering any individual feature information. Yet, covariate data exist in most applications, and this results in additional challenges in technical derivations and computation. Although there exists some other literature that also considers incorporating covariate information into the ranking problems \citep{guo2018experimental,schafer2018dyad,zhao2022learning,chau2023spectral,finch2022introduction}, their motivations, model settings, methodology, and theoretical contributions are different from us. In specific, most of them assume the underlying scores of all compared items are fully explained by covariates without studying the effects of individual intrinsic scores (i.e., no $\alpha_i^*$). In addition, we allow the comparisons to be realized through (sparse) comparison graphs, which take on extra challenges. Moreover, in terms of the theoretical contribution, most of them only establish the $\ell_2$- statistical rates for estimating $\bbeta^*$ whereas we not only obtain $\ell_{\infty}$-and $\ell_2$- statistical rates for estimators of $\{\alpha_i\}_{i=1}^n$ and $\bbeta^*$ but quantify their uncertainty as well.
In addition, even though our CARE model is related to \cite{turner2012bradley,li2022bayesian}, we propose a systematic model estimation and inference framework. In contrast, these previous works only formulate ranking with covariates intuitively and do not discuss methodological implementations or theoretical guarantees.

Therefore, this paper takes up this challenge by presenting a systematic framework for model estimation and uncertainty quantification of our CARE model over a random comparison graph. Notably, this framework admits all previous advancements made on the BTL model, which do not incorporate covariate information as special cases. 

\subsection{Notation}
We introduce some useful notations before proceeding. We denote by $[M] = \left\{1,2,\dots, M\right\}$ for any positive integer $M$. For any vector $\mathbf{u}$ and $q\ge 0$, we use $\|\mathbf{u}\|_{\ell_q}$ to represent the vector $\ell_q$ norm of $\mathbf{u}$. In addition, the inner product $\langle \mathbf{u}, \mathbf{v} \rangle $ between any pair of vectors $\mathbf{u}$ and $\mathbf{v}$ is defined as the Euclidean inner product $\mathbf{u}^\top\mathbf{v}$. For vector $\mathbf{u}\in \mathbb{R}^m$ and index $i\in [m]$, we denote by $\mathbf{u}_{-i}$ the vector we get by deleting the $i$-th element in $\mathbf{u}$. For any given matrix $\mathbf{X}\in\mathbb{R}^{d_1\times d_2}$, we use $\|\mathbf{X}\|$, $\|\mathbf{X}\|_{F}$, $\|\mathbf{X}\|_{*}$ and $\|\mathbf{X}\|_{2,\infty}$ to represent the operator norm, Frobenius norm, nuclear norm and two-to-infinity norm of matrix $\mathbf{X}$ respectively. Moreover, we use  $\mathbf{X}\succcurlyeq 0$ or $\mathbf{X}\preccurlyeq 0$ to denote positive semidefinite or negative semidefinite
of matrix $\mathbf{X}$.
Moreover, we use the notation $a_n\lesssim b_n$ or $a_n = O(b_n)$ for non-negative sequences $\left\{a_n\right\}$ and $\left\{b_n\right\}$ if there exists a constant $\nu_1$ such that $a_n\leq \nu_1 b_n$. We use the notation $a_n\gtrsim b_n$ for non-negative sequences $\left\{a_n\right\}$ and $\left\{b_n\right\}$ if there is a constant $\nu_2$ such that $a_n\geq \nu_2 b_n$. For simplicity, we define function $\displaystyle\phi(t) = {e^t}/({e^t+1})$. We write $a_n\asymp b_n$ if $a_n\lesssim b_n$ and $b_n\lesssim a_n$. For matrix $\boldsymbol{A}$, we denote by $\boldsymbol{A}^{+}$ its pseudoinverse \citep{banerjee1973generalized}. 
%\begin{itemize}
%	\item {\red $\ell$ or $l$? }
%	\item  {\red  $O(\cdot )$ to $\cO(\cdot)$}
%\end{itemize}

\subsection{Roadmap}

	{The remaining paper is organized as follows. We describe the problem formulation for our BTL model with covariate  and derive the corresponding maximum likelihood estimators for its involved parameters in \S \ref{sec:algo}. In \S \ref{Consistency}, we establish  the statistical estimation results for our MLE. In \S\ref{Inference}, we further conduct uncertainty quantification for the obtained MLE. In \S\ref{Experiments}, we corroborate our theoretical results by conducting large-scale numerical studies via both synthetic and real dataset.% numerical experiments.

	}

\section{Model Formulation} \label{sec:algo}
In this section, we introduce the Covariate-Assisted Ranking Estimation (CARE) model, which incorporates covariate information into the BTL model.
In the traditional BTL model \citep{bradley1952rank,luce2012individual}, it is assumed that  each item $i\in[n]$ has a latent score $\theta_i^*$  and the outcomes of comparisons  are modeled as the realizations of the Bernoulli trials:
\begin{align}\label{comp}
    \mathbb{P}(\text {item } j \text { is preferred over item } i)=\frac{e^{\theta_j^*}}{e^{\theta_j^*}+e^{\theta_i^*}},\quad \forall 1\leq i\neq j \leq n.
\end{align}
It is worth mentioning that the function $\exp(\cdot)$ in \eqref{comp} can be replaced by any increasing differentiable functions.

In many applications, one observes individual features $\bx_i\in\mathbb{R}^d$ and would like to incorporate them for conducting more accurate ranking. As an extension of the parameterization $ \exp(\theta_i^*)$ \citep{chen2019spectral,chen2022partial},  we model the scores $\exp(\theta_i^*)$ as $\exp\left(\alpha_i^*+\bx_i^\top\bb^*\right)$ for $1\leq i\leq n$. The linear term $\bx_i^\top\bb^*$ captures the part of the scores explained by the variables $\bx_i$   and $\alpha_i^*$ represents the intrinsic score  that can not be explained by the covariate $\bx_i$.  %(or the part of the scores explained by variables other than $\bx_i$ that we don't observe). %These variables $\alpha_i^*$ may also contain the part of the scores explained by variables other than $\bx_i$ that we don't observe. 
This leads to modeling the outcomes of comparisons as the Bernoulli trials with probabilities
\begin{align}
    \mathbb{P}(\text {item } j \text { is preferred over item } i)=\frac{e^{\alpha_j^*+\bx_j^\top \bb^*}}{e^{\alpha_i^*+\bx_i^\top \bb^*}+e^{\alpha_j^*+\bx_j^\top \bb^*}} ,\quad \forall 1\leq i\neq j \leq n. \label{CARE}
\end{align}
We call this model Covariate Assisted Ranking Estimation (CARE) model.
%In our model, the linear term $\bx_i^\top\bb^*$ captures the part of the scores which can be explained by the variables $\bx_i$ for $1\leq i\leq n$. For each item $i$, we also have a variable $\alpha_i^*$ which captures the intrinsic score of item $i$. These variables $\alpha_i^*$ may also contain the part of the scores explained by variables other than $\bx_i$ that we don't observe. In the later sections, we will conduct uncertainty quantification for the estimator of variables $\alpha_i^*$ for $1\leq i\leq n$ and $\bb^*$, which is substantially different from the inference of the original BTL model.

We do not assume that all pairs are compared, but  only those in the comparison graph $\mathcal{G}=(\mathcal{V}, \mathcal{E})$. Here  $\mathcal{V} := \{1,2,\dots, n\}$ and $\mathcal{E}$ represent the collections of vertexes ($n$ items) and edges, respectively.  More specifically, $(i,j)\in \mathcal{E}$ if and only if item $i$ and item $j$ are compared. Throughout our paper,  the comparison graph is assumed to follow the Erdős-Rényi random graph $\mathcal{G}_{n,p}$ \citep{erdos1960evolution} where each edge appears independently with probability $p$.  In short, items i and j with $(i,j)\in[n]\times[n]$ are compared at random with probability $p$.

In addition, for any $(i,j) \in \mathcal{E}$,  we observe $L$  independent and identically distributed realizations from the Bernoulli random variables
\begin{align*}
   P\Big( y_{i, j}^{(l)} = 1\Big)  = \frac{e^{\alpha_j^*+\bx_j^\top \bb^*}}{e^{\alpha_i^*+\bx_i^\top \bb^*}+e^{\alpha_j^*+\bx_j^\top \bb^*}}.
\end{align*}

%It is straightforward that 
%We denote $\displaystyle y_{i,j} = \frac{1}{L}\sum_{l=1}^L y_{i,j}^{(l)}$ as the averaged  sufficient statistics are given by $\left\{y_{i,j}\mid (i,j)\in \mathcal{E}\right\}$.

%\subsection{Maximum Likelihood Estimator (MLE)}\label{mle}
%In this subsection, we estimate those involved parameters in our model via maximum likelihood estimation.
%Now we study a model with parameters $\alpha_1,\alpha_2,\dots,\alpha_n, \bb$. We are going to use maximum likelihood estimator to estimate unknown parameters $\alpha_1^*,\alpha_2^*,\dots, \alpha_n^*,\bb^*$. 
%Before proceeding, we first introduce some related notations. Specifically, we 
Let $\tx_i = \left(\boldsymbol{e}_i^\top, \bx_i^\top\right)^\top$ and $\tb = \left(\ba^\top,\bb^\top\right)^\top$, where $\{\boldsymbol{e}_i\}_{i=1}^n$ stand for the canonical basis vectors in $\mathbb{R}^n$ and $\ba = (\alpha_1,\alpha_2,\dots,\alpha_n)^\top\in\mathbb{R}^n.$ %We also set $\ba = (\alpha_1,\alpha_2,\dots,\alpha_n)^\top\in\mathbb{R}^n$ as the vector consists of the first $n$ elements of $\tb$, so we have $\tb= (\ba^\top,\bb^\top)^\top$ according to this notation. 
Then, the log-likelihood function conditioned on $\mathcal{G}$ is given by
\begin{align*}
    L\cdot \sum_{(i, j) \in \mathcal{E}, i>j}\left\{y_{j, i} \log \frac{e^{\tx_i^\top\tb}}{e^{\tx_i^\top\tb}+e^{\tx_j^\top\tb}}+\left(1-y_{j, i}\right) \log \frac{e^{\tx_j^\top\tb}}{e^{\tx_i^\top\tb}+e^{\tx_j^\top\tb}}\right\}.
\end{align*}
where $ y_{i,j} = \frac{1}{L}\sum_{l=1}^L y_{i,j}^{(l)}$ is a sufficient statistic.
{The identifiability question arises naturally since we over-parametrized the problem.  To remedy this issue, we restrict the parameter space of $\tb$ onto some constrained set $\Theta$ with a natural interpretation.  
%Then,
%any $\tb\in\mathbb{R}^{n+d}$ can be decomposed into two components $\tb = \tb_1+\tb_2$ with $\tb_1\in \Theta$ and $ \tb_2\perp\Theta$. Since $(\tx_i-\tx_j)^\top\tb_2 = 0, \forall i,j\in[n]$, the component $\tb_2$ will never affect the value of likelihood. 
 In specific, we denote $\bar{\bx}_i=[1,\bx_i^\top]^\top, \forall i\in[n],$ let $\bar{\bX}=[\bar{\bx}_1,\cdots,\bar{\bx}_n]^\top\in \RR^{n\times (d+1)}$ and consider the constrained set $\Theta=\{(\balpha,\bbeta):\bar{\bX}^\top\balpha=\mathbf{0}\}$. Throughout the paper, we assume that $\bar{\bX}$ has rank $d+1$. Under these identifiability constraints, if the true parameter vector $\tb^*=\left(\alpha_1^*,\alpha_2^*,\dots,\alpha_n^*,\bb^{*\top}\right)^\top = \left(\ba^{*\top},\bb^{*\top}\right)^\top\in \Theta$, the identifiability condition implies that $\sum_{i=1}^n \alpha_i^* = 0$ and $\bX^\top\balpha^*=0$ with $\mathbf{X}=[\bx_1,\cdots,\bx_n]^\top\in \RR^{n\times d}$.  {It admits clear interpretation:  $\bX \bbeta^*$ represents the scores that be captured by the covariates, whereas the $\balpha^*\in \RR^n$ represents the residual scores (or equivalently intrinsic scores) that can not be explained by the involved features (i.e., it is orthogonal to the linear space spanned by the columns of covariates).} {Next, we prove the identifiable property of $\Theta$ rigoriously in Proposition \ref{prop_iden}. }%The following proposition guarantees that $\Theta$ is identifiable. 
{
\begin{proposition}\label{prop_iden}
	CARE model Eq.~\eqref{CARE}  with parameter space $\Theta = \{ (\ba,\bb):\bar{\bX}^\top\ba = \boldsymbol{0}   \}$ is identifiable.
\end{proposition}

}

%\textcolor{red}{To do: Add a Lemma to this identification condition.}

We denote by $\boldsymbol{Z} \in\mathbb{R}^{(n+d)\times (d+1)}$ a matrix by padding   $\boldsymbol{0} \in \mathbb{R}^{d\times (d+1)}$ matrix to $\bar{\bX}$, i.e. $\boldsymbol{Z}_{1:n,\cdot} = \bar{\bX}$ and $\boldsymbol{Z}_{n+1:n+d,\cdot} = \boldsymbol{0}$. As a result, $\Theta$ can be also written as $\{\tb\in \mathbb{R}^{n+d}:\boldsymbol{Z}^\top\tb = \boldsymbol{0}\}$.  Denote by $\cP = \boldsymbol{I}-\boldsymbol{Z}(\boldsymbol{Z}^\top\boldsymbol{Z})^{-1}\boldsymbol{Z}^\top$ the projection matrix onto space $\Theta$.
\iffalse
The identifiability condition of $\Theta$ is proved  by the following Proposition \ref{pro1}. %linear space. inside the linear space $\Theta$ which is proved to be identifiable by the following Proposition \ref{pro1}. %The following proposition also ensures that $\Theta$ is identifiable.
\begin{proposition}\label{pro1}
The parameter space $\displaystyle\Theta = \text{span}\{\tx_i-\tx_j:i,j\in[n]\}\subset \mathbb{R}^{n+d}$ is identifiable.
\end{proposition}
\fi}
%\noindent As a result, without loss of generality we can project the true parameter vector onto $\Theta$ and assume that the true parameter vector $\tb^*=\left(\alpha_1^*,\alpha_2^*,\dots,\alpha_n^*,\bb^{*\top}\right)^\top = \left(\ba^{*\top},\bb^{*\top}\right)^\top\in \Theta$. 
\iffalse
\begin{assumption}\label{ass0}[Incoherence Condition]
Assume that there exists a constant $c_0$ such that % change the word
\begin{align*}
   \|\cP_{\bar{\bX}}\|_{2,\infty}=\|\bar{\bX}(\bar{\bX}^\top\bar{\bX})^{-1}\bar{\bX}^\top\|_{2,\infty}\le c_0\sqrt{(d+1)/n}.
\end{align*}
\end{assumption}
In Assumption \ref{ass0}, we assume that the $\|\cdot\|_{2,\infty}$-norm of the projection matrix $\cP_{\bar{\bX}}$ is of order $c_0\sqrt{(d+1)/n}.$ To demonstrate the rationality behind this, we first observe that the $\|\cP_{\bar{\bX}}\|_{F}^2\leq d+1$. Therefore, a sufficient condition for this assumption is when $\cP_{\bar{\bX}}$ satisfies the incoherence condition. Namely, its rows are balanced. When there does not exist the covariate (i.e. $\bar{\bX}=\mathbf{1}$), we have $\cP_{\bar{\bX}}=\mathbf{1}\mathbf{1}^\top/n.$ In this scenario, this assumption holds automatically with $c_0=1.$ The following results of this paper are established on this incoherence condition.
\fi

Given the aforementioned identifiable condition, we consider the following constrained maximum likelihood estimator (MLE)
\begin{align}
    \tb_M = \argmin_{\tb\in\Theta} \mathcal{L}(\tb),  \label{MLEun}
\end{align}
where
\begin{align}
	\mathcal{L}(\tb): 
	&=\sum_{(i, j) \in \mathcal{E}, i>j}\left\{-y_{j, i}\left(\tx_i^\top\tb-\tx_j^\top\tb\right)+\log \left(1+e^{\tx_i^\top\tb-\tx_j^\top\tb}\right)\right\}.\label{nll}
\end{align}
Note that when there is no covariate $\{\bx_i\}_{i=1}^{n}$, we have $\Theta=\{\balpha\in \RR^n,\mathbf{1}^\top \balpha=\mathbf{0}\}$ and the scores are identifiable up to a constant shift. Therefore, our formulation includes those studies of the BTL model without covariate information as special cases \citep{chen2019spectral}.

The inference question arises naturally if some covariates can explain the underlying scores, namely if some or all components of $\bbeta^*$ are statistically significant.  Similarly, one might ask if the covariates are adequate for determining the underlying scores by testing whether some or all components of $\balpha^*$ are zero.  In general, we would expect the variations among the components of $\balpha^*$ to be smaller than the original scores $\{\theta_i^*\}_{i=1}^n$. This enables us to improve data predictions by shrinking or regularizing the estimate of $\balpha^*.$ %This allows us to get a more plausible estimation of $\balpha^*$ via shrinkage or regularization and then enhance real data predictions.   %For example, we can impose additional sparsity assumptions on $\balpha^*$, which amounts to assuming that the covariates capture most of underlying latent scores.

{In the following context, we rescale $\bx_i$ to $\bx_i / K$, where $K>0$ is a positive number such that $\| \bx_i\|_2 \leq \sqrt{(d+1)/n}$ for all $\bx_i$ after the transformation.
The likelihood function, prediction and the column space spanned by $\bar{\bX}$ are not affected by the scaling  but this normalization facilitates us with scaling issues in the technical derivations.  Therefore, the content that follows will be based on the scaled  data and parameters.}

\section{Rate of Convergence of Maximum Likelihood Estimator}\label{Consistency}

%In this section, we show the statistical consistency results for the maximum likelihood estimator $\tb_M$ in   \eqref{MLEun}. As mentioned in \S\ref{mle}, \eqref{MLEun} is solved based on gradient descent. Therefore, we start 
In this section, we show the statistical consistency results for the maximum likelihood estimator $\tb_M$ in   \eqref{MLEun}.  Before proceeding to the main results, we begin by introducing several key assumptions on the design matrix. First, we assume the projected matrix $\cP_{\bar{\bX}} := \bar{\bX}(\bar{\bX}^\top\bar{\bX})^{-1}\bar{\bX}^\top$ satisfies the following incoherence condition.
\begin{assumption}\label{ass0}[Incoherence Condition]
Assume that there exists a positive constant $c_0$ such that % change the word
\begin{align*}
   \|\cP_{\bar{\bX}}\|_{2,\infty}=\|\bar{\bX}(\bar{\bX}^\top\bar{\bX})^{-1}\bar{\bX}^\top\|_{2,\infty}\le c_0\sqrt{(d+1)/n}.
\end{align*}
\end{assumption}

%In Assumption \ref{ass0}, we assume that the $\|\cdot\|_{2,\infty}$-norm of the projection matrix $\cP_{\bar{\bX}}$ is of order $c_0\sqrt{(d+1)/n}.$ 
To demonstrate the rationality behind Assumption \ref{ass0}, we first note that the $\|\cP_{\bar{\bX}}\|_{F}^2\leq d+1$. Therefore, a sufficient condition for this assumption to hold is when the rows of $\cP_{\bar{\bX}}$ are nearly balanced, %(this condition can be satisfied up to logarithm terms with high-probablity, when the covariates are previously drawn i.i.d. from many distributions, e.g, standard Gaussian), 
with row sum of squares all of the order $(d+1)/n$ or smaller. When there does not exist the covariate (i.e. $\bar{\bX}=\mathbf{1}$), we have $\cP_{\bar{\bX}}=\mathbf{1}\mathbf{1}^\top/n.$ In this scenario, this assumption holds automatically with $c_0=1.$ The following results of this paper are established under this incoherence condition.

We next introduce a key assumption on the covariates $\bx_i$  which guarantees a well-behaved landscape of the loss function as well as good statistical properties of the MLE estimator.  In specific, we put the following assumption on $\boldsymbol{\Sigma} = \sum_{i>j}(\tx_i-\tx_j)(\tx_i-\tx_j)^\top$.  

\begin{assumption}\label{ass1}
Assume that there exist positive constants $c_1$ and $c_2$ such that % change the word
\begin{align*}
    c_2n\leq \lambda_{\text{min}, \perp}(\boldsymbol{\Sigma})\leq \Vert \boldsymbol{\Sigma}\Vert\leq c_1n,
\end{align*}
where $\Vert \boldsymbol{\Sigma}\Vert$ is the operator norm of $\boldsymbol{\Sigma}$ and  
\begin{align*}
    \lambda_{\text{min}, \perp}(\boldsymbol{\Sigma}):=\min \left\{\mu |\boldsymbol{z}^\top \boldsymbol{\Sigma}\boldsymbol{z}\geq \mu\Vert \boldsymbol{z}\Vert_2^2\text{ for all }\boldsymbol{z}\in \Theta \right\}.
\end{align*}
%for matrix $\boldsymbol{A}\in \mathbb{R}^{(n+d)\times(n+d)}$.
\end{assumption}

In this Assumption \ref{ass1}, we assume that $\displaystyle\boldsymbol{\Sigma}$ is well-behaved in all directions inside our parameter space $\Theta$, namely, both of its largest and smallest eigenvalues are of order $n$. This assumption follows directly after we rescale the $\|\bx_i\|_2$ such that $\| \bx_i\|_2 \leq \sqrt{(d+1)/n}$ for all $\bx_i,i\in [n]$.
When there is no covariate (i.e., $d=0$), then $\mathbf{\Sigma}=\sum_{i>j}(\mathbf{e}_i-\mathbf{e}_j)(\mathbf{e}_i-\mathbf{e}_j)^\top$ and Assumption \ref{ass1} holds naturally with $c_1 = c_2 = 1$ \citep{chen2019spectral}. 

We next introduce the condition number of this problem as 
\begin{align*}
	\kappa_1 :=\frac{\max_i w_i^*}{\min_i w_i^*} = \exp\left( \max_{i,j\in [n]}\left(\alpha_i^*+\bx_i^\top \bb^*-\alpha_j^*-\bx_j^\top \bb^* \right)\right),
\end{align*}
where $w_i^* = \exp(\alpha_i^*+\bx_i^\top \bb^*)$,
which extends the condition number in \cite{chen2019spectral} when there does not exist covariate. With all aforementioned assumptions in hand, we next present our first main theorem on the statistical rates of convergence of the MLE $\tb_M$.    Recall that we assume that the true parameter vector $\tb^*=\left(\alpha_1^*,\alpha_2^*,\dots,\alpha_n^*,\bb^{*\top}\right)^\top = \left(\ba^{*\top},\bb^{*\top}\right)^\top\in \Theta$, without loss of generality.

\begin{theorem}\label{noregularization}
	(Rate of convergence)
Suppose $np > c_p\log n$ for some $c_p>0$ and $d+1< n, (d+1)\log n\lesssim np$.  Consider $L \leq  c_4\cdot n^{c_5}$ for any absolute constants $c_4,c_5>0$. Let $\tb_M = (\wh{\ba}_M^\top,\wh{\bb}_M^\top)^\top$ be the solution of the MLE given in \eqref{MLEun}.  Then with probability at least $1-O(n^{-6})$, we have
\begin{align*}
\Vert \wh\ba_M-\ba^*\Vert_\infty&\lesssim \kappa_1^2\sqrt{\frac{(d+1)\log n}{npL}}; \qquad 
\left\Vert\wh \bb_M-\bb^*\right\Vert_2 \lesssim\kappa_1\sqrt{\frac{\log n}{pL}};\\
\left\Vert\tX\tb_M-\tX\tb^*\right\Vert_\infty&\lesssim \kappa_1^2 \sqrt{\frac{(d+1)\log n}{npL}}; \qquad 
 \frac{\left\Vert e^{\tX\tb_M}-e^{\tX\tb^*}\right\Vert_\infty}{\left\Vert e^{\tX\tb^*}\right\Vert_\infty} \lesssim \kappa_1^2 \sqrt{\frac{(d+1)\log n}{npL}},
\end{align*}
where $\tX = [\tx_1,\tx_2,\dots,\tx_n]^\top$.  
\end{theorem}

Recall that we rescaled the covariates such that $\max_{i\in [n]}\Vert\bx_i\Vert_2\leq\sqrt{(d+1)/n}$.  This scaling has an impact on the definition of $\bbeta^*$ and influences its rate.  %\textcolor{red}{To do: Explain the rates of $\beta_M$ e.g. why not involving $d$.} 
This explains why  $\hat{\bbeta}_M$ converges slower {and does not depend on $d$}.   However, {when we view $\bx_i^\top\bb$ as a whole, the estimation rate is $\sqrt{(d+1)\log n/npL}$, and} this does not impact the estimation of the individual score, as shown in Theorem~\ref{noregularization}.

\iffalse
\begin{remark}
Theorem \ref{noregularization} asserts that the statistical errors of the maximum likelihood estimators for $\ba$ 
%(i.e., $\Vert \ba_M-\ba^*\Vert_\infty$), %the linear combination of $\bb$ (i.e., $\Vert\bX^\top\tb_M-\bX^\top\tb^*\Vert_\infty$) 
and the ranking scores $\tX^\top\tb^*$ %(i.e., $\Vert\tX^\top\tb_M-\tX^\top\tb^*\Vert_\infty$) 
are of the order $\displaystyle\sqrt{{(d+1)\log n}/{npL}}$. When the number of covariates is fixed, this is the best rate one can hope for if one ignores the logarithmic term. %as  the effective sample size under our model is $n^2pL$. 
To understand this, we note that the expected number of comparisons involving item $i$ is $(n-1)pL$ for all $i\in [n]$. Besides, the parameter $\alpha_i$ only appears in the model of the comparisons when item $i$ is involved. 
Thus, if we ignore the logarithmic term and consider the case with a finite $d$, the best estimation error of $\alpha_i$ is $\mathcal{O}(\displaystyle\sqrt{{1}/{npL}})$ for $i\in [n]$, which matches the obtained bound for $\Vert \wh\ba_M-\ba^*\Vert_\infty$. Similar reasoning applies to estimating ranking scores $\tX\tb^*$.  When there are no covariates (i.e., $d=0$), our results in Theorem \ref{noregularization} reduce to the minimax optimal statistical error bound for MLE in $\ell_{\infty}$-norm under minimal sampling requirement ($p\asymp \log n/n$) for the  Erdős-Rényi random graph to connect \citep{chen2019spectral}. 
\end{remark}
\fi

\begin{remark}\label{bestratewecanhope}
  Following a similar proof, it holds that $$ \frac{\Vert \wh\ba_M-\ba^*\Vert_2}{\Vert \ba^*\Vert_2}\lesssim \kappa_1^2\sqrt{\frac{(d+1)\log n}{npL}},\qquad \,\,\frac{\left\Vert e^{\tX\tb_M}-e^{\tX\tb^*}\right\Vert_2}{\left\Vert e^{\tX\tb^*}\right\Vert_2}\lesssim \kappa_1^2 \sqrt{\frac{(d+1)\log n}{npL}},$$ which are the relative $\ell_2$-statistical rates of the intrinsic scores $\alpha_i^*,i\in[n]$ and overall  scores $\alpha_i^*+\bx_i^\top\bb^*,i\in [n]$, respectively.  Combining this relative statistical rate in $\ell_2$-norm with that in $\ell_{\infty}$-norm mentioned in Theorem \ref{noregularization}, we conclude that the estimation errors of latent scores and overall scores spread out across all items.
\end{remark}
\begin{remark}
{We note that we are further able to get rid of the factor $\sqrt{d+1}$ in the statistical rates of $\|\hat\balpha_M-\hat\balpha^*\|_{\infty}$ ( $\|\hat\balpha_M-\hat\balpha^*\|_{2}$ follows directly). This involves analyzing the non-asymptotic expansion of the $\hat\balpha_M$, and the details will be discussed in the following section.}
\end{remark}

\section{Uncertainty Quantification of MLE}\label{Inference}
Most existing works on ranking mainly study the first-order asymptotic behavior of their estimators \citep{hunter2004mm,chen2015spectral,jang2016top,shah2017simple,chen2019spectral}.   Deriving limiting distributional results in ranking models is important for uncertainty quantification, especially when covariate information is incorporated into the ranking problem, are not studied in detail.  

This section is devoted to understanding the sampling variability of the MLE $\tb_M$ under the CARE model. Directly studying the asymptotic behavior of $\tb_M$ is very challenging. To address this issue, we approximate $\tb_M$ by considering 
\begin{align}\label{eq11}
	\ob := \argmin_{\tb\in\Theta} \ocL(\tb),
\end{align}
here $\ocL(\tb)$ is the quadratic expansion of the loss function $\cL(\tb)$ around $\tb^*$ given by
\begin{align}\label{quadratic_expan}
    \ocL(\tb) = \cL(\tb^*)+\left(\tb-\tb^*\right)^\top\nabla \cL(\tb^*)+\frac{1}{2}\left(\tb-\tb^*\right)^\top\nabla^2 \cL(\tb^*)\left(\tb-\tb^*\right).
\end{align}
According to this definition, $\ob$ can also be given by the following linear equations
\begin{align*}
\left\{ \begin{array}{ll}
&\mathcal{P}\nabla\cL(\tb^*)+\mathcal{P}\nabla^2\cL(\tb^*)\left(\ob-\tb^*\right) = \boldsymbol{0};  \\
&\mathcal{P}\ob = \ob.
\end{array}
\right.
\end{align*}
{Here $\mathcal{P}$ represents the linear projection onto space $\Theta.$}
Observe that $\ob$ serves as a candidate approximator of $\tb_M$ whose uncertainty is easier to quantify according to Berry-Esseen theorem \citep{berry1941accuracy,esseen1942liapunov}. %We are going to quantify the uncertainty of $\tb_M$ by $\ob$ treating $\ob$ as an approximator. On one hand, the uncertainty of $\ob-\tb^*$ is easy to quantify according to the Berry-Esseen theorem \cite{berry1941accuracy,esseen1942liapunov}. 
It is not a statistical estimator but an auxiliary random variable that we used for the technical proof.
{It is worth mentioning that when there is no covariate, our linear expansion reduces to $\overline{\boldsymbol{\beta}}\in \mathbb{R}^{n}=[\nabla^2 \cL(\widetilde{\bbeta}^*)]^{+}\nabla \cL(\widetilde{\bbeta}^*),$ which is equal to the expansion in \cite{gao2021uncertainty} up to the off-diagonal terms in $\nabla^2 \cL(\widetilde{\bbeta}^*).$}

The critical difficulty falls in proving that the difference $\Delta\tb:= \tb_M-\ob$ is negligible compared to $\ob-\tb^*$ under certain conditions. 
To accommodate this, we derive novel proofs by leveraging the `leave-one-out'  \citep{chen2019spectral,chen2021spectral} technique to control the approximation error $\Delta\tb := \tb_M-\ob$ in $\ell_2$-norm and $\Delta\ba := \Delta\tb_{1:n}$ in $\ell_\infty$-norm.  The upper bounds are summarized in the following Theorem \ref{inferenceresidual}.  %The combination of these two parts gives the uncertainty quantification of $\tb_M-\ob$. To begin with, the following theorem gives bounds for $\Delta\tb = (\Delta\ba^\top,\Delta\bb^\top)^\top$.
 
\begin{theorem}\label{inferenceresidual}
(Approximation error)
Under the assumptions of Theorem \ref{noregularization}, if $\kappa_1^2\sqrt{(d+1)\log n/npL} \le c$ and $\kappa_1^2\left(\sqrt{{(d+1)}/{np}}+{\log n}/{np}\right) \leq c$ for some fixed constant $c>0$, we have 
	\begin{align*}
		\widehat{\balpha}_M = \overline{\balpha} + \Delta \balpha \text{ and } 	\widetilde{\bbeta}_M = \overline{\bbeta}+ \Delta \widetilde{\bbeta},
	\end{align*} 
and with probability at least $1-O(n^{-5}),$ it holds that 
\begin{align*}
    \left\Vert \Delta \tb \right\Vert_2&\lesssim\kappa_1^4\frac{(d+1)^{0.5}\log n}{\sqrt{n}pL},\\
        \left\Vert\Delta\ba\right\Vert_\infty &\lesssim\kappa_1^6\frac{(d+1)\log n}{np L}+\frac{\kappa_1^4}{np}\sqrt{\frac{(d+1)\log n}{L}}\left(\sqrt{d+1}+\frac{\log n}{\sqrt{np}}\right).%{\kappa_1^4}\sqrt{\frac{(d+1)\log n}{npL}}\left(\sqrt{\frac{d+1}{np}}+\frac{\log n}{{np}}\right).
\end{align*}
\iffalse
In addition, if $$\displaystyle\kappa_1^2\left(\sqrt{\frac{d+1}{np}}+\frac{\log n}{np}\right) = o(1),$$ then with probability at least $1-O(n^{-5}),$ we have
\begin{align*}
    \left\Vert\Delta\ba\right\Vert_\infty \lesssim\kappa_1^6\frac{(d+1)\log n}{np L}+\frac{\kappa_1^4}{np}\sqrt{\frac{(d+1)\log n}{L}}\left(\sqrt{d+1}+\frac{\log n}{\sqrt{np}}\right).
\end{align*}
\fi
\end{theorem}
\iffalse
\fi
\begin{remark}
	The assumptions $\kappa_1^2\sqrt{(d+1)\log n/npL} = \mathcal{O}(1)$ and $\displaystyle \kappa_1^2\left(\sqrt{(d+1)/{np}}+{\log n}/{np}\right)\leq c$ for some fixed constant $c>0$  are mild.  When $d$ and $\kappa_1$ are bounded,  they hold when $p\gtrsim \log n/n$.   This matches the lower bound of the sampling probability $p$ to ensure the connectivity of the Erdős-Rényi random graph, which is a necessary requirement for item ranking. Besides, $\kappa_1^2\sqrt{(d+1)\log n/npL} = \mathcal{O}(1)$ is also required by the consistency results of our estimator according to Theorem \ref{noregularization}.
\end{remark}

 \begin{remark}
{Given the non-asymptotic expansion, and approximation error  $\Delta\balpha:=\hat\balpha_M-\bar{\bbeta}_{1:n}$ presented in Theorem \ref{inferenceresidual}, we are able to achieve a tighter $\ell_{\infty}$ statistical error bound for $\|\hat\balpha_M-\balpha^*\|_{\infty}.$ Specifically, under the assumptions of Theorems 4 and 7, as long as
 \begin{align*}    \kappa_1^5(d+1)\sqrt{\frac{\log n}{npL}}+\kappa_1^3\sqrt{\frac{(d+1)}{np}}\left(\sqrt{d+1}+\frac{\log n}{\sqrt{np}}\right)\lesssim1.
 \end{align*}
Then with probability at least $1-O(n^{-5})$ we have 
	\begin{align*}
		\left\Vert \widehat{\balpha}_M - \balpha^* \right\Vert_\infty&\lesssim \kappa_1\sqrt{\frac{\log n}{npL}}.
	\end{align*}}
\end{remark}
Next, we utilize Berry-Essen theorem \citep{berry1941accuracy,esseen1942liapunov} to derive the asymptotic distribution of the linear combinations of $\tb_M$, respectively.   Since it holds for any linear combinations, the result applies to any finite-dimensional distribution of $\tb_M$.
\iffalse
\begin{corollary}
(Asymptotic normality of the main term)
Under the assumptions of Theorem \ref{inferenceresidual}, if $\max_{i\in [n]}\Vert \bx_i\Vert_2\leq 0.5$, for $k\in[n]$,  in the following decomposition 
\begin{align*}
    \sqrt{L}\left(\alpha_{M,k}-\alpha^*_k\right) = \sqrt{L}\left(\alpha_{M,k}-\oa_k\right) + \sqrt{L}\left(\oa_k-\alpha_k^*\right),
\end{align*}
we have an approximation error
\begin{align*}
    \left|\frac{\sqrt{L}\left(\alpha_{M,k}-\oa_k\right)}{\sqrt{\left(\nabla^2\cL(\tb^*)^{+}\right)_{k,k}}}\right|\lesssim \kappa_1^6\frac{(d+1)\log n}{\sqrt{np L}}+\kappa_1^4\sqrt{\frac{(d+1)\log n}{np}}\left(\sqrt{d+1}+\frac{\log n}{\sqrt{np}}\right)
\end{align*}
with probability exceeding $1-O(n^{-5})$ (randomness comes from $\mathcal{G}$ and $y_{i,j}^{(l)}$) and asymptotic normality
\begin{align*}
\sup_{x\in\mathbb{R}}\left|\mathbb{P}\left(\frac{\sqrt{L}\left(\oa_k-\alpha^*_k\right)}{\sqrt{\left(\nabla^2\cL(\tb^*)^{+}\right)_{k,k}}}\leq x\bigg| \mathcal{G}\right)-\mathbb{P}\left(\mathcal{N}(0,1)\leq x\right)\right| \lesssim\frac{\kappa_1}{\sqrt{npL}}
\end{align*}
with probability exceeding $1-O(n^{-10})$ (randomness comes from $\mathcal{G}$).
\end{corollary}
\fi

\begin{theorem}\label{betainference}
	(Asymptotic normality of MLE)
Given $\boldsymbol{c}\in \mathbb{R}^{n+d}$, let $\bc = \mathcal{P}\boldsymbol{c}$ be the projection of $\boldsymbol{c}$ onto linear space $\Theta$. Under the assumptions of Theorem \ref{inferenceresidual}, we have the following decomposition 
\begin{align*}
    \sqrt{L}\left(\boldsymbol{c}^\top\tb_M-\boldsymbol{c}^\top\tb^*\right) = \sqrt{L}\left(\boldsymbol{c}^\top\tb_M-\boldsymbol{c}^\top\ob\right) + \sqrt{L}\left(\boldsymbol{c}^\top\ob-\boldsymbol{c}^\top\tb^*\right),
\end{align*}
where 
\begin{align*}
    \left|\frac{\sqrt{L}\left(\boldsymbol{c}^\top\tb_M-\boldsymbol{c}^\top\ob\right)}{\sqrt{\bc^\top\left[\mathcal{P}\nabla^2\cL(\tb^*)\mathcal{P}\right]^{+}\bc}}\right|\lesssim& \left[\kappa_1^6\frac{(d+1)\log n}{\sqrt{np L}}+\kappa_1^4\sqrt{\frac{(d+1)\log n}{np}}\left(\sqrt{d+1}+\frac{\log n}{\sqrt{np}}\right)\right]\frac{\Vert\boldsymbol{c}_{1:n}\Vert_1}{\Vert \bc\Vert_2} \\
    &+\kappa_1^4\frac{(d+1)^{0.5}\log n}{\sqrt{pL}}\frac{\Vert \boldsymbol{c}_{n+1:n+d}\Vert_2}{\Vert \bc\Vert_2},
\end{align*}
 with probability exceeding $1-O(n^{-5})$ (randomness comes from $\mathcal{G}$ and $y_{i,j}^{(l)}$) and
\begin{align*}
\sup_{x\in\mathbb{R}}\left|\mathbb{P}\left(\frac{\sqrt{L}\left(\boldsymbol{c}^\top\ob-\boldsymbol{c}^\top\tb^*\right)}{\sqrt{\bc^\top\left[\mathcal{P}\nabla^2\cL(\tb^*)\mathcal{P}\right]^{+}\bc}}\leq x\bigg| \mathcal{G}\right)-\mathbb{P}\left(\mathcal{N}(0,1)\leq x\right)\right| \lesssim\frac{\kappa_1}{\sqrt{npL}},
\end{align*}
with probability exceeding $1-O(n^{-10})$ (randomness comes from $\mathcal{G}$). Combining the approximation error and asymptotic distribution together, and by taking all randomness into consideration, we further obtain
\begin{align*}
&\sup_{x\in\mathbb{R}}\left|\mathbb{P}\left(\frac{\sqrt{L}\left(\boldsymbol{c}^\top\tb_M-\boldsymbol{c}^\top\tb^*\right)}{\sqrt{\bc^\top\left[\mathcal{P}\nabla^2\cL(\tb^*)\mathcal{P}\right]^{+}\bc}}\leq x\right)-\mathbb{P}\left(\mathcal{N}(0,1)\leq x\right)\right|\\ \lesssim&\frac{\kappa_1}{\sqrt{npL}}+\left[\kappa_1^6\frac{(d+1)\log n}{\sqrt{np L}}+\kappa_1^4\sqrt{\frac{(d+1)\log n}{np}}\left(\sqrt{d+1}+\frac{\log n}{\sqrt{np}}\right)\right]\frac{\Vert\boldsymbol{c}_{1:n}\Vert_1}{\Vert \bc\Vert_2} \\
&+\kappa_1^4\frac{(d+1)^{0.5}\log n}{\sqrt{pL}}\frac{\Vert \boldsymbol{c}_{n+1:n+d}\Vert_2}{\Vert \bc\Vert_2}+\frac{1}{n^5}.
\end{align*}
\end{theorem}
In Theorem \ref{betainference}, we first obtain the distributional guarantee of $\frac{\sqrt{L}\left(\boldsymbol{c}^\top\ob-\boldsymbol{c}^\top\tb^*\right)}{\sqrt{\bc^\top\left[\mathcal{P}\nabla^2\cL(\tb^*)\mathcal{P}\right]^{+}\bc}}$,
 conditioned on the comparison graph $\mathcal{G}$, in the sense that we only take the randomness of $y_{i,j}^{(l)}$ into consideration. These results are stronger than the distributional guarantees, which make use of all the randomness from $\mathcal{G}$ and $y_{i,j}^{(l)}$,  by the dominated convergence theorem.  Besides, combining with the approximation results, we further derive distributional guarantees for linear combinations of $\tilde{\bb}_M$ by taking all the randomness of into consideration. The asymptotic variance of $\mathbf{c}^\top \tb_M$ is $\bc^\top[\mathcal{P}\nabla^2\cL(\tb^*)\mathcal{P}]^{+}\bc$ as presented in Theorem \ref{betainference}, where $[\mathcal{P}\nabla^2\cL(\tb^*)\mathcal{P}]^{+}$ is exactly the inverse of the fisher information (conditioned on graph $\cG$) that is projected into the space $\Theta$.
 
 %{\red To do: explain the meaning of the variance term.}

Moreover, recall that we have scaled the covariates to satisfy $\max_{i\in [n]}\Vert\bx_i\Vert_2\leq\sqrt{(d+1)/n}$ in the data preprocessing step. Then, if $\displaystyle\kappa_1, {\Vert \boldsymbol{c}_{1:n}\Vert_1}/{\Vert \bc\Vert_2}$ and $\displaystyle{\sqrt{n}\Vert\boldsymbol{c}_{n+1:n+d}\Vert_2}/({\sqrt{d+1}\Vert \bc\Vert_2})$ are bounded, the asymptotic normality holds when
\begin{align}\label{sample_c}
\max\left\{\frac{(d+1)\log n}{\sqrt{npL}}, \frac{(d+1)\sqrt{\log n}}{\sqrt{np}}, \frac{\log n\sqrt{(d+1)\log n}}{np}\right\} = o(1).
\end{align} 
in the sense that it allows the comparison graph to be sparse (when $\displaystyle p\asymp {1}/{n}$ up to logarithmic terms) when the covariate dimension is bounded. This admits all existing uncertainty quantification results for the BTL model without covariates as special cases. %This admits all existing uncertainty quantification results for the BTL model without covariates as special cases.

Finally, we comment on the condition $\displaystyle\max\{\left\Vert \boldsymbol{c}_{1:n}\right\Vert_1 , \sqrt{{n}/{(d+1)}}\left\Vert \boldsymbol{c}_{n+1:n+d}\right\Vert_2\}\lesssim \Vert\bc\Vert_2$.  This is only a mild requirement.  For instance, when $d=0$ and $\mathbf{c}\in \RR^n$ is sparse (the original BTL model), this inequality holds naturally. In addition, the comparison of preference ratings is another significant illustration that meets the requirement. In specific, for testing $H_0:\tx_i^\top\tb^*\leq \tx_j^\top\tb^*\text{ v.s. } H_a:\tx_i^\top\tb^* >\tx_j^\top\tb^*$, we choose $\boldsymbol{c} = \tx_i-\tx_j\in \RR^{n+d}$. In this scenario, the the condition is met since $\max\{\left\Vert \boldsymbol{c}_{1:n}\right\Vert_1 , \sqrt{{n}/{d+1}}\Vert \boldsymbol{c}_{n+1:n+d}\Vert_2\} = 2,$  and $\Vert \bc\Vert_2 = \Vert\boldsymbol{c}\Vert_2 \asymp \sqrt{2 + 2(d+1)/n}$.  %Hence, the forementioned condition is also satisfied.
\iffalse
\begin{remark}
	Recall that we have scaled the covariates to satisfy $\max_{i\in [n]}\Vert\bx_i\Vert_2\leq\sqrt{(d+1)/n}$ in the data preprocessing step. Although Theorem \ref{betainference} provides results for general $\boldsymbol{c}$,  in order for the approximation error to be negligible under the minimum sample complexity, this theorem mainly applies to the case that $\displaystyle\max\left\{\left\Vert \boldsymbol{c}_{1:n}\right\Vert_1 , \sqrt{{n}/{(d+1)}}\left\Vert \boldsymbol{c}_{n+1:n+d}\right\Vert_2\right\}\lesssim \Vert\bc\Vert_2$ so that the second factors in the approximation errors are bounded.
	An important example that satisfies such conditions is the comparison of preference scores. For instance, if we wish to test $H_0:\tx_i^\top\tb^*\leq \tx_j^\top\tb^*\text{ v.s. } H_a:\tx_i^\top\tb^* >\tx_j^\top\tb^*$, we choose $\boldsymbol{c} = \tx_i-\tx_j$ in Theorem \ref{betainference}. In this scenario, we obtain
	\begin{align*}
	\max\left\{\left\Vert \boldsymbol{c}_{1:n}\right\Vert_1 , \sqrt{\frac{n}{d+1}}\left\Vert \boldsymbol{c}_{n+1:n+d}\right\Vert_2\right\} = 2,
	\end{align*}
	while $\Vert \bc\Vert_2 = \Vert\boldsymbol{c}\Vert_2 \leq \sqrt{2 + 2(d+1)/n}$.  Hence, the aforementioned condition is satisfied.
\end{remark}
\fi

%{\red To do: maybe briefly summarize the proof challenges of dealing with covariates compared with previous BTL literature. (in high-level)}

An important corollary of Theorem \ref{betainference} is the limiting distribution of $\hat \alpha_{M,S_k}$, where $S_k$ is any subset over $[n]$ with size $k<\infty.$ The corresponding theoretical property is summarized in the following Corollary \ref{alphainference}.
{ 
\begin{corollary}\label{alphainference}
Assume the assumptions of Theorem \ref{inferenceresidual} hold. Then for any fixed $k\in [n]$, as long as 
\begin{align*}
	\kappa_1^6\frac{\sqrt{k}(d+1)\log n}{\sqrt{np L}}+\kappa_1^4\sqrt{\frac{k(d+1)\log n}{np}}\left(\sqrt{d+1}+\frac{\log n}{\sqrt{np}}\right) = o(1),
\end{align*}
 we have
\begin{align*}
	\sqrt{L}\left(\left[\mathcal{P}\nabla^2\mathcal{L}(\widetilde{\bbeta}^*)\mathcal{P}\right]^{+}_{S_k, S_k}\right)^{-1/2}(\widehat{\balpha}_{M, S_k} - \balpha_{S_k}^*)\leadsto \mathcal{N}\left(\boldsymbol{0}, \mathbf{I}_k\right),
\end{align*}
where $S_k$ is any subset over $[n]$ with size $k.$
\end{corollary}
}

Although our results are established under a more general setting than the existing literature on the BTL model, which does not involve the covariates, our specific result in Corollary \ref{alphainference} with $d=0$ can still compare favorably.   Compared with \cite{liu2022lagrangian}, we need  a much smaller sample complexity to establish the asymptotic normality. Specifically, they require $\displaystyle {n\log n}/{\sqrt{L}}+{\log n}/{\sqrt{pL}}=o(1)$ to derive the asymptotic normality. This condition  requires $L \gg n^2$. In contrast,  we allow $L=\cO(1)$ and our requirement for the sample complexity is minimax optimal up to logarithm  terms \citep{negahban2012iterative,chen2019spectral}. Moreover, \cite{liu2022lagrangian} use the Lagrangian debiasing method to derive the estimators, which involves an additional tuning parameter. 
%Compared with them, our method does not involve any tuning procedure in a sense that for a given proper initializer, we achieve our MLE in \eqref{MLEun} by simply running the gradient descent as we discussed in \S\ref{mle}. 
Compared with \cite{han2020asymptotic}, we allow sparse compare graphs ($p\gtrsim 1/n$ by ignoring logarithm terms), whereas they require a much denser comparison graph ($\displaystyle p\gtrsim {1}/{n^{1/10}}$)  than ours. %and therefore requires a much denser comparison graph than ours.

We now compare our results with \cite{gao2021uncertainty} and \cite{fan2022ranking} (under the pairwise comparison regime). As the analysis in these two works does not incorporate the condition number $\kappa_1$ or covariate, we will consider the regime $\kappa_1 = O(1)$ and $d = 0$ in our theorems for comparison. First of all, both papers show that the asymptotic normality holds even for the sparse regime $\displaystyle p\asymp {1}/{n}$, up to a logarithmic order. However, the choice of approximators and  approximation errors are very different.  Instead of using the Taylor expansion $\ocL(\cdot)$ given by Eq.~\eqref{quadratic_expan}, \cite{gao2021uncertainty,fan2022ranking} consider the following quadratic approximation:
\begin{align*}
   \ocL_{\text{diag}}(\tb) = {\cL(\tb^*)}+\left(\tb-\tb^*\right)^\top\nabla \cL(\tb^*)+\frac{1}{2}\left(\tb-\tb^*\right)^\top\left(\textbf{diag}\nabla^2 \cL(\tb^*)\right)\left(\tb-\tb^*\right),
\end{align*}
which only keeps the diagonal part of $\nabla^2 \cL(\tb^*)$, and define the approximator $\ob_{\text{diag}} = \argmin\ocL_{\text{diag}}(\tb)$. In addition, it is worth noting that their approach, which only keeps the diagonal of the Hessian matrix to handle the approximation error, cannot be applied directly to our setting (with covariates) due to several reasons. {  Firstly, results in \cite{gao2021uncertainty} without involving covariates imply that for any pair of indices \(i, j\) within the sample space \([n]\), the estimators \(\hat\alpha_{M,i}\) and \(\hat\alpha_{M,j}\) (corresponds to $\hat\theta_i,\hat\theta_j$ in their paper) are asymptotically independent, which justifies their utilization of a quadratic approximation for a single variable while keeping others fixed. In our model that incorporates covariates, the corresponding estimators \(\widehat{\alpha}_{M, i}\) and \(\widehat{\alpha}_{M, j}\) do not exhibit asymptotic independence. Secondly, the approach in \cite{gao2021uncertainty} overlooks the off-diagonal elements of the Hessian matrix \(\nabla^2 \mathcal{L}(\tilde{\mathbf{\beta}}^*)\), which is an acceptable approximation when no covariates are involved since these off-diagonal entries are negligible compared with asymptotic variances. However, in our study, after we take the covariates into consideration, this approximation no longer holds. As a result, an expansion similar to [(2.10), \citep{gao2021uncertainty}] would not give a valid expansion in our case.} Therefore, this motivates us to drive our own method presented above. % which just keeps the diagonal cannot be applied to our model which contains covariates. 

In terms of the approximation errors, they show that with high probability
\begin{align}
    \left|\frac{\sqrt{L}\left(\wh\alpha_{M,k}-\ob_{\text{diag},k}\right)}{\sqrt{\left(\left[\mathcal{P}\nabla^2\cL(\tb^*)\mathcal{P}\right]^{+}\right)_{k,k}}}\right|\lesssim \sqrt{\frac{\log n}{np}}+\frac{(\log n)^3}{(np)^2}+\left(\frac{\log n}{npL}\right)^{1/4}+\frac{(\log n)^{7/4}}{(np)^{5/4}L^{1/4}},\label{gaorate}
\end{align}
while we prove that for our approximation $\ob$, with high probability for $k\in [n]$ it holds that
\begin{align}
    \left|\frac{\sqrt{L}\left(\wh\alpha_{M,k}-\oa_k\right)}{\sqrt{\left(\left[\mathcal{P}\nabla^2\cL(\tb^*)\mathcal{P}\right]^{+}\right)_{k,k}}}\right|\lesssim \sqrt{\frac{\log n}{np}}+ \frac{(\log n)^{1.5}}{np}.\label{ourrate}
\end{align}
When $L = o(np/\log n)$, the term $\displaystyle\left({\log n}/{npL}\right)^{1/4}$ dominates the right-hand side of \eqref{gaorate} and hence  also dominates the error rate given by \eqref{ourrate} as long as $np \gtrsim  (\log n)^2$.  In other words, our approximation error is an order of magnitude smaller than that in \cite{gao2021uncertainty} and \cite{fan2022ranking} (when considering the pairwise comparison regime).  This holds true for the common case where $L\asymp 1$.  We next explain the underlying rationale.
{ In scenarios where the likelihood function does not contain covariates, the off-diagonal elements of the Hessian matrix, though in smaller order compared to the asymptotic variance, still contribute to bias in the remainder term beyond the non-asymptotic expansion. This inherent bias, which emerges when examining marginal likelihoods, is more pronounced than using joint likelihood. Our approach, which employs joint likelihood, effectively accounts for these off-diagonal elements, thereby mitigating their impact and resulting in a more accurate approximation.}

%{\red Add a corollary on inference of $\beta_j$ here.}
Besides investigating the asymptotic behavior of $\hat\alpha_i,i\in[n],$ studying the asymptotic property for $\hat\beta_j,j\in[d]$ is another crucial task as it depicts whether some covariates have any power for explaining latent scores. We deduce these from the Theorem \ref{inferenceresidual} and Theorem \ref{betainference}, and summarize them together with a refined $\ell_2$-upper bound of $\| \wh\bb_M-\bb^*\|_2$ in Corollary \ref{cor:inferenceb}.
\begin{corollary}\label{cor:inferenceb}
		Under the assumptions of Theorem \ref{inferenceresidual} and Theorem \ref{betainference}, we first obtain a refined upper bound for $\|\wh\bb_M-\bb^*\|_2$, namely
		\begin{align}\label{refinedbetabound}
			{	\left\Vert\wh \bb_M-\bb^*\right\Vert_2 \lesssim 
				\sqrt{\frac{(d+1)}{n}}\cdot \max\Bigg\{ \kappa_1^4\frac{\log n}{pL}, \kappa_1\sqrt{\frac{\log n}{{pL}}}\Bigg\}.}
		\end{align}
	In addition, we also achieve the distributional results for every $\hat\beta_j,j\in [d],$
	\begin{align*}
	&\sup_{x\in\mathbb{R}}\left|\mathbb{P}\left(\frac{\sqrt{L}\left(\boldsymbol{e}_{n+j}^\top\tb_M-\boldsymbol{e}_{n+j}^\top\tb^*\right)}{\sqrt{\bar{\be}_{n+j}^\top\left[\mathcal{P}\nabla^2\cL(\tb_M)\mathcal{P}\right]^{+}\bar{\be}_{n+j}}}\leq x\right)-\mathbb{P}\left(\mathcal{N}(0,1)\leq x\right)\right|\\ \lesssim&\frac{\kappa_1}{\sqrt{npL}}+\kappa_1^4\frac{(d+1)^{0.5}\log n}{\sqrt{pL}}+\frac{1}{n^5}.
	\end{align*}
	Therefore, when $(d+1)^{0.5}\log n/ (pL)\rightarrow0$, for any $j\in [d]$, we have
	\begin{align*}
\mathbb{P}\left( \beta_j^*\in [C_L(\tb_M),C_U(\tb_M)]\right)=1-\alpha,
	\end{align*}
	where $[C_L(\tb_M),C_U(\tb_M)]=[\boldsymbol{e}_{n+j}^\top\tb_M- \frac{\sqrt{\bar{\be}_{n+j}^\top\left[\mathcal{P}\nabla^2\cL(\tb_M)\mathcal{P}\right]^{+}\bar{\be}_{n+j}}z_{{\alpha/2}}}{\sqrt{L}}, \boldsymbol{e}_{n+j}^\top\tb_M+ \frac{\sqrt{\bar{\be}_{n+j}^\top\left[\mathcal{P}\nabla^2\cL(\tb_M)\mathcal{P}\right]^{+}\bar{\be}_{n+j}}z_{{\alpha/2}}}{\sqrt{L}}]$
	with  $z_{\alpha/2}$ being the upper $\alpha/2$-quantile of the standard Gaussian distribution.
	\end{corollary}

%Under mild conditions, the asymptotic distribution of for each $\hat\beta_j,j\in[d]$ and the $(1-\alpha)$- confidence interval for $\beta_j^*$ are derived in Corollary \ref{cor:inferenceb}. This will enable us to determine each covariate's significance.

\iffalse
{\red where to put?}
\begin{proposition}(Refined $\ell_2$-bound for $\| \wh\bb_M-\bb^*\|_2$) Under conditions of Theorem \ref{inferenceresidual} and Theorem \ref{betainference}. We obtain a refined upper bound for $\|\wh\bb_M-\bb^*\|_2$, namely
	\begin{align*}
		{\left\Vert\wh \bb_M-\bb^*\right\Vert_2 \lesssim 
			\sqrt{\frac{(d+1)}{n}}\cdot \max\bigg\{ \kappa_1^4\frac{\log n}{pL}, \kappa_1\frac{\log n}{\sqrt{pL}}\bigg\}.}
	\end{align*}
\end{proposition}
\fi
{ In Corollary \ref{cor:inferenceb}, we first obtain a refined upper bound of $\|\wh\bb_M-\bb^*\|_2$, and we next explain the rationality behind this.  In Theorem  \ref{noregularization}, a rough upper bound for $\|\wh\bb_M-\bb^*\|_2$  is obtained via concentration since no precise distributional results are involved in that stage.
		Given the statistical rates derived in Theorem \ref{noregularization}, we then analyze the non-asymptotic approximation and distribution of $\wh\bb_M,$ in Theorems  \ref{inferenceresidual} and \ref{betainference}, respectively. Finally, based on these distributional results, a refined upper bound for $\|\wh\bb_M-\bb^*\|_2$ is achieved. 
		
		If  we let $\kappa_1,d=\cO(1)$, one observes that the final upper bound  of $\|\wh\bb_M-\bb^*\|_2$ involves both rates $ \frac{1}{\sqrt{n}pL}$ and $\frac{1}{\sqrt{npL}},$ by ignoring logarithm terms. We conjecture the term $\frac{1}{\sqrt{n}pL},$ which comes from the approximation error ($\|\wh\bb_M-\ob_{n+1:n+d}\|_2$) in Theorem \ref{inferenceresidual}, can be improved to $\frac{1}{npL}$. In this case,  $\|\wh\bb_M-\bb^*\|_2$ should be bounded by the rate of  $\cO(\sqrt{\frac{1}{npL}})$ (the same order as the variance of $\wh\bb_M$). The numerical studies in \S\ref{synthetic1} validates this conjecture by showing that the rates of $\|\wh\bb_M-\bb^*\|_2$ is proporation to $\frac{1}{\sqrt{pL}}$ after we fix $n$ and $d$. However, improving the approximation error is highly non-trivial and needs more complex theoretical analysis.  Therefore, we will leave this as our future work. {It is worth mentioning that the assumption $(d+1)^{0.5}\log n/ (pL)\rightarrow0$ is only required while doing inference for $\beta_j^*,j\in [d]$. In all the previously mentioned theorems and corollaries, we do not need this condition and allow $L=1$. }

	Besides the refined $\ell_2$-bounds, the asymptotic distribution for each $\hat\beta_j,j\in[d]$ and the $(1-\alpha)$- confidence interval for $\beta_j^*$ are also derived in Corollary \ref{cor:inferenceb}. This will enable us to determine each covariate's significance in real data studies.

		 \iffalse  The condition $pL\rightarrow \infty$ can also be relaxed to $npL\rightarrow \infty$ when the aforementioned approximation error $\|\wh\bb_M-\ob_{n+1:n+d}\|_2$ is improved to $\frac{1}{npL}.$\fi }%It is worth mentioning that improving the condition $pL\rightarrow \infty$ to the condition $npL\rightarrow \infty$ is equivalent with improving the approximation error from $\frac{1}{\sqrt{n}pL}$ to $\frac{1}{npL}$. }

% Acknowledgements and Disclosure of Funding should go at the end, before appendices and references
\section{Numerical Results}\label{Experiments}
In this section, we conduct numerical experiments using synthetic and real data to validate our theories. In \S \ref{synthetic1} and \S \ref{synthetic2}, we leverage synthetic data to corroborate the statistical rates given in \S \ref{Consistency} and distributional results given in \S \ref{Inference}, respectively. In addition, in \S \ref{real_data_pokeman}, we illustrate further our model and methods by using the mutual funds holding data.

\subsection{Rate of Convergence}\label{synthetic1}
%In this section, we focus on validating our statistical consistency results derived in \S\ref{Consistency} through synthetic data. 
We begin with the data generation process.  Throughout the synthetic data experiments, we set $n$ to be $200$ and $d$ to be $5$. The covariates are generated independently with $(\bx_i)_j\sim \text{Uniform}[-0.5, 0.5]$ for all $i\in [n], j\in [d]$. For matrix $\boldsymbol{X} = [\bx_1,\bx_2,\dots,\bx_n]^\top\in \RR^{n\times d}$,  its columns are then normalized such that they have mean $0$ and standard deviation $1$. Next, we scale $\bx_i$ by $\bx_i / K$ so that $ \max_{i\in [n]}\Vert\bx_i\Vert_2/K = \sqrt{(d+1)/{n}}$.  We generate $\breve{\ba}\in \mathbb{R}^n$ by sampling its entries independently from $\text{Uniform}[0.5, \log (5)-0.5]$. Also, a $\breve{\bb}\in\mathbb{R}^d$ is generated uniformly from the hypersphere $\displaystyle\{\bb:\Vert \bb\Vert_2 = 0.5\sqrt{{n}/{(d+1)}}\}$. Then we project $(\breve{\ba}^\top,\breve{\bb}^\top)^\top$ onto linear space $\Theta$ and let it be $\tb^*$. In this way, we ensure $\kappa_1\leq 5$.  %Next, we follow the BTL model to generate the comparison results.

To validate the statistical rates in Theorem \ref{noregularization}, we use the above method to generate the covariates $\bx_i$ and $\tb^*$ three times.   This gives us three different instances of the covariates  $\bx_i$ and the parameter $\tb^*$.  For each given instance, we consider $6$ different $(p, L)$ pairs, which are listed below.
\begin{table}[h]
\centering
\begin{tabular}{c| c c c c c c} 
 \hline
 $p$ & 1 & 0.5 & 0.222 & 0.625 & 0.4 & 0.278 \\ [0.5ex] 
 \hline
  $L$ & 50 & 25 & 25 & 5 & 5 & 5 \\
 \hline
\end{tabular}
\end{table}

For each $(p,L)$ pair, comparion graph $\mathcal{E}$, $\{y_{i,j}^{(\ell)},l\in[L],(i,j)\in\mathcal{E}\}$ is generated and the MLE $\tb_M$ is calculated based on the available data. This process is repeated $200$ times, and the averaged $\Vert \wh\ba_M-\ba^*\Vert_\infty$ and $\Vert \wh\bb_M-\bb^*\Vert_2/\Vert \bb^*\Vert_2$, as well as their associated standard deviations, are recorded. The results are depicted in Figure \ref{estimationerror} for each of the three instances. Note that $\Vert \wh\ba_M-\ba^*\Vert_\infty$ and $\Vert \wh\bb_M-\bb^*\Vert_2/\Vert \bb^*\Vert_2$ are nearly proportional to $\displaystyle{1}/{\sqrt{pL}}$, lending further support of the results in Theorem \ref{noregularization}. The results are insensitive to three different instances, as expected.
%We plot the mean as well as the standard deviation of $\Vert \ba_M-\ba^*\Vert_\infty$ and $\Vert \bb_M-\bb^*\Vert_2/\Vert \bb^*\Vert_2$ over $20$ Monte Carlo simulations in Figure \ref{estimationerror}. It is worth note that, $\Vert \ba_M-\ba^*\Vert_\infty$ and $\Vert \bb_M-\bb^*\Vert_2/\Vert \bb^*\Vert_2$ are almost proportional to $\displaystyle\frac{1}{\sqrt{pL}}$, confirming our results in Theorem \ref{noregularization}.
\begin{figure}[ht]
\begin{center}
\includegraphics[width=15cm]{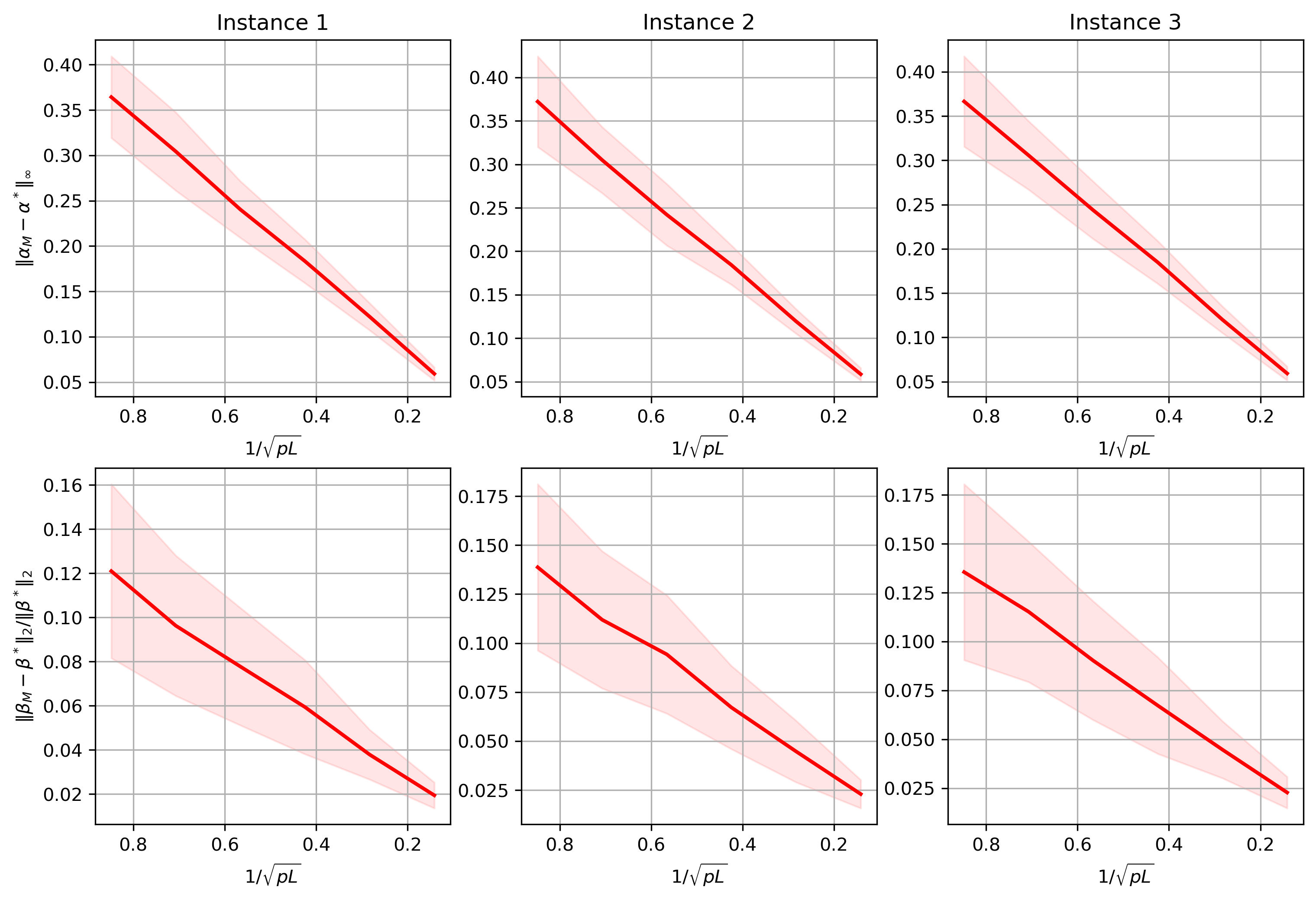}
\end{center}
\caption{Statistical rates of $\Vert \wh\ba_M-\ba^*\Vert_\infty$ and $\Vert \wh\bb_M-\bb^*\Vert_2/\Vert \bb^*\Vert_2$ for three simulated instances (realization of simulated models). The solid red lines and light areas represent the averaged $\Vert \wh\ba_M-\ba^*\Vert_\infty$, $\Vert \wh\bb_M-\bb^*\Vert_2/\Vert \bb^*\Vert_2$ and their associated standard errors based on 200 Monte Carlo simulations.}
\label{estimationerror}
\end{figure}

\subsection{Distributional Results}\label{synthetic2}
%In this subsection, we investigate the asymptotic normality of the MLE $\tb_M$ with synthetic data. 
We employ the same method given in \S \ref{synthetic1} to generate the covariates $\bx_i$ and $\tb^*$ once and fix them throughout the simulation.  Letting the effective sample size $\displaystyle n_a = {n}/[(d+1)\log n]$ (here $n=200$ and $d=5$), we choose the $6$ pairs $(p,L)$  with $p = 1.25 / n_a$ or $p= 2/n_a$ and $L \in \left\{2,6,20\right\}$.  For each $(p, L)$ pair,  the graph $\mathcal{E}$ and data $\{y_{i,j}^{(\ell)},l\in[L],(i,j)\in\mathcal{E}\}$ are generated $250$ times and the MLEs $\tb_M$ for all simulations are recorded.  Figure \ref{alpha1QQ} presents the Q-Q plots for checking the normality of $(\wh\ba_M)_1$, the first component of  $\wh\ba_M$.    The results show that  $(\wh\ba_M)_1$ follows closely the normal distribution.

\begin{figure}[ht]
\begin{center}
\includegraphics[width=15cm]{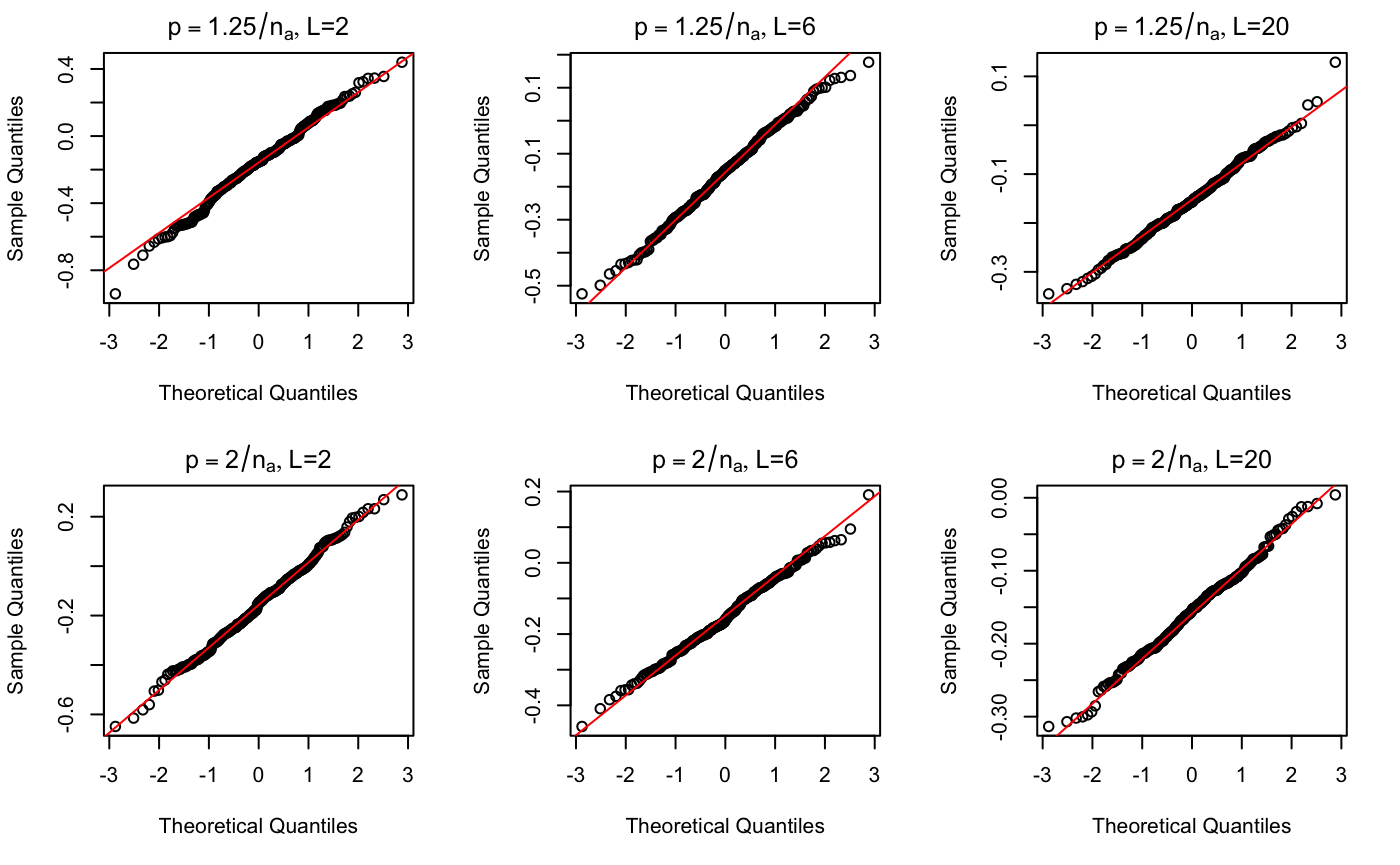}
\end{center}
\caption{Q-Q Plots for checking the normality of $(\wh\ba_M)_1$ based on 250 simulations.} 
\label{alpha1QQ}
\end{figure}

In addition to  checking the asymptotic normality, we now verify the asymptotic variance of our estimator. As an illustration, we consider the linear combination $\boldsymbol{c}^\top\tb_M$, where $\boldsymbol{c} = \boldsymbol{e}_1+\boldsymbol{e}_{201}$ and $\boldsymbol{e}_i$ is the $i$-th vector from the standard basis of $\mathbb{R}^{205}$.   Based on 250 simulations with  $(p,L) = (1.25/n_a, 2)$ and $(p,L) = (2/n_a, 20)$, the histograms of the following two standardized random variables are plotted:
\begin{align}\label{AB}
    A = \frac{\sqrt{L}\left(\boldsymbol{c}^\top\tb_M-\boldsymbol{c}^\top\tb^*\right)}{\sqrt{\bc^\top\left[\mathcal{P}\nabla^2\cL(\tb^*)\mathcal{P}\right]^{+}\bc}} \text{ and } B = \frac{\sqrt{L}\left(\boldsymbol{c}^\top\tb_M-\boldsymbol{c}^\top\tb^*\right)}{\sqrt{\bc^\top\left[\mathcal{P}\nabla^2\cL(\tb_M)\mathcal{P}\right]^{+}\bc}}.
\end{align}
This uses the asymptotic theory with plug-in asymptotic variance using the true and estimated parameters,
where $\bc = P_\Theta(\boldsymbol{c})$ is the projection of $\boldsymbol{c}$ onto linear space $\Theta$. 
Figure \ref{alpha1beta1variance} shows that the histograms closely follow the standard Gaussian density. The first row of Figure \ref{alpha1beta1variance} is presented in the regime with $(p,L) = (1.25/n_a, 2)$. It holds that, even when the sample size is very small, the two density plots are still very close to the standard Gaussian density. The second row of Figure \ref{alpha1beta1variance} is drawn based on the setting where $(p,L) = (2/n_a, 20)$. In this case, the density plots of $A$ and $B$ are more stable and close to the standard Gaussian curve. These results, in turn, support our theoretical results in Theorem \ref{betainference}.

\begin{figure}[ht]
\begin{center}
\includegraphics[width=12cm]{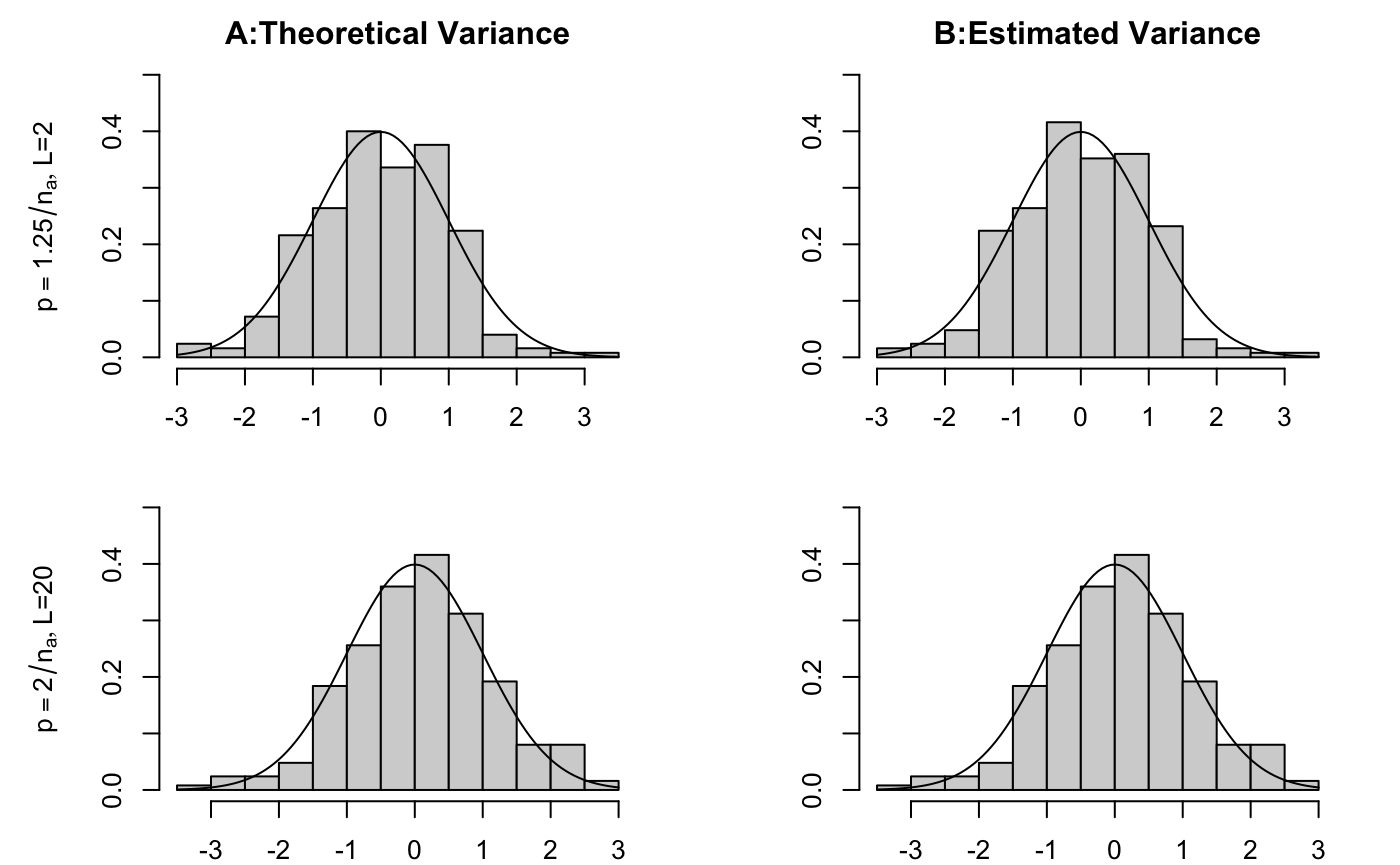}
\end{center}
\caption{Histograms of standardized random variables $A$ and $B$ in \eqref{AB} along with the density of the standardized Gaussian random variable. The first row is based on $(p,L) = (1.25/n_a, 2)$ and the second row is based on $(p,L) = (2/n_a, 20)$.}
\label{alpha1beta1variance}
\end{figure}
\subsection{Comparison with BTL Model Without Covariates}

{In this section, we conduct a series of simulations to compare the results with the model without using covariates in terms of the proportion of information being explained, prediction accuracy, and sensitivity analysis. 
\begin{itemize}

\item We show the first merit of our proposed method, in terms of the proportion of information being explained.

We record the following quantity $1 - \left\|\widehat{\balpha}_M\right\|_2^2 / \|\widehat{\balpha}_M +\bX\widehat{\bbeta}_M \|_2^2$, which quantifies the proportion of information being explained by the covariates. We consider the same $(p,L)$ pair presented in Section 5.1. For each $(p, L)$ pair, we generate the comparison graph $\mathcal{E}$, the comparison results $\{y_{i,j}^{(\ell)},l\in[L],(i,j)\in\mathcal{E}\}$ and solve the MLE $\widetilde{\bbeta}_M$ for $200$ times. We report the mean and standard deviation of the concerned quantity for each $(p, L)$ pair based on these $200$ repetitions. The results are presented in the following Table 1.

\begin{table}[h]\label{proportion_exp}
	\begin{center}
		\begin{tabular}[!htbp]{ p{2.5cm}|p{6cm}}
			$(p, L)$ & $1 - \left\|\widehat{\balpha}_M\right\|_2^2 / \|\widehat{\balpha}_M +\bX\widehat{\bbeta}_M \|_2^2$      \\
			\hline
			$(1, 50)$ &$0.370\pm 0.008$   \\ 
			$(0.5, 25)$ &$0.390\pm 0.016$    \\ 
			$(0.222, 25)$ &$0.415\pm 0.023$  \\
            $(0.625, 5)$ &$0.449\pm 0.028$  \\ 
			$(0.4, 5)$ &$0.486\pm 0.030$ \\ 
			$(0.278, 5)$ &$0.531\pm 0.032$  
		\end{tabular}
	\end{center}
	\caption{Mean and standard deviation of $1 - \left\|\widehat{\balpha}_M\right\|_2^2 / \|\widehat{\balpha}_M +\bX\widehat{\bbeta}_M \|_2^2$ for each $(p, L)$ pair based on $200$ repetitions.}
\end{table}

We conclude from Table 1, involving the covariates helps reduce the magnitude of unexplained information for all settings. Therefore, this further helps in making out-of-sample predictions, and we will discuss this point in the following step. 

\item  Second, we present the prediction performance of our model and compare it with BTL model without covariates. 
\iffalse
In terms of the simulation setting, we uniformly sample an orthogonal matrix $\bOmega$ from the space of all orthogonal matrices. Then we generate a new set of covariates $\bz_{i}\in \mathbb{R}^{d}$ (out-of-sample covariates) by setting $\bz_{i} = \bOmega\bx_{i}$. This ensures the projection operator onto the column space of $[\bx_1,\bx_2,\dots,\bx_n]^\top$ and $[\bz_1,\bz_2,\dots,\bz_n]^\top$ be the same.
\fi 
We generate a new set of covariates $[\bz_1,\bz_2,\dots,\bz_n]^\top$ where $\bz_{i}\in \mathbb{R}^{d}$ falls in the same column space of $\Xb$. 
With these new covariates, for any $i\neq j$, the out-of-sample probability $ \mathbb{P}\{\text {item } j \text { is preferred over item } i\}$ becomes 
\begin{align*}
    p_{i, j}^* = \frac{e^{\alpha_j^*+\bz_j^\top \bbeta^*}}{e^{\alpha_i^*+\bz_i^\top \bbeta^*}+e^{\alpha_j^*+\bz_j^\top \bbeta^*}}.
\end{align*}
Using the covariate information, the probability predicted by our model is
\begin{align*}
    \widehat{p}_{i, j}^{c} = \frac{e^{\widehat{\alpha}_{M, j}+\bz_j^\top \widehat{\bbeta}_{M}}}{e^{\widehat{\alpha}_{M, i}+\bz_i^\top \widehat{\bbeta}_{M}}+e^{\widehat{\alpha}_{M, j}+\bz_j^\top \widehat{\bbeta}_{M}}}.
\end{align*}
However, when one does not take the covariates information into consideration, the best one can do is to use the estimated score $\widehat{\theta}_i$ under the original BTL model, as \cite{chen2019spectral, gao2021uncertainty} did. In this case, the predicted probability is
\begin{align*}
    \widehat{p}_{i, j}^{nc} = \frac{e^{\widehat{\theta}_{j}}}{e^{\widehat{\theta}_{i}}+e^{\widehat{\theta}_{j}}}.
\end{align*}
In Table 2 we present the mean square error $\sum_{i<j}(\widehat{p}_{i,j}^{c} - p_{i,j}^*)^2$ and $\sum_{i<j}(\widehat{p}_{i,j}^{nc} - p_{i,j}^*)^2$ for the six $(p, L)$ pairs we mentioned above. The results show the mean and standard deviation of the mean square error calculated over $200$ repetitions. As we can see from Table 2, the estimator $\widehat{p}_{i, j}$, which takes the covariate information into consideration, performs much better than the one without covariate information.

\begin{table}[h]
	\begin{center}
		\begin{tabular}[!htbp]{ p{2.5cm}|p{5cm}p{5cm}}
			$(p, L)$ &  With covariates &  Without covariates   \\
			\hline
			$(1, 50)$ &$0.970\pm 0.094$ &  $110.325 \pm 1.377$ \\ 
			$(0.5, 25)$ &$3.909\pm 0.394$ & $113.453 \pm 3.086$    \\ 
			$(0.222, 25)$ &$8.975\pm 0.940$ & $118.796\pm 4.293$ \\
            $(0.625, 5)$ &$15.537\pm 1.572$ & $125.149\pm 5.441$ \\ 
			$(0.4, 5)$ &$24.874\pm 2.718$ & $134.796 \pm 7.432$ \\ 
			$(0.278, 5)$ &$36.487\pm 4.243$ & $146.286\pm 9.961$  
		\end{tabular}
	\end{center}

 \caption{Mean and standard deviation of mean square error $\sum_{i<j}(\widehat{p}_{i,j}^{c} - p_{i,j}^*)^2$ and $\sum_{i<j}(\widehat{p}_{i,j}^{nc} - p_{i,j}^*)^2$ under each setting based on $200$ simulations. }
\end{table}

%In the next step, we will also conduct a sensitivity analysis to compare our prediction results with a scenario where no covariates are considered.

\item Lastly, we test the sensitivity of our model by violating the linearity assumption of our model. And subsequently, we also compare its performance to that of the BTL model without covariates. 

We modify the underlying model to be
\begin{align*}
    p_{i, j}^* = \frac{e^{\alpha_j^*+\bw_j^\top \bbeta^* + g(\bw_j)}}{e^{\alpha_i^*+\bw_i^\top \bbeta^* + g(\bw_i)}+e^{\alpha_j^*+\bw_j^\top \bbeta^* + g(\bw_j)}},
\end{align*}
where $\bw_i = \bx_i$ when we fit the model and $\bw_i = \bz_i$ when we make prediction. Here $g(\cdot)$ is a nonlinear function. In our experiment, we fix $d = 5$ and let 
\begin{align}
    g(\bw_i) = c (0.1 (\bw_i)_1(\bw_i)_2 + 0.2(\bw_i)_2(\bw_i)_3+0.3 (\bw_i)_4(\bw_i)_5), \label{nonlinearity}
\end{align}
and $c$ is chosen to be $20$ and $100$ to accommodate different levels of non-linearity. With covariate information, the winning probability is still predicted as
\begin{align*}
    \widehat{p}_{i, j}^{c} = \frac{e^{\widehat{\alpha}_{M, j}+\bz_j^\top \widehat{\bbeta}_{M}}}{e^{\widehat{\alpha}_{M, i}+\bz_i^\top \widehat{\bbeta}_{M}}+e^{\widehat{\alpha}_{M, j}+\bz_j^\top \widehat{\bbeta}_{M}}}.
\end{align*}
while the predicted probability of the original BTL model is 
\begin{align*}
    \widehat{p}_{i, j}^{nc} = \frac{e^{\widehat{\theta}_{j}}}{e^{\widehat{\theta}_{i}}+e^{\widehat{\theta}_{j}}}.
\end{align*}
In Tables 3 and 4 we present the results when $c$ is chosen to be $20$ and $100$. We again consider the mean square error $\sum_{i<j}(\widehat{p}_{i,j}^{c} - p_{i,j}^*)^2$ and $\sum_{i<j}(\widehat{p}_{i,j}^{nc} - p_{i,j}^*)^2$ as the metric. The experiments are repeated $200$ times for each $(p, L)$ pair and we report the mean and standard deviation.

\begin{table}[h]
	\begin{center}
		\begin{tabular}[!htbp]{ p{2.5cm}|p{5cm}p{5cm}}
			$(p, L)$ &  With covariates &  Without covariates   \\
			\hline
			$(1, 50)$ &$4.160 \pm 0.271$ &  $201.799 \pm 2.056$ \\ 
			$(0.5, 25)$ &$7.173 \pm 0.625$ & $204.962 \pm 3.729$    \\ 
			$(0.222, 25)$ &$12.411 \pm 1.203$ & $209.903 \pm 6.010$ \\
            $(0.625, 5)$ &$18.871 \pm 1.885$ & $216.580 \pm 7.871$ \\ 
			$(0.4, 5)$ &$27.797 \pm 2.896$ & $225.813 \pm 11.310$ \\ 
			$(0.278, 5)$ &$39.372 \pm 4.024$ & $237.060 \pm 11.710$  
		\end{tabular}
	\end{center}
 \caption{Mean and standard deviation of mean square error $\sum_{i<j}(\widehat{p}_{i,j}^{c} - p_{i,j}^*)^2$ and $\sum_{i<j}(\widehat{p}_{i,j}^{nc} - p_{i,j}^*)^2$  under each setting based on $200$ simulations. The level of non-linearity in \eqref{nonlinearity} is chosen to be $c = 20$.}
 \label{c=20}
\end{table}

\begin{table}[h]
	\begin{center}
		\begin{tabular}[!htbp]{ p{2.5cm}|p{5cm}p{5cm}}
			$(p, L)$ &  With covariates &  Without covariates   \\
			\hline
			$(1, 50)$ &$49.074 \pm 0.960$ &  $102.936 \pm 1.438$ \\ 
			$(0.5, 25)$ &$51.679 \pm 1.986$ & $105.463 \pm 2.870$    \\ 
			$(0.222, 25)$ &$56.958 \pm 3.026$ & $111.269 \pm 4.317$ \\
            $(0.625, 5)$ &$63.858 \pm 4.076$ & $117.639 \pm 5.670$ \\ 
			$(0.4, 5)$ &$72.375 \pm 5.574$ & $125.958 \pm 7.556$ \\ 
			$(0.278, 5)$ &$84.260 \pm 7.122$ & $138.484 \pm 9.488$  
		\end{tabular}
	\end{center}
 \caption{Mean and standard deviation of mean square error $\sum_{i<j}(\widehat{p}_{i,j}^{c} - p_{i,j}^*)^2$ and $\sum_{i<j}(\widehat{p}_{i,j}^{nc} - p_{i,j}^*)^2$ under each setting based on $200$ simulations. The level of non-linearity in \eqref{nonlinearity} is chosen to be $c = 100$.}
 \label{c=100}
\end{table}
From Table 3 and 4 we can see that our model consistently performs better than the original BTL model when different levels of non-linearity exist.
\end{itemize}
}

\subsection{Application to Pokemon Challenge Dataset}\label{real_data_pokeman}

{We apply the proposed method to study the Pokemon Challenge Dataset. The original dataset can be found at \url{https://www.kaggle.com/c/intelygenz-pokemon-challenge/data}. This dataset records the pairwise competition records among 800 Pokemon, whose covariate information is also recorded. This dataset contains $50000$ competition results, each competition takes place between two pokemons and has one winner. %We split the dataset into training set and test set.% The covariates are first normalized to have mean $0$ and standard deviation $1$, and then multiplied by $2/27$. 
%Each data is selected to be a training data with probability $0.7$ independently.

Our experiments mainly focus on predicting the ability of the mega evolved pokemons. We think that a mega evolved pokemon has the same intrinsic ability (same $\alpha_i$) as pre-evolutionary pokemon, and the mega evolution may only change the covariates $\bx$. We have $48$ mega evolved pokemons in this dataset, and we randomly select $28$ of them to test our predictive performance.  Among these remaining $800-28 = 772$ pokemons for training purposes, %we construct the comparison graph using the training dataset. This graph is not connected initially, and 
we select the largest connected component of their comparison graph. Eventually we have $757$ pokemons left for training. For each pokemon, we select $\log(\emph{Attack})$, $\log(\emph{HP})$, \emph{Mega or not} as their covariates. Here \emph{Attack} and \emph{HP} denote the ability to attack and durability, respectively. The variable \emph{Mega or not} takes binary value and represents whether this pokemon is mega evolved or not. We optimize the likelihood of our CARE model in \eqref{nll} using training data and record $\hat\balpha_M$ and $\hat\bbeta_M$. %The detailed optimization procedure of training is presented in \S\ref{detail_real_data}

We first investigate the statistical significance of these 3 variables we are interested in. This amounts to testing the following hypothesis testing problems for each feature:
\begin{align*}
    H_0: \beta_i^* = 0\quad  \text{ v.s. } \quad H_a: \beta_i^* \neq 0,\; i\in [3].
\end{align*}
The test statistics are given by $\beta_{M, i} / \sqrt{([\cP\nabla^2\cL(\tb_M)\cP]^+)_{n+i, n+i}}$ for all $i\in [3]$ and the corresponding p-values are calculated via the asymptotic normality results in \S\ref{Inference}. The results are depicted in Table \ref{pvalues}, and these three variables are all statistically significant.

\begin{table}[h]
	\begin{center}
		\begin{tabular}[!htbp]{ p{2.5cm}|p{5cm}p{5cm}}
			 &  Estimate &  p-value   \\
			\hline
			Attack &$2.743$ &  $<1e-5$ \\ 
			HP &$3.759$ & $<1e-5$    \\ 
			Mega or not &$1.603$ & $<1e-5$
		\end{tabular}
	\end{center}
 \caption{P-values for the regression coefficients}
 \label{pvalues}
\end{table}

We then evaluate the competitions performance of the $28$ mega evolved pokemons in the test sample, whose pre-evolutionary versions are the training data. We predict the score of an evolved pokemon as $\hat\theta^{\text{predicted}}:=\textsf{SOFT}(\hat\alpha_{M, \text{pe}}, \tau_{\text{pe}})+\boldsymbol{z}^\top \hat\bbeta_M,$ where $\hat\alpha_{M, \text{pe}}$ is the estimated intrinsic score of the pre-evolutionary version and $\bz$ is the new covariates of this mega evolved pokemon. Here we apply a soft-thresholding to the $\ba$ part with $\textsf{SOFT}(x, \tau) = \sign(x) \cdot\max\{|x| - \tau, 0\}$ and $\tau_i = \Phi^{-1}(1 - 0.025 / 757)\cdot \sqrt{([\cP\nabla^2\cL(\tb_M)\cP]^+)_{i, i}}$ for each pokemon $i$ in the training dataset. This corresponds to set those estimates $\hat\alpha_i$ that are statistically indifferent from
zero to zero. The use of significant level $0.025 / 757$ is to control the familywise false positive rates at a level of $0.05$. 

In order to formalize the metric we used to test our procedure, we introduce the following notation first. Let $\mathcal{T}_2$ be set of the pokemons we want to predict and let $\mathcal{T}_1$ be set of the pokemons in the training data who have competitions with pokemons in $\mathcal{T}_2$. Given any two pokemons $i$ and $j$, we let $D_{i, j}$ be the competitions between $i$ and $j$ and define $y_{i, j} = \sum_{y\in D_{i, j}}y / |D_{i, j}|$. For each pokemon $i$ from $\mathcal{T}_1$, we further define the estimated score as $\hat\theta^{\text{estimated}}_i:=\textsf{SOFT}(\hat\alpha_{M, i}, \tau_i)+\bx_i^\top \wh\bb_M$. We use the following quantities $MSE_c$ and $Logistic_c$ as measures of the prediction loss for pokemons in $\cT_2$ when we have covariate information
\begin{align*}
    MSE_{c} =& \sum_{i\in \cT_1, j \in \cT_2}\left(\frac{e^{\hat\theta^{\text{predicted}}_j}}{e^{\hat\theta^{\text{estimated}}_i}+ e^{\hat\theta^{\text{predicted}}_j}} - y_{i,j}\right)^2 + \sum_{i, j \in \cT_2}\left(\frac{e^{\hat\theta^{\text{predicted}}_j}}{e^{\hat\theta^{\text{predicted}}_i}+ e^{\hat\theta^{\text{predicted}}_j}} - y_{i, j}\right)^2,\\
    Logistic_{c} =& -\sum_{i\in \cT_1, j \in \cT_2} \left(y_{i, j}\log\left(\frac{e^{\hat\theta^{\text{predicted}}_j}}{e^{\hat\theta^{\text{estimated}}_i}+ e^{\hat\theta^{\text{predicted}}_j}}\right)+(1-y_{i, j})\log\left(\frac{1}{e^{\hat\theta^{\text{estimated}}_i}+ e^{\hat\theta^{\text{predicted}}_j}}\right)\right) \\
    &-\sum_{i, j \in \cT_2}\left(y_{i, j}\log\left(\frac{e^{\hat\theta^{\text{predicted}}_j}}{e^{\hat\theta^{\text{estimated}}_i}+ e^{\hat\theta^{\text{predicted}}_j}}\right)+(1-y_{i, j})\log\left(\frac{1}{e^{\hat\theta^{\text{estimated}}_i}+ e^{\hat\theta^{\text{predicted}}_j}}\right)\right).
\end{align*}
As a comparison, for the original BTL model, we let $\widehat{\theta}^{\text{BTL}}$ be the estimated score under the original BTL model. Since there no covariate information is involved, for pokemons from $\mathcal{T}_2$, their $\hat\theta^{\text{predicted}}$ are set to be the estimated scores of their pre-evolutionary versions. In addition, for pokemons from $\mathcal{T}_1$, their scores $\hat\theta^{\text{estimates}}$ are set to be $\hat\theta^{\text{BTL}}$ estimated via the training data. The loss for our method and original BTL model are reported in Table \ref{prediction}. As we can see, our model achieves a significant improvement over the original BTL model.

\begin{table}[h]
	\begin{center}
		\begin{tabular}[!htbp]{p{3cm} |p{5cm}p{5cm}}
			   Loss&With covariates &  Without covariates   \\
			\hline
			
			MSE&$268.93$ & $302.04$  \\
            Logistic& $752.10$ & $875.45$
		\end{tabular}
	\end{center}
 \caption{Prediction errors for the model with covariates and without using covariates, respectively}
 \label{prediction}
\end{table}
}

\section{Conclusion and Discussion}
In this paper, we study a Covariate-Assisted Ranking Estimation (CARE) model systematically. This allows us to incorporate the covariate information of compared items into the ranking framework, which includes the standard BTL model as a particular case. We derive the minimal sample complexity required for statistical consistency and uncertainty quantification for MLE based on novel proof techniques and illustrate the theory and methods using  the mutual fund holding dataset. The empirical results lend further support to the CARE model over the classical BTL model.

There are a few future directions worth exploring. First, it is worth extending the idea of incorporating covariates into a more general ranking framework, such as the Plackett-Luce or nonparametric models, under a more general comparison graph. In contrast, our work only studies the BTL model with the Erdős-Rényi comparison graph. Second, it would be interesting if some structure assumptions exist on the parameters $\{\alpha_i^*\}_{i=1}^n$ and $\bb^*$, such as sparsity. In this scenario, one shall leverage certain regularizers on $\balpha$ and $\bb$ in the likelihood function to achieve a solution that generalizes well. Third, except for the covariate, one may also incorporate time information into the ranking framework as in many real applications, the underlying scores of compared items change over time. Lastly, in our paper, we consider the scenario where the underlying score of the $i$-th item is given by $\alpha_i^*+\mathbf{x}_i^\top\bb^*,i\in[n],$ in the sense that the overall score of the $i$-th item is the summation of its intrinsic score $\{\alpha_i^*\}_{i=1}^n$ and its covariate times one specific evaluation criterion $\bb$. It would be interesting if we do a ranking based on data evaluated from multiple sources. In specific, suppose that we have $L$ users and $n$ items and the score of the $i$-th item, $i\in[n]$ evaluated by the $\ell$-th person, $\ell\in [L],$ is $\alpha_i^*+\mathbf{x}_i^\top\bb_{\ell}$.  It would be interesting to derive novel statistical estimation and uncertainty quantification principles for ranking models under this setting. We will leave these open problems for future research.

\acks{Research supported by the NSF grants DMS-2052926, DMS-2053832, DMS-2210833, and ONR  N00014-22-1-2340.}


\newpage

\appendix
\begin{center}
{ \textbf{\LARGE Supplementary materials to ``Uncertainty Quantification of MLE for Entity Ranking
		with Covariates"}}
	\end{center}
\section{Proof Outline of Estimation Results}\label{EstimationOutline}
In this section, we first present the proof outline for Theorem \ref{noregularization}. The detailed proof of Theorem \ref{noregularization} is presented in \S \ref{regularizedtoriginal}. 

 To understand statistical error of  $\tb_M$, we begin with analyzing the regularized MLE $\tb_\lambda$ 
\begin{align}
    \tb_\lambda = \argmin\limits_{\tb\in \Theta} \mathcal{L}_\lambda(\tb),\label{MLE}
\end{align}
where $\mathcal{L}_\lambda(\tb) := \mathcal{L}(\tb)+\frac{\lambda}{2}\Vert \tb\Vert_2^2$ for $\lambda > 0,$ and then make connections with $\tb_M$ by a properly chosen $\lambda$. 
{The introduction of \(\ell_2\)- regularization as a intermediate step is essential for ensuring the MLE fall in a bounded area around the ground truth through this regularized MLE. Therefore, strong convexity of the loss holds in this bounded area.
Prevailing methodologies for examining the BTL model (\cite[Theorem 6]{chen2019spectral} and \cite[Lemma 8.5]{chen2022partial}) also relies on this $\ell_2$ regularization to ensure strong convexity of the loss.}

Before proceeding, we introduce the following two quantities $\kappa_2$ and $\kappa_3$ indicating the difficulty of recovering $\tb^*$
\begin{align*}
    \kappa_2 := \max_{i\in [n]}|\alpha_i^*|, \quad \kappa_3 := \frac{\left\Vert \tb^*\right\Vert_2}{\sqrt{n(d+1)}}.
\end{align*}
The regularized MLE solves a strong convex problem whose  estimation error bounds are derived in Theorem \ref{mainthm} below. 

\begin{theorem}\label{mainthm}
Suppose $np > c_p\log n$ for some $c_p>0$ and $d< n, d\log n\lesssim np$. We consider $L \leq  c_4\cdot n^{c_5}$ for any absolute constants $c_4,c_5>0$ and$$\displaystyle\lambda = c_{\lambda}\min\left\{\frac{\kappa_1}{\kappa_2},\frac{1}{\kappa_3\sqrt{d+1}}\right\}\sqrt{\frac{np\log n}{L}}$$ for some $c_\lambda>0$. Let $\tb_\lambda = (\wh\ba_\lambda^\top,\wh\bb_\lambda^\top)^\top$ be the solution of the regularized MLE Eq.~\eqref{MLE}. Then with probability at least $1-O(n^{-6})$, we have
\begin{align*}
\Vert \wh\ba_\lambda-\ba^*\Vert_\infty&\lesssim \kappa_1^2\sqrt{\frac{(d+1)\log n}{npL}}; \quad
\left\Vert \wh\bb_\lambda-\bb^*\right\Vert_2 \lesssim\kappa_1\sqrt{\frac{\log n}{pL}};\\
\left\Vert\tX\tb_\lambda-\tX\tb^*\right\Vert_\infty&\lesssim \kappa_1^2 \sqrt{\frac{(d+1)\log n}{npL}}; \quad
 \frac{\left\Vert e^{\tX\tb_\lambda}-e^{\tX\tb^*}\right\Vert_\infty}{\left\Vert e^{\tX\tb^*}\right\Vert_\infty} \lesssim \kappa_1^2 \sqrt{\frac{(d+1)\log n}{npL}},
\end{align*}
where $\tX = [\tx_1,\tx_2,\dots,\tx_n]^\top$.
\end{theorem}
Theorem \ref{mainthm} presents the statistical rates of the regularized MLE in \eqref{MLE}. Before presenting the formal proof for Theorem \ref{mainthm}, in the following several subsections  \ref{app:preliminary}, \ref{app:gradient} and \ref{leaveoneout} , we focus on establishing its associated building blocks.

%{\blue Since our final goal is to prove the statistical rates of the MLE of unregularized loss function \eqref{MLEun},  in \S \ref{regularizedtoriginal}, we provide the formal proof of Theorem \ref{noregularization}  with the help of Theorem \ref{mainthm}}. 

{Given that our final objective is to establish the statistical rates of the Maximum Likelihood Estimator (MLE) applied to the unregularized loss function \eqref{MLEun}, we present a formal proof of Theorem \ref{noregularization} in \S \ref{regularizedtoriginal}, leveraging the insights provided by Theorem \ref{mainthm}.}
 
 \subsection{Preliminaries and Basic Results}\label{app:preliminary}
In this subsection, we study the theoretical properties of the gradient and Hessian of the loss function in \eqref{MLE}.   Their expressions are given by 
\begin{align}
    \nabla \mathcal{L}(\tb) &= \sum_{(i,j)\in\mathcal{E}, i>j}\left\{-y_{j,i}+\frac{e^{\tx_i^\top \tb}}{e^{\tx_i^\top \tb}+e^{\tx_j^\top \tb}}\right\}(\tx_i-\tx_j), \label{formulagradient} \\
    \nabla^2 \mathcal{L}(\tb) &= \sum_{(i,j)\in\mathcal{E}, i>j}\frac{e^{\tx_i^\top \tb}e^{\tx_j^\top \tb}}{\left(e^{\tx_i^\top \tb}+e^{\tx_j^\top \tb}\right)^2}(\tx_i-\tx_j)(\tx_i-\tx_j)^\top.\label{formulaHessian}
\end{align}
The gradient of $\cL(\tb)$ at $\tb^*$ is controlled by the following lemma.
\begin{lemma}\label{gradient}
   With $\lambda$ given by Theorem \ref{mainthm}, the following event 
\begin{align*}
    \mathcal{A}_1 = \left\{\left\|\nabla \mathcal{L}_{\lambda}\left(\tb^*\right)\right\|_{2} \leq C_0\sqrt{\frac{n^{2} p \log n}{L}}\right\}
\end{align*}
happens with probability exceeding $1-O(n^{-11})$ for some $C_0>0$ which only depend on $c_\lambda$.
\end{lemma}
\begin{proof}
For the proof of Lemma \ref{gradient}, we refer to \S\ref{proofgradient} for more details. 
\end{proof}
Next, we proceed to analyzing the Hessian matrix $\nabla^2 \mathcal{L}_\lambda(\tb)$. First, we consider $\LG = \sum_{(i,j)\in\mathcal{E}, i>j}(\tx_i-\tx_j)(\tx_i-\tx_j)^\top$ and study  its eigenvalues in Lemma \ref{eigenlem}.
\begin{lemma}\label{eigenlem}
Suppose $pn > c_p\log n$ for some $c_p>0$. The following event
\begin{align*}
    \mathcal{A}_2 = \left\{\frac{1}{2}c_2pn\leq \lambda_{\text{min}, \perp}(\LG)\leq \Vert \LG\Vert\leq 2c_1pn \right\}
\end{align*}
happens with probability exceeding $1-O(n^{-11})$ when $n$ is large enough.
\end{lemma}
\begin{proof}
	See \S\ref{proofeigenlemma} for a detailed proof.
\end{proof}
In the rest of the content, without loss of generality, we assume the conditions stated in Lemma \ref{eigenlem} hold. 
Moreover, with the help of Lemma \ref{eigenlem}, we next analyze the Hessian $\nabla^2 \mathcal{L}_\lambda(\tb)$ and summarize its theoretical properties in Lemma \ref{ub} and Lemma \ref{lb}, respectively.

\begin{lemma}\label{ub}  Suppose event $\mathcal{A}_2$ holds, we obtain
\begin{align*}
    \lambda_{\text{max}}(\nabla^2 \mathcal{L}_\lambda(\tb))\leq \lambda +\frac{1}{2}c_1pn,\quad \forall \tb\in\mathbb{R}^{n+d}.
\end{align*}
\end{lemma}
\begin{proof}
Since $\displaystyle\frac{e^{\tx_i^\top \tb}e^{\tx_j^\top \tb}}{\left(e^{\tx_i^\top \tb}+e^{\tx_j^\top \tb}\right)^2}\leq \frac{1}{4}$, we have 
\begin{align*}
    \lambda_{\text{max}}(\nabla^2 \mathcal{L}_\lambda(\tb))\leq \lambda +\frac{1}{4}\Vert \LG\Vert\leq \lambda +\frac{1}{2}c_1pn,\quad \forall \tb\in\mathbb{R}^{n+d}.
\end{align*}
\end{proof}

\begin{lemma}\label{lb}
Suppose event $\mathcal{A}_2$ happens. Then for all $\tb$ such that $\Vert\ba-\ba^*\Vert_\infty\leq C_1$, $\Vert \bb-\bb^*\Vert_2\leq C_2$, we have
\begin{align*}
    \lambda_{\text{min},\perp}(\nabla^2 \mathcal{L}_\lambda(\tb))\geq\lambda+\frac{c_2pn}{8\kappa_1e^C},
\end{align*}
where $ C = 2C_1+2\sqrt{\frac{c_3(d+1)}{n}}C_2$.
\end{lemma}

\begin{proof}
	The formal proof of this Lemma \ref{lb} can be found in \S\ref{proofLemmaA4}.
\end{proof}
%What's more, we assume $L=O(n^5)$ in the following proof. It's easy to modify the proof to prove results when $L\leq c_4\cdot n^{c_5}$.

In the following two subsections \S\ref{app:gradient} and \S\ref{leaveoneout}, we understand the statistical rates of the regularized estimator $\tb_\lambda$ via analyzing the iterates in our gradient method.

\subsection{Convergence of Projected Gradient Descent}\label{app:gradient}

In this subsection, we consider a sequence of iterates $\{\tb^t\}_{t = 0,1,\dots}$ which is generated by the following projected gradient descent algorithm with stepsize $\eta = \frac{2}{2\lambda+c_1np}$ and the number of iterations $T = n^5$.
\begin{algorithm}[H]
\caption{Gradient descent for regularized MLE.}
\begin{algorithmic}
\STATE \textbf{Initialize} $\tb^0 = \tb^*$

\FOR{$t=0,1,\dots, T-1$ } 
\STATE {$\zeta = \tb^t-\eta \nabla \mathcal{L}_\lambda (\tb^t)$}
\STATE {$\tb^{t+1} = \mathcal{P}\zeta$}

\ENDFOR
\end{algorithmic}
\end{algorithm}

\noindent 
We initialize at $\tb^*$ and the target loss given in \eqref{MLE} is strongly convex.
{The projected gradient descent is employed (since the likelihood has a linear constraint) to ensure the iterates $\widetilde{\boldsymbol{\beta}}^t$ converge to the $\ell_2$-regularized MLE $\widetilde{\boldsymbol{\beta}}_\lambda$ exponentially. In the following section, using the leave-one-out analysis, we also show $\widetilde{\boldsymbol{\beta}}^t$ at the same time stays close to $\widetilde{\boldsymbol{\beta}}^*$ for all $t\leq \textbf{poly}(n)$. Therefore, combine these two parts together, we are able to conclude that $\widetilde{\boldsymbol{\beta}}_\lambda$ is also close to $\widetilde{\boldsymbol{\beta}}^*$.}

%and we show that the iterates $\tb^t$ has geometric convergence rate to $\tb_{\lambda}$. 

We summarize the theoretical findings in the following Lemma \ref{lem2}, Lemma \ref{lem3} and Lemma \ref{lem5}, respectively.

\begin{lemma}\label{lem2}
Under event $\mathcal{A}_2$, we have
\begin{align*}
    \left\Vert \tb^t-\tb_\lambda\right\Vert_2\leq \rho^t\left\Vert \tb^0-\tb_\lambda\right\Vert_2,
\end{align*}
where $\displaystyle\rho = 1-\frac{2\lambda}{2\lambda+c_1 np}$.
\end{lemma}

Next, we prove that the initial point is not far from $\tb_{\lambda}$.
\begin{lemma}\label{lem3}
On the event $\mathcal{A}_1$ happens, it follows that
\begin{align*}
    \left\Vert \tb^0-\tb_\lambda\right\Vert_2 = \left\Vert \tb_\lambda -\tb^*\right\Vert_2 \leq \frac{2C_0}{c_\lambda}\max\left\{\frac{\kappa_2}{\kappa_1}, \kappa_3\sqrt{d+1}\right\}\sqrt{n}.
\end{align*}
\end{lemma}
\begin{proof}
	See \S\ref{proofLemmaA6} for a detailed proof.
\end{proof}
\noindent Combining Lemma \ref{lem2} and Lemma \ref{lem3}, we obtain the following result on the optimization error. % The statistical error $\tb^T$ can be found in the following section through the leave-one-out trick.

\begin{lemma}\label{lem5}
	On event $\mathcal{A}_1 \cap \mathcal{A}_2$, there exists some constant $C_7$ such that 
\begin{align*}
    \left\Vert \tb^T-\tb_\lambda\right\Vert_2 \leq C_7\kappa_1\sqrt{\frac{(d+1)\log n}{npL}}.
\end{align*}
\end{lemma}
\begin{proof}
	See \S\ref{prooflem5} for a detailed proof.
\end{proof}
In this subsection, we prove that the iterate $\tb^T$ converges to $\tb_{\lambda}$ geometrically and enjoys a good optimization error after $T=n^5$ iterations. In order to prove the distance between $\tb_{\lambda}$ and $\tb^*$, i.e. the statistical error of $\tb_{\lambda}$,  in the next subsection, we leverage the leave-one-out technique and use induction to prove that the iterate $\tb^T$ stays close to our initial point $\tb^0=\tb^*$, even after $T=n^5$ iterations.

\subsection{Analysis of Leave-one-out Sequences}\label{leaveoneout}
In this section, we construct the leave-one-out sequences \citep{ma2018implicit,chen2019spectral,chen2020noisy} and bound the statistical error by induction. To construct the leave-one-out sequence, we consider the following loss function for any $m\in[n]$.
\begin{align*}
    \mathcal{L}^{(m)}(\tb) =& \sum_{(i, j) \in \mathcal{E}, i>j, i\neq m, j\neq m}\left\{-y_{j, i}\left(\tx_i^\top\tb-\tx_j^\top\tb\right)+\log \left(1+e^{\tx_i^\top\tb-\tx_j^\top\tb}\right)\right\}\\
    &+p\sum_{i\neq m}\left\{-\frac{e^{\tx_i^\top\tb^*}}{e^{\tx_i^\top\tb^*}+e^{\tx_m^\top\tb^*}}(\tx_i^\top\tb-\tx_m^\top\tb) + \log \left(1+e^{\tx_i^\top\tb-\tx_m^\top\tb}\right)\right\} ;\\
    \mathcal{L}_\lambda^{(m)}(\tb)=&\mathcal{L}^{(m)}(\tb)+\frac{\lambda}{2}\Vert\tb\Vert_2^2.
\end{align*}
Then for any $m\in[n]$, we construct the leave-one-out sequence $\left\{\tb^{t,(m)}\right\}_{t = 0,1,\dots}$ in the way of Algorithm \ref{alg2}.

\begin{algorithm}[H]
\caption{Construction of leave-one-out sequences.}
	\begin{algorithmic}[1]		
		\STATE \textbf{Initialize} $\tb^{0,(m)} = \tb^*$
        \FOR{$t=0,1,\dots, T-1$ } 
        \STATE {$\zeta = \tb^{t,(m)}-\eta \nabla \mathcal{L}^{(m)}_\lambda \left(\tb^{t,(m)}\right)$} 
        \STATE {$\tb^{t+1,(m)} = \mathcal{P}\zeta$}
        \ENDFOR
	\end{algorithmic}\label{alg2}
\end{algorithm}

With the help of the leave-one-out sequences, we do induction to demonstrate that the iterate $\tb^T$ will not be far away from $\tb^*$ when $T=n^5.$
In specific, we take again $\eta = \frac{2}{2\lambda+c_1 np}$ and $T = n^5$.  With the leave-one-out sequences in hand, we prove the following bounds by induction for $t \leq T$. 
\begin{align}
    \left\Vert \tb^t-\tb^*\right\Vert_2&\leq C_3\kappa_1\sqrt{\frac{\log n}{pL}};\tag{A} \label{inductionA}\\
    \max_{1\leq m\leq n}\left\Vert\tb^t- \tb^{t,(m)}\right\Vert_2&\leq C_4\kappa_1\sqrt{\frac{(d+1)\log n}{npL}}\leq C_4\kappa_1\sqrt{\frac{\log n}{pL}}; \tag{B}\\
    \max_{1\leq m\leq n}|\alpha_m^{t,(m)}-\alpha_m^*|&\leq C_5\kappa_1^2\sqrt{\frac{(d+1)\log n}{npL}};\tag{C} \\
    \Vert \ba^t-\ba^*\Vert_\infty&\leq C_6 \kappa_1^2\sqrt{\frac{(d+1)\log n}{npL}} .\tag{D}\label{inductionD}
   % \left\Vert \bb^t-\bb^*\right\Vert_2&\leq C_7\kappa_1\sqrt{\frac{(d+1)\log n}{npL}}.\tag{E}
\end{align}
For $t = 0$, since $\tb^0 = \tb^{0,(1)} = \tb^{0,(2)} = \dots = \tb^{0,(n)} = \tb^*$, the \eqref{inductionA}$\sim$ \eqref{inductionD} hold automatically. In the following lemmas, we prove the conclusions of \eqref{inductionA}-\eqref{inductionD} for the $(t+1)$-th iteration are true when the results hold for the $t$-th iteration.
\begin{lemma}\label{induction1}
Suppose bounds \eqref{inductionA}$\sim$ \eqref{inductionD} hold for the $t$-th iteration. With probability exceeding $1-O(n^{-11})$ we have
\begin{align*}
    \left\Vert \tb^{t+1}-\tb^*\right\Vert_2&\leq C_3\kappa_1\sqrt{\frac{\log n}{pL}},
\end{align*}
as long as $\displaystyle 0<\eta\leq \frac{2}{2\lambda+c_1np}$, $\displaystyle C_3\geq \frac{20C_0}{c_2}$ and $n$ is large enough.
\end{lemma}
\begin{proof}See \S\ref{rsc} for a detailed proof.
\end{proof}

\begin{lemma}\label{induction3}
Suppose bounds \eqref{inductionA}$\sim$ \eqref{inductionD} hold for the $t$-th iteration. With probability exceeding $1-O(n^{-11})$ we have
\begin{align*}
    \max_{1\leq m\leq n}\left\Vert\tb^{t+1}- \tb^{t+1,(m)}\right\Vert_2&\leq C_4\kappa_1\sqrt{\frac{(d+1)\log n}{npL}},
\end{align*}
as long as $\displaystyle 0<\eta\leq \frac{2}{2\lambda+c_1np}$, $\displaystyle C_4\gtrsim \frac{1}{c_2}$ and $np\gtrsim (d+1)\log n$.
\end{lemma}
\begin{proof}
	See \S\ref{prooflem10} for a detailed proof.
\end{proof}

\begin{lemma}\label{induction2}
Suppose bounds \eqref{inductionA}$\sim$ \eqref{inductionD} hold for the $t$-th iteration. With probability exceeding $1-O(n^{-11})$ we have
\begin{align*}
    \max_{1\leq m\leq n}|\alpha_m^{t+1,(m)}-\alpha_m^*|&\leq C_5\kappa_1^2\sqrt{\frac{(d+1)\log n}{npL}},
\end{align*}
as long as $C_5\geq 30c_0(C_0+c_1C_3+c_1C_4)$, $C_5\geq 7.5(1+2\sqrt{c_3})(C_3+C_4)$, $C_5\geq 30c_\lambda/\sqrt{d+1}$ and $n$ is large enough.
\end{lemma}
\begin{proof}
	See \S\ref{prooflem9} for a detailed proof.
\end{proof}

\begin{lemma}\label{induction4}
Suppose bounds \eqref{inductionA}$\sim$ \eqref{inductionD} hold for the $t$-th iteration. With probability exceeding $1-O(n^{-11})$ we have
\begin{align*}
    \Vert \ba^{t+1}-\ba^*\Vert_\infty&\leq C_6\kappa_1^2\sqrt{\frac{(d+1)\log n}{npL}},
\end{align*}
as long as $C_6\geq C_4+C_5$ and $n$ is large enough.
\end{lemma}
\begin{proof}
	See \S\ref{prooflem11} for a detailed proof.
	\end{proof}
	
With these necessary building blocks at hand, we next prove Theorem \ref{mainthm} by providing statistical rates of the regularized estimator in \eqref{MLE}.

\subsection{Proof of Theorem \ref{mainthm}}\label{proofthm1}
In this subsection, we aim at providing proof for Theorem \ref{mainthm} by combining the results established above. 

For any $m\in[n]$, by Taylor expansion, one obtains
\begin{align*}
\left|e^{\tx_m^\top\tb_\lambda}-e^{\tx_m^\top\tb^*}\right|&\leq e^{\omega_m^*}\left|\tx_m^\top\tb_\lambda-\tx_m^\top\tb^*\right|. \\
&\leq e^{\tx_m^\top\tb^* +\left|\tx_m^\top\tb_\lambda-\tx_m^\top\tb^*\right|}\left|\tx_m^\top\tb_\lambda-\tx_m^\top\tb^*\right|
%&\leq \max_{1\leq m\leq n }e^{\tx_m^\top\tb^*}e^{\max\limits_{1\leq m\leq n }\left|\tx_m^\top\tb_\lambda-\tx_m^\top\tb^*\right|}\max_{1\leq m\leq n}\left|\tx_m^\top\tb_\lambda-\tx_m^\top\tb^*\right|,
\end{align*}
where $\omega_m^*$ is some real number between $\tx_m^\top\tb_\lambda$ and $\tx_m^\top\tb^*$. As a result, it holds that
\begin{align*}
\frac{\left\Vert e^{\tX\tb_\lambda}-e^{\tX\tb^*}\right\Vert_\infty}{\left\Vert e^{\tX\tb^*}\right\Vert_\infty}&\leq \frac{\max_{1\leq m\leq n }e^{\tx_m^\top\tb^*}e^{\max\limits_{1\leq m\leq n }\left|\tx_m^\top\tb_\lambda-\tx_m^\top\tb^*\right|}\max_{1\leq m\leq n}\left|\tx_m^\top\tb_\lambda-\tx_m^\top\tb^*\right|}{\max_{1\leq m\leq n }e^{\tx_m^\top\tb^*}} \\
&\leq e^{\max\limits_{1\leq m\leq n }\left|\tx_m^\top\tb_\lambda-\tx_m^\top\tb^*\right|}\max_{1\leq m\leq n}\left|\tx_m^\top\tb_\lambda-\tx_m^\top\tb^*\right|.
\end{align*}
By the Cauchy-Schwartz inequality, we have
\begin{align*}
\left|\tx_m^\top\tb_\lambda-\tx_m^\top\tb^*\right|&\leq |(\wh\alpha_\lambda)_m-\alpha^*_m|+\left|\bx_m^\top\wh\bb_\lambda-\bx_m^\top\bb^*\right| \\
&\leq \Vert\wh\ba_\lambda-\ba^*\Vert_\infty +\Vert\bx_m\Vert_2\Vert\wh\bb_\lambda-\bb^*\Vert_2 \\
&\leq \Vert\ba^T-\ba^*\Vert_\infty+\Vert\wh\ba_\lambda-\ba^T\Vert_\infty + \Vert\bx_m\Vert_2\Vert\bb^T-\bb^*\Vert_2+\Vert\bx_m\Vert_2\Vert\wh\bb_\lambda-\bb^T\Vert_2 \\
&\leq \Vert\ba^T-\ba^*\Vert_\infty+\left\Vert \tb_\lambda-\tb^T\right\Vert_2 + \Vert\bx_m\Vert_2\Vert\bb^T-\bb^*\Vert_2+\Vert\bx_m\Vert_2\left\Vert \tb_\lambda-\tb^T\right\Vert_2 \\
&\leq C_8\kappa_1^2 \sqrt{\frac{(d+1)\log n}{npL}},
\end{align*}
where $\displaystyle C_8\geq C_6+\left(1+\sqrt{c_3}\right)C_7+\sqrt{c_3}C_3$ and $n$ is large enough. The last inequlity holds from the results derived in \S\ref{app:gradient} and \S\ref{leaveoneout}. Then for $n$ and $L$ such that $ C_8\kappa_1^2 \sqrt{\frac{(d+1)\log n}{npL}}\leq 0.6$, we have
\begin{align*}
\frac{\left\Vert e^{\tX\tb_\lambda}-e^{\tX\tb^*}\right\Vert_\infty}{\left\Vert e^{\tX\tb^*}\right\Vert_\infty}
%&\leq e^{\max\limits_{1\leq m\leq n }\left|\tx_m^\top\tb_\lambda-\tx_m^\top\tb^*\right|}\max_{1\leq m\leq n}\left|\tx_m^\top\tb_\lambda-\tx_m^\top\tb^*\right| \\
&\leq 2C_8\kappa_1^2 \sqrt{\frac{(d+1)\log n}{npL}}.
\end{align*}
Similarly, we also obtain
\begin{align*}
\Vert\wh\ba_\lambda-\ba^*\Vert_\infty&\leq \Vert\ba^T-\ba^*\Vert_\infty+\Vert\ba_\lambda-\ba^T\Vert_\infty \\
&\leq C_6\kappa_1^2\sqrt{\frac{(d+1)\log n}{npL}}+C_7\kappa_1\sqrt{\frac{(d+1)\log n}{npL}} \\
&\lesssim \kappa_1^2\sqrt{\frac{(d+1)\log n}{npL}};\\
\Vert\wh\bb_\lambda-\bb^*\Vert_2&\leq\Vert\bb^T-\bb^*\Vert_2+\Vert\bb_\lambda-\bb^T\Vert_2\\
&\leq C_3\kappa_1\sqrt{\frac{\log n}{pL}}+C_7\kappa_1\sqrt{\frac{(d+1)\log n}{npL}}
% &\leq C_3\kappa_1\sqrt{\frac{\log n}{pL}}+C_7\kappa_1\sqrt{\frac{\log n}{pL}} \\
\lesssim \kappa_1\sqrt{\frac{\log n}{pL}}.
\end{align*}

Next, we use the proof ideas and conclusions from Theorem \ref{mainthm} to prove the statistical rate of a non-regularized MLE defined in \eqref{MLEun}.

\subsection{Proof of Theorem \ref{noregularization}}\label{regularizedtoriginal}
With all necessary building blocks at hand, in this subsection, we provide the formal proof for Theorem \ref{noregularization}. In specific, we assume $L = O(n^2)$ in the following proof and it is easy to extend the proof to the regime $L\leq c_4\cdot n^{c_5}$. The way to solve this is changing the power $11$ in Lemma \ref{gradient} and Lemma \ref{eigenlem} to a larger number.
\begin{proof}
	Consider the following MLE
	\begin{align}
	\tb_{\text{con}} := \argmin_{\tb\in\Theta,\Vert\ba-\ba^*\Vert_\infty\leq 0.025,\Vert\bb-\bb^*\Vert_2\leq 0.025\sqrt{n/(2c_3d+2c_3)}} \cL(\tb). \label{conMLE}
	\end{align}
	We choose $c_\lambda$ in the definition of $\lambda$ such that
	\begin{align*}
	\lambda n^3\geq 1\text{ and } \left(\kappa_2 n+\sqrt{n(d+1)}\kappa_3+(C_3+C_7)\kappa_1\sqrt{\frac{\log n}{pL}}\right)\lambda\leq \frac{c_2pn}{20}C_9\sqrt{\frac{\log n}{n^2pL}}
	\end{align*}
	for some $C_9>0$. As long as $L\leq c_4\cdot n^{c_5}$ for some absolute constants $c_4,c_5>0$, the proof of Lemma \ref{lem5} is still valid and Theorem \ref{mainthm} still holds for this $\lambda$. For $n$ large enough such that $\tb_\lambda$ satisfies the constraints in Eq.~\eqref{conMLE}, by the optimality of $\tb_{\text{con}}$ we know that $\cL(\tb_\lambda)\geq \cL(\tb_{\text{con}})$. On the other hand, by Taylor's expansion
	\begin{align*}
	\cL(\tb_{\text{con}}) = \cL(\tb_\lambda)+\nabla\cL(\tb_\lambda)^\top \left(\tb_{\text{con}}-\tb_\lambda\right)+\frac{1}{2}\left(\tb_{\text{con}}-\tb_\lambda\right)^\top \nabla^2\cL(\boldsymbol{c})\left(\tb_{\text{con}}-\tb_\lambda\right),
	\end{align*}
	where $\boldsymbol{c} = \xi \tb_{\text{con}} +(1-\xi)\tb_\lambda $ for some $\xi \in [0,1]$.

	This leads to
	\begin{align}
	\nabla\cL(\tb_\lambda)^\top \left(\tb_{\text{con}}-\tb_\lambda\right)+\frac{1}{2}\left(\tb_{\text{con}}-\tb_\lambda\right)^\top \nabla^2\cL(\boldsymbol{c})\left(\tb_{\text{con}}-\tb_\lambda\right)\leq 0.\label{eq1}
	\end{align}
	Next, we first define the norm $\Vert\cdot\Vert_c$ as
	\begin{align*}
	\left\Vert \tb\right\Vert_c := \max_{i\in [n]}\left|\tb_i\right|+\sqrt{\frac{c_3(d+1)}{n}}\sqrt{\sum_{j = n+1}^{n+d}\tb_j^2}, \qquad\forall\tb\in\mathbb{R}^{n+d}.
	\end{align*}
	%Then for any $i\in [n]$ and $\tb\in \mathbb{R}^{n+d}$, we have $\displaystyle\left|\tx_i^\top\tb\right|\leq \left\Vert\tb\right\Vert_c$.
	
	Therefore, we have %\r{Q: need to be recalled}
	\begin{align*}
	\Vert\boldsymbol{c}-\tb^*\Vert_c \leq \max\left\{\Vert \tb_\lambda-\tb^*\Vert_c,\Vert \tb_{\text{con}}-\tb^*\Vert_c\right\}\leq  0.1
	\end{align*}
	as long as 
	\begin{align*}
	\left[2(C_6+C_7)\kappa_1^2+2\sqrt{2c_3}(C_3+C_7)\kappa_1\right]\sqrt{\frac{(d+1)\log n}{npL}}\leq 0.1.
	\end{align*}
As a result, by Lemma \ref{lb} we have
	\begin{align}
	\lambda_{\text{min},\perp}(\nabla^2\cL(\boldsymbol{c}))\geq \lambda+\frac{c_2pn}{8\kappa_1e^C}\geq \frac{c_2 pn}{10\kappa_1}.\label{eq2}
	\end{align}
	Combine Eq.~\eqref{eq1} and Eq.~\eqref{eq2} we have
	\begin{align*}
	\frac{c_2 pn}{20\kappa_1}\left\Vert \tb_{\text{con}}-\tb_\lambda\right\Vert_2^2&\leq \frac{1}{2}\left(\tb_{\text{con}}-\tb_\lambda\right)^\top \nabla^2\cL(\boldsymbol{c})\left(\tb_{\text{con}}-\tb_\lambda\right) \\
	&\leq -\nabla\cL(\tb_\lambda)^\top \left(\tb_{\text{con}}-\tb_\lambda\right) \\
	&\leq \left\Vert \nabla\cL(\tb_\lambda)\right\Vert_2\left\Vert \tb_{\text{con}}-\tb_\lambda\right\Vert_2.
	\end{align*}
	Therefore, after some simple calculations, it holds that %we have
	\begin{align*}
	\left\Vert \tb_{\text{con}}-\tb_\lambda\right\Vert_2\leq \frac{20\kappa_1}{c_2 pn}\left\Vert \nabla\cL(\tb_\lambda)\right\Vert_2 \leq \frac{20\kappa_1\lambda}{c_2 pn}\left\Vert \tb_\lambda\right\Vert_2\leq C_9\kappa_1\sqrt{\frac{\log n}{n^2pL}}.
	\end{align*}
	And, when $n$ is large enough, we have
	\begin{align}
	&\Vert \wh\ba_{\text{con}}-\ba^*\Vert_\infty\leq \Vert \wh\ba_{\text{con}}-\wh\ba_\lambda\Vert_\infty+\Vert \wh\ba_\lambda-\ba^*\Vert_\infty\leq C_9\kappa_1\sqrt{\frac{\log n}{n^2pL}}+ (C_6+C_7)\kappa_1^2\sqrt{\frac{(d+1)\log n}{npL}}\leq 0.01,\label{eq3} \\ 
	&\Vert \wh\bb_{\text{con}}-\bb^*\Vert_2\leq \Vert \wh\bb_{\text{con}}-\wh\bb_\lambda\Vert_2+\Vert \wh\bb_\lambda-\bb^*\Vert_2\leq C_9\kappa_1\sqrt{\frac{\log n}{n^2pL}}+ (C_3+C_7)\kappa_1\sqrt{\frac{\log n}{pL}}\leq 0.01\sqrt{\frac{n}{2c_3(d+1)}},\label{eq4}
	\end{align}
where $\wh\ba_{\text{con}} := (\tb_{\text{con}})_{1:n}$ and $\wh\bb_{\text{con}} := (\tb_{\text{con}})_{n+1:n+d}$. As a result, $\tb_{\text{con}}$ falls in the interior of the inequality constraints in Eq.~\eqref{conMLE}.   By the convexity of $\cL$ and its strong convexity in $\Theta$,  we have  $\tb_{\text{con}} = \tb_M$.  Therefore, by Eq.~\eqref{eq3} and Eq.~\eqref{eq4} we have
	\begin{align}
	&\Vert \wh\ba_M-\ba^*\Vert_\infty = \Vert \wh\ba_{\text{con}}-\ba^*\Vert_\infty\leq C_{10}\kappa_1^2\sqrt{\frac{(d+1)\log n}{npL}}, \\
	&\Vert \wh\bb_M-\bb^*\Vert_2 = \Vert \wh\bb_{\text{con}}-\bb^*\Vert_2\leq C_{11}\kappa_1\sqrt{\frac{\log n}{pL}}.
	\end{align}
	The result of $\left\Vert\tX^\top\tb_M-\tX^\top\tb^*\right\Vert_\infty$ and $\frac{\left\Vert e^{\tX^\top\tb_M}-e^{\tX^\top\tb^*}\right\Vert_\infty}{\left\Vert e^{\tX^\top\tb^*}\right\Vert_\infty}$ hold based on the same derivations in Section \ref{proofthm1} and we omit the corresponding details.
\end{proof}

\newpage
\section{Proof of Inference Results in Section \ref{Inference}}\label{InferenceOutline}
{In this section we first introduce our extension to top-K testing and provide the proof outlines for Theorem \ref{inferenceresidual}. Then, we prove Theorem \ref{betainference} and Corollaries  \ref{alphainference} and \ref{cor:inferenceb}, based on the results of Theorem \ref{inferenceresidual}. The details of  proofs in this section is deferred to Section \ref{ProofDetails}. We next introduce some building blocks for the proof of main theorems.}
\subsection{Extension to Top-K set hypothesis testing}
{In this section, we extend the method to conduct the top-K test using estimators for overall scores $\theta_i^*:=\alpha_i^*+\xb_i^\top\bbeta^*,i\in[n]$ (for simplicity, here we define the overall scores as $\theta_i^*$ and its estimators as $\hat\theta_i:=\hat\alpha_i+\xb_i^\top\hat\bbeta,i\in [n]$). %and intrinsic scores only $\alpha_i^*,i\in [n].$ 

We use the following statistics to construct top-K test simultaneously for all elements in $\cM.$
\begin{align*}
\cT := \max_{m\in \cM}\max_{k\neq m}\frac{\widehat{\theta}_k - \widehat{\theta}_m - (\theta_k^* - \theta_m^*)}{\widehat{\sigma}_{m,k}}
\end{align*}
 It's distribution can be approximated by the bootstrap counterpart
\begin{align*}
\cG:= \max_{m\in \cM}\max_{k\neq m}\sum_{l=1}^L \sum_{(i,j)\in \cE,i>j} \frac{(\eb_m-\eb_k)^\top(\nabla^2\cL(\hat\theta)^{+})(\eb_i-\eb_j)}{\widehat{\sigma}_{m,k}L}(\phi(\hat\theta_i-\hat\theta_j)-y_{j, i}^{(l)})\omega_{j, i}^{(l)}.
\end{align*}
We are able to achieve similar theoretical results by following the similar proof procedure of Theorem 5 in \cite{fan2023spectral} for $\cT$ and $\cG$.
Let $c_{1-\alpha}$ be the $(1-\alpha)$-th quantile of $\cG$, we have the following theorem for the test statistics $\cT$.
\begin{theorem}\label{T23thm}
Under the conditions of Theorem 1, we have
    \begin{align*}
        &\left|\PP(\cT>c_{1-\alpha})-\alpha\right|\rightarrow 0.
 %       &\left|P(\cT_3>c_{3,1-\alpha})-\alpha\right|\lesssim \left(\frac{\kappa_1^3\log^5 n}{np}\right)^{1/4} + \frac{\kappa_1^{3.5}(d+1)\log n}{\sqrt{np}}\left(\kappa_1^3\sqrt{\frac{\log n}{L}}+\kappa_1^4\sqrt{\frac{\log n}{L(d+1)}}+\sqrt{k+d}\right).
    \end{align*}
\end{theorem}

Next, we introduce some applications on constructing (simultaneous) one-sided confidence intervals for out-of-sample ranks via the distribution of $\cT$ in the following two examples.
\begin{example}
    For an item $m$ of interest, and let $K$ be the targeted rank threshold, we are interested in the following testing problem
    \begin{align}
        H_0: r(m)\leq K \quad \text{ versus }\quad  H_1:r(m)>K. \label{ranktest}
    \end{align}
    Let $\widehat{c}_{1-\alpha}$ be the estimated $(1-\alpha)$-th quantile of $\cT$ from the bootstrap samples. As a result, we have
    \begin{align*}
        \PP\left(\theta_k^*-\theta_m^*\geq \widehat{\theta}_k - \widehat{\theta}_m-\widehat{c}_{1-\alpha} \widehat{\sigma}_{m,k}\right) \geq 1-\alpha.
    \end{align*}
    Similarly, this implies 
    \begin{align*}
        \PP\left(r(m)\geq 1+\sum_{k\neq m}\mathbf{1}(\widehat{\theta}_k - \widehat{\theta}_m> \widehat{c}_{1-\alpha} \widehat{\sigma}_{m,k})\right) \geq 1-\alpha.
    \end{align*}
    This yields a critical region at a significance level of $alpha$ for the test \eqref{ranktest}
    \begin{align*}
        \left\{1+\sum_{k\neq m}\mathbf{1}(\widehat{\theta}_k - \widehat{\theta}_m> \widehat{c}_{1-\alpha} \widehat{\sigma}_{m,k})>K\right\}.
    \end{align*}
\end{example}
}
\subsection{Proof Outline}
\begin{lemma}\label{Lilemma}
	For $i\in [n]$, with probability at least $1-O(n^{-10})$ we have
	\begin{itemize}
		\item $\displaystyle\left|\left(\nabla\cL(\tb^*)\right)_i\right|\lesssim\sqrt{\frac{np\log n}{L}}$;
		\item $\displaystyle\sum_{j\neq i}\left(\nabla^2\cL(\tb^*)\right)_{i,j}^2\lesssim np(1+dp)$, \quad $\displaystyle\sum_{k>n}\left(\nabla^2\cL(\tb^*)\right)_{i,k}^2\lesssim  ndp^2$,  \quad $\displaystyle\sum_{j\in [n],j\neq i}\left|\left(\nabla^2\cL(\tb^*)\right)_{i,j}\right|\lesssim np$.
		%\item $\left|\tx_i\tb_M-\tx_i\tb^*\right|\lesssim \left\|\tb_M-\tb^*\right\|_c$,
		\item  $ \left|y_{j,i}-\mathbb{E}y_{j,i}\right|\lesssim\sqrt{\frac{\log n}{L}}$, for any $i,j\in[n], i\neq j$.
	\end{itemize}
	%And, with probability at least $1-O(n^{-10})$, for any $i,j\in[n], i\neq j$, we have
	%\begin{align*}
	%\left|y_{j,i}-\mathbb{E}y_{j,i}\right|\lesssim\sqrt{\frac{\log n}{L}}.
	%\end{align*}
\end{lemma}
\begin{proof}
	See \S\ref{prooflemb1} for a detailed proof.
	\end{proof}
Recall that we define the norm $\Vert\cdot\Vert_c$ as
\begin{align*}
    \left\Vert \tb\right\Vert_c := \max_{i\in [n]}\left|\tb_i\right|+\sqrt{\frac{c_3(d+1)}{n}}\sqrt{\sum_{j = n+1}^{n+d}\tb_j^2}, \qquad\forall\tb\in\mathbb{R}^{n+d}.
\end{align*}
Then for any $i\in [n]$ and $\tb\in \mathbb{R}^{n+d}$, we have $\displaystyle\left|\tx_i^\top\tb\right|\leq \left\Vert\tb\right\Vert_c$.

\subsection{Proof Outline of Theorem \ref{inferenceresidual}}
In this subsection, we provide the proof outline for Theorem \ref{inferenceresidual}. The following lemma gives a bound for $\left\Vert\Delta\tb\right\Vert_2:=\|\tb_M-\ob\|_2$, which validates the first part of Theorem \ref{inferenceresidual}.
\begin{theorem}\label{inferencemainthm}
Under the assumptions of Theorem \ref{noregularization}, with probability at least $1-O(n^{-6})$ we have
\begin{align*}
    \left\Vert \tb_M-\ob\right\Vert_2\lesssim\kappa_1^4\frac{(d+1)^{0.5}\log n}{\sqrt{n}pL}.
\end{align*}
\end{theorem}
\begin{proof}
	See \S\ref{proveinferencemainthm} for a detailed proof.
\end{proof}
Next, we  consider controlling the magnitude of $\Delta\alpha_i = \wh\alpha_{M,i}-\oa_i$ for $i\in [n]$ in order to prove the second part of Theorem \ref{inferenceresidual}. 

Recall that we define the quadratic approximation $\ocL(\cdot)$ to $\cL(\cdot)$ in \eqref{quadratic_expan}, which is also given as below:
\begin{align*}
\ocL(\tb) = \cL(\tb)+\left(\tb-\tb^*\right)^\top\nabla \cL(\tb^*)+\frac{1}{2}\left(\tb-\tb^*\right)^\top\nabla^2 \cL(\tb^*)\left(\tb-\tb^*\right).
\end{align*}

We adopt the following notation for a given vector $\mathbf{x}\in \mathbb{R}^{n+d}$,
\begin{align*}
    \ocL\bigg|_{\mathbf{x}_{-i}}(x_i) = \ocL(\check{\bb})\bigg|_{\check{\bb}_i=x_i, \check{\bb}_{-i}=\mathbf{x}_{-i}},
\end{align*}
which acts as a marginal likelihood function for the $i$-th coordinate $x_i$ given the other coordinates $\mathbf{x}_{-i}$ fixed. Accroding to this definition, we have the following proposition. The proof of Proposition \ref{oaproposition} is included in \S \ref{proofpro1pro2}.

\begin{proposition}\label{oaproposition}
For $i\in [n]$, $\oa_i$ is the minimizer of the univariate function $\ocL|_{\ob_{-i}}$.
\end{proposition}
By changing the $n+d-1$ coordinates that we fix, we define $\oa_i'$ as the minimizer of $\ocL|_{\tb_{M,-i}}(x_i)$. %\r{Q:  do we fix at the truth? A: No, we fixed at the minimizer of the quadratic approximation.}
 Here we fix $\tb_{M,-i}$ and optimize $\ocL$ based on this fixed parameter. Explicitly, the minimizer is calculated as
\begin{align*}
    \oa_i' = \alpha_i^*-\frac{\left(\nabla\cL(\tb^*)\right)_i+\sum\limits_{j\neq i}\left(\tb_{M,j}-\tb^*_j\right)\left(\nabla^2\cL(\tb^*)\right)_{i,j}}{\left(\nabla^2\cL(\tb^*)\right)_{i,i}}.
\end{align*}

\noindent In order to bound $|\wh\alpha_{M,i}-\oa_i|,$ we  bound $\left|\oa_i'-\oa_i\right|$ and $ |\wh\alpha_{M,i}-\oa_i'|$ separately. 

In terms of $|\oa_i'-\oa_i|$, we provide an upper bound for this quantity in Lemma \ref{alphainferencelem1}.
\begin{lemma}\label{alphainferencelem1}
Under the assumptions of Theorem \ref{noregularization}, as long as
$\kappa_1^2\sqrt{(d+1)\log n/{npL} }= \mathcal{O}(1)$, for $i\in [n]$, with probability at least $1-O(n^{-6})$ we have
\begin{align*}
    \left|\oa_i'-\oa_i\right| \lesssim \kappa_1^5 \frac{(d+1)\log n}{npL} +\frac{\kappa_1^4}{np}\sqrt{\frac{(d+1)\log n}{L}}\left(\sqrt{d+1}+\frac{\log n}{\sqrt{np}}\right)+\kappa_1^2\left(\sqrt{\frac{d+1}{np}}+\frac{\log n}{np}\right)\left\Vert\wh\ba_M-\oba\right\Vert_\infty.
\end{align*}
\end{lemma}
\begin{proof}The detailed proof of Lemma \ref{alphainferencelem1} is given in \S\ref{proofalphainferencelem1}.
	\end{proof}

\noindent On the other hand, for a given $\mathbf{x}\in \mathbb{R}^{n+d},$ we let % \r{Q: do we repeat? A: Actually no, the loss functions are different.}
\begin{align}
    \cL \bigg|_{\mathbf{x}_{-i}}(x_i) = \cL(\check{\bb})\bigg|_{\check{\bb}_i=x, \check{\bb}_{-i}=\mathbf{x}_{-i}}.\label{Lconditioned}
\end{align}
Here we consider the marginal loss of $\cL(\cdot)$ in the $i$-th coordinate given other coordinates fixed. Similar to Proposition \ref{oaproposition}, we have the following proposition for $\cL|_{\tb_{M,-i}}(x)$. The proof of Proposition \ref{alphaproposition} is also included in \S \ref{proofpro1pro2}.

\begin{proposition}\label{alphaproposition}
For $i\in [n]$, $\wh\alpha_{M,i}$ is the minimizer of the univariate function $\cL|_{\tb_{M,-i}}(x)$.
\end{proposition}

We now describe the intuition of bounding $ |\wh\alpha_{M,i}-\oa_i'|$. 
Since $\alpha_{M,i}$ is the minimizer of $\cL|_{\tb_{M,-i}}(x)$ and $\oa_i'$ is the minimizer of $\ocL|_{\tb_{M,-i}}(x)$. Therefore, as long as $\ocL|_{\tb_{M,-i}}(\cdot)$ and $\cL|_{\tb_{M,-i}}(\cdot)$ are close enough,  the difference between their minimizers $\wh\alpha_{M,i}$ and $\oa_i'$ is small. We summarize this finding in the following lemma \ref{alphainferencelem2}.
\begin{lemma}\label{alphainferencelem2}
Under the assumptions of Theorem \ref{noregularization}, for $i\in [n]$, with probability at least $1-O(n^{-6})$ we have
\begin{align*}
    \left|\wh\alpha_{M,i}-\oa_i'\right| \lesssim\kappa_1^6\frac{(d+1)\log n}{npL}.
\end{align*}
\end{lemma}
\begin{proof}
The proof of Lemma \ref{alphainferencelem2} is presented in \S\ref{proof_lemb3}.
\end{proof}
\noindent Finally, combining the conclusions of Lemma \ref{alphainferencelem1} and Lemma \ref{alphainferencelem2} we get the following theorem for $\left|\wh\alpha_{M,i}-\oa_i\right|$.

\begin{theorem}\label{inferencemainthm2}
Under the assumptions of Theorem \ref{noregularization}, as long as $\kappa_1^2\sqrt{(d+1)\log n/ npL} = \mathcal{O}(1)$ and $\kappa_1^2\left(\sqrt{(d+1)/np}+\log n/np\right)\leq c$ for some fixed constant $c>0$, with probability at least $1-O(n^{-5})$, for $i\in [n]$ we have
\begin{align*}
    \left|\wh\alpha_{M,i}-\oa_i\right| \lesssim\kappa_1^6\frac{(d+1)\log n}{np L}+\frac{\kappa_1^4}{np}\sqrt{\frac{(d+1)\log n}{L}}\left(\sqrt{d+1}+\frac{\log n}{\sqrt{np}}\right).
\end{align*}
If we further assume $np\gtrsim (\log n)^2$, then with probability at least $1-O(n^{-5})$, for $i\in [n]$ we have
\begin{align*}
    \left|\wh\alpha_{M,i}-\oa_i\right| \lesssim\kappa_1^6\frac{(d+1)\log n}{np L}+\kappa_1^4\frac{d+1}{np}\sqrt{\frac{\log n}{L}}.
\end{align*}
\end{theorem}

\iffalse
\begin{proof}[Proof of Theorem \ref{inferenceresidual}]
The conclusion of Theorem \ref{inferenceresidual} follows directly from conclusions of Theorem \ref{inferencemainthm} and Theorem \ref{inferencemainthm2}.
\end{proof}
\fi 

\begin{proof}[Proof of Theorem \ref{inferencemainthm2}]
Combining Lemma \ref{alphainferencelem1} and Lemma \ref{alphainferencelem2}, we know that \begin{align*}
    \left\Vert\wh\ba_M-\oba\right\Vert_\infty \lesssim& \kappa_1^6 \frac{(d+1)\log n}{npL} +\frac{\kappa_1^4}{np}\sqrt{\frac{(d+1)\log n}{L}}\left(\sqrt{d+1}+\frac{\log n}{\sqrt{np}}\right)\\
    &+\kappa_1^2\left(\sqrt{\frac{d+1}{np}}+\frac{\log n}{np}\right)\left\Vert\wh\ba_M-\oba\right\Vert_\infty
\end{align*}
with probability at least $1-O(n^{-5})$. To reveal the constant hidden in the above inequality, we write it as
\begin{align*}
	\left\Vert\wh\ba_M-\oba\right\Vert_\infty \leq& C_{\text{Hidden}}\left(\kappa_1^6 \frac{(d+1)\log n}{npL} +\frac{\kappa_1^4}{np}\sqrt{\frac{(d+1)\log n}{L}}\left(\sqrt{d+1}+\frac{\log n}{\sqrt{np}}\right)\right)\\
	&+C_{\text{Hidden}}\kappa_1^2\left(\sqrt{\frac{d+1}{np}}+\frac{\log n}{np}\right)\left\Vert\wh\ba_M-\oba\right\Vert_\infty
\end{align*}
with probability at least $1-O(n^{-5})$. We choose $c = 1/(2C_{\text{Hidden}})$ in Theorem \ref{inferencemainthm2}. As a result, as long as $\kappa_1^2\left(\sqrt{(d+1)/{np}}+\log n/np\right) \leq c$, we have
\begin{align*}
    \left\Vert\wh\ba_M-\oba\right\Vert_\infty &\leq \frac{C_{\text{Hidden}}}{1-0.5}\left(\kappa_1^6 \frac{(d+1)\log n}{npL} +\frac{\kappa_1^4}{np}\sqrt{\frac{(d+1)\log n}{L}}\left(\sqrt{d+1}+\frac{\log n}{\sqrt{np}}\right)\right) \\
    &\lesssim\kappa_1^6 \frac{(d+1)\log n}{npL} +\frac{\kappa_1^4}{np}\sqrt{\frac{(d+1)\log n}{L}}\left(\sqrt{d+1}+\frac{\log n}{\sqrt{np}}\right)
\end{align*}
with probability at least $1-O(n^{-5})$.
\end{proof}

\subsection{Proof of Theorem \ref{inferenceresidual}}
\begin{proof}[Proof of Theorem \ref{inferenceresidual}]
The conclusion of Theorem \ref{inferenceresidual} follows directly from conclusions of Theorem \ref{inferencemainthm} and Theorem \ref{inferencemainthm2}.
\end{proof}

We next prove Theorem \ref{betainference} based on the results of Theorem \ref{inferenceresidual}.
\subsection{Proof of Theorem \ref{betainference}}
This subsection aims at deriving theoretical proof for Theorem \ref{betainference}. %We first introduce some properties about $\tb$. 

\begin{proof}%[Proof of Theorem \ref{betainference}]
The following content is conditioned on the event $\mathcal{A}_2$.
	Recall that $\bc$ is the projection of $\boldsymbol{c}$ onto linear space $\Theta$. Therefore, we obtain
	\begin{align*}
	\boldsymbol{c}^\top\ob-\boldsymbol{c}^\top\tb^* = \bc^\top\ob-\bc^\top\tb^*.
	\end{align*}
By Proposition \ref{pro2} we have $\mathcal{P}\nabla\cL(\tb^*)+\mathcal{P}\nabla^2\cL(\tb^*)(\ob-\tb^*) = \boldsymbol{0}$. Since $\ob, \tb^*\in\Theta$, we also have $\mathcal{P}\nabla\cL(\tb^*)+\mathcal{P}\nabla^2\cL(\tb^*)\mathcal{P}(\ob-\tb^*) = \boldsymbol{0}$. Let $\boldsymbol{v} = \left[\mathcal{P}\nabla^2\cL(\tb^*)\mathcal{P}\right]^{+}\bc$. Then we have
\begin{align*}
    0 &= \boldsymbol{v}^\top\left(\mathcal{P}\nabla\cL(\tb^*)+\mathcal{P}\nabla^2\cL(\tb^*)\mathcal{P}(\ob-\tb^*)\right) \\
    &= \boldsymbol{v}^\top\mathcal{P}\nabla\cL(\tb^*) + \bc^\top\ob-\bc^\top\tb^*.
\end{align*}
As a result, we have
	\begin{align*}
	\bc^\top\ob-\bc^\top\tb^* &= -\boldsymbol{v}^\top\mathcal{P}\nabla \cL(\tb^*) \\
	&=\sum_{l=1}^L\sum_{(i,j)\in\mathcal{E}, i>j}\frac{1}{L}\left\{y_{j,i}^{(l)}-\frac{e^{\tx_i^\top \tb}}{e^{\tx_i^\top \tb}+e^{\tx_j^\top \tb}}\right\}\boldsymbol{v}^\top\mathcal{P}(\tx_i-\tx_j).
	\end{align*}
	For $(i,j)$ such that $(i,j)\in\mathcal{E}, i>j$ and $l\in[L]$, we define 
	\begin{align*}
	X_{i,j}^{(l)} = \frac{1}{L}\left\{y_{j,i}^{(l)}-\frac{e^{\tx_i^\top \tb}}{e^{\tx_i^\top \tb}+e^{\tx_j^\top \tb}}\right\}\boldsymbol{v}^\top\mathcal{P}(\tx_i-\tx_j).
	\end{align*}
	Then we have
	\begin{align*}
	\frac{\mathbb{E}\left|X_{i,j}^{(l)}\right|^3}{\mathbb{E}\left(X_{i,j}^{(l)}\right)^2}\leq &= \frac{\left|\boldsymbol{v}^\top\mathcal{P}(\tx_i-\tx_j)\right|}{L}\cdot \left(\phi(\tx_i^\top\tb^*-\tx_i^\top\tb^*)^2+(1-\phi(\tx_i^\top\tb^*-\tx_i^\top\tb^*))^2\right) \\
	&\leq \frac{\left\|\bc\right\|_2\left\|\left[\mathcal{P}\nabla^2\cL(\tb^*)\mathcal{P}\right]^{+}\right\|\left\|\mathcal{P}\left(\tx_i-\tx_j\right)\right\|_2}{L}\lesssim \frac{\left\|\bc\right\|_2}{\lambda_{\text{min},\perp}(\nabla^2\cL(\tb^*))L}\lesssim \frac{\kappa_1}{npL}\left\|\bc\right\|_2
	\end{align*}
	with probability exceeding $1-O(n^{-10})$ (randomness comes from $\mathcal{G}$). As a result, by \cite{berry1941accuracy} we have
	\begin{align*}
	\sup_{x\in\mathbb{R}}\left|\mathbb{P}\left(\frac{\bc^\top\ob-\bc^\top\tb^*}{\sqrt{\text{Var}\left[\bc^\top\ob\mid \mathcal{G}\right]}}\leq x\bigg| \mathcal{G}\right)-\mathbb{P}\left(\mathcal{N}(0,1)\leq x\right)\right|&\leq \frac{\max\limits_{(i,j)\in\mathcal{E},i>j, l\in[L]} \mathbb{E}\left|X_{i,j}^{(l)}\right|^3/\mathbb{E}\left(X_{i,j}^{(l)}\right)^2}{\sqrt{\text{Var}\left[\bc^\top\ob\mid \mathcal{G}\right]}} \\
	&\lesssim\frac{\kappa_1}{npL\sqrt{\text{Var}\left[\bc^\top\ob\mid \mathcal{G}\right]}}\left\|\bc\right\|_2
	\end{align*}
	with probability exceeding $1-O(n^{-10})$ (randomness comes from $\mathcal{G}$). And, we know that $\displaystyle\text{Var}\left[\bc^\top\ob\mid \mathcal{G}\right] = \frac{1}{L} \bc^\top\left[\mathcal{P}\nabla^2\cL(\tb^*)\mathcal{P}\right]^{+}\bc$ and
	\begin{align*}
	\bc^\top\left[\mathcal{P}\nabla^2\cL(\tb^*)\mathcal{P}\right]^{+}\bc\gtrsim \frac{1}{pn}\Vert \bc\Vert_2^2
	\end{align*}
	with probability exceeding $1-O(n^{-10})$. Therefore,
	\begin{align*}
	\sup_{x\in\mathbb{R}}\left|\mathbb{P}\left(\frac{\sqrt{L}\left(\bc^\top\ob-\bc^\top\tb^*\right)}{\sqrt{\bc^\top\left[\mathcal{P}\nabla^2\cL(\tb^*)\mathcal{P}\right]^{+}\bc}}\leq x\bigg| \mathcal{G}\right)-\mathbb{P}\left(\mathcal{N}(0,1)\leq x\right)\right| \lesssim\frac{\kappa_1}{\sqrt{npL}}
	\end{align*}
	with probability exceeding $1-O(n^{-10})$ (randomness comes from $\mathcal{G}$). On the other hand, by Theorem \ref{inferenceresidual} we have
	\begin{align*}
	\left|\frac{\sqrt{L}\left(\boldsymbol{c}^\top\tb_M-\boldsymbol{c}^\top\ob\right)}{\sqrt{\bc^\top\left[\mathcal{P}\nabla^2\cL(\tb^*)\mathcal{P}\right]^{+}\bc}}\right|\lesssim& \left[\kappa_1^6\frac{(d+1)\log n}{\sqrt{np L}}+\kappa_1^4\sqrt{\frac{(d+1)\log n}{np}}\left(\sqrt{d+1}+\frac{\log n}{\sqrt{np}}\right)\right]\frac{\Vert\boldsymbol{c}_{1:n}\Vert_1}{\Vert \bc\Vert_2} \\
	&+\kappa_1^4\frac{(d+1)^{0.5}\log n}{\sqrt{pL}}\frac{\Vert \boldsymbol{c}_{n+1:n+d}\Vert_2}{\Vert \bc\Vert_2}
	\end{align*}
	with probability at least $1-O(n^{-5})$. Therefore, we conclude the first part of Theorem \ref{betainference}.
	
	We next take all randomness into consideration. For simplicity we denote by $$ \Gamma = \left[\kappa_1^6\frac{(d+1)\log n}{\sqrt{np L}}+\kappa_1^4\sqrt{\frac{(d+1)\log n}{np}}\left(\sqrt{d+1}+\frac{\log n}{\sqrt{np}}\right)\right]\frac{\Vert\boldsymbol{c}_{1:n}\Vert_1}{\Vert \bc\Vert_2} +\kappa_1^4\frac{(d+1)^{0.5}\log n}{\sqrt{pL}}\frac{\Vert \boldsymbol{c}_{n+1:n+d}\Vert_2}{\Vert \bc\Vert_2}.$$
	To begin with, for fixed $x$ we have
	\begin{align*}
	&\left|\mathbb{P}\left(\frac{\sqrt{L}\left(\bc^\top\ob-\bc^\top\tb^*\right)}{\sqrt{\bc^\top\left[\mathcal{P}\nabla^2\cL(\tb^*)\mathcal{P}\right]^{+}\bc}}\leq x\right)-\mathbb{P}\left(\mathcal{N}(0,1)\leq x\right)\right| \\
	=& \left|\mathbb{E}_\mathcal{G}\left[\mathbb{P}\left(\frac{\sqrt{L}\left(\bc^\top\ob-\bc^\top\tb^*\right)}{\sqrt{\bc^\top\left[\mathcal{P}\nabla^2\cL(\tb^*)\mathcal{P}\right]^{+}\bc}}\leq x\bigg|\mathcal{G}\right)\right]-\mathbb{P}\left(\mathcal{N}(0,1)\leq x\right)\right| \\
	\leq&\mathbb{E}_\mathcal{G}\left|\mathbb{P}\left(\frac{\sqrt{L}\left(\bc^\top\ob-\bc^\top\tb^*\right)}{\sqrt{\bc^\top\left[\mathcal{P}\nabla^2\cL(\tb^*)\mathcal{P}\right]^{+}\bc}}\leq x\bigg|\mathcal{G}\right)-\mathbb{P}\left(\mathcal{N}(0,1)\leq x\right)\right| \\
	\lesssim &\frac{\kappa_1}{\sqrt{npL}}+\frac{1}{n^{10}}.
	\end{align*}
	As a result, we have
	\begin{align}
	\sup_{x\in\mathbb{R}}\left|\mathbb{P}\left(\frac{\sqrt{L}\left(\bc^\top\ob-\bc^\top\tb^*\right)}{\sqrt{\bc^\top\left[\mathcal{P}\nabla^2\cL(\tb^*)\mathcal{P}\right]^{+}\bc}}\leq x\right)-\mathbb{P}\left(\mathcal{N}(0,1)\leq x\right)\right|\lesssim\frac{\kappa_1}{\sqrt{npL}}+\frac{1}{n^{10}}.\label{eq26}
	\end{align}
	Consider event $\displaystyle A = \left\{\left|\frac{\sqrt{L}\left(\bc^\top\tb_M-\bc^\top\ob\right)}{\sqrt{\bc^\top\left[\mathcal{P}\nabla^2\cL(\tb^*)\mathcal{P}\right]^{+}\bc}}\right|\leq \Lambda \Gamma\right\}$, where $\Lambda>0$ is some constant such that $\mathbb{P}(A^c) = O(n^{-5})$. Then we consider the following three events
	\begin{align*}
	&B_1 = \left\{  \frac{\sqrt{L}\left(\bc^\top\tb_M-\bc^\top\tb^*\right)}{\sqrt{\bc^\top\left[\mathcal{P}\nabla^2\cL(\tb^*)\mathcal{P}\right]^{+}\bc}}\leq x\right\}, B_2 = \left\{\frac{\sqrt{L}\left(\bc^\top\ob-\bc^\top\tb^*\right)}{\sqrt{\bc^\top\left[\mathcal{P}\nabla^2\cL(\tb^*)\mathcal{P}\right]^{+}\bc}}\leq x- \Lambda\Gamma\right\}, \\
	&B_3 = \left\{\frac{\sqrt{L}\left(\bc^\top\ob-\bc^\top\tb^*\right)}{\sqrt{\bc^\top\left[\mathcal{P}\nabla^2\cL(\tb^*)\mathcal{P}\right]^{+}\bc}}\leq x+\Lambda\Gamma\right\}.
	\end{align*}
	Then we have
	\begin{align}
	\left|\mathbb{P}\left(\frac{\sqrt{L}\left(\bc^\top\tb_M-\bc^\top\tb^*\right)}{\sqrt{\bc^\top\left[\mathcal{P}\nabla^2\cL(\tb^*)\mathcal{P}\right]^{+}\bc}}\leq x\right)-\mathbb{P}\left(\mathcal{N}(0,1)\leq x\right)\right|  = &\left|\mathbb{P}(B_1\cap A)+\mathbb{P}(B_1\cap A^c)-\mathbb{P}\left(\mathcal{N}(0,1)\leq x\right)\right| \\
	\lesssim &\left|\mathbb{P}(B_1\cap A)-\mathbb{P}\left(\mathcal{N}(0,1)\leq x\right)\right|+\frac{1}{n^5}.\label{eq27}
	\end{align}
	On the other hand, for $B_1\cap A$ we have
	\begin{align*}
	B_2\cap A \subset B_1\cap A \subset B_3\cap A.
	\end{align*}
	As a result, we know that
	\begin{align}
	\mathbb{P}(B_1\cap A)&\leq \mathbb{P}\left(B_3\cap A\right) \leq \mathbb{P}(B_3).\label{eq28}
	\end{align}
	By Eq.~\eqref{eq26} we have
	\begin{align*}
	\left|\mathbb{P}(B_3)- \mathbb{P}\left(\mathcal{N}(0,1)\leq x+\Lambda \Gamma\right)\right|\lesssim\frac{\kappa_1}{\sqrt{npL}}+\frac{1}{n^{10}}.
	\end{align*}
	On the other hand, we have
	\begin{align*}
	\left|\mathbb{P}\left(\mathcal{N}(0,1)\leq x+\Lambda \Gamma\right)-\mathbb{P}(\mathcal{N}(0,1)\leq x)\right|\leq \Lambda\Gamma.
	\end{align*}
	Therefore, we have
	\begin{align}
	\left|\mathbb{P}(B_3)-\mathbb{P}\left(\mathcal{N}(0,1)\leq x\right)\right|\lesssim \frac{\kappa_1}{\sqrt{npL}}+\Gamma+\frac{1}{n^{10}}.\label{eq29}
	\end{align}
	For $B_1\cap A$ we also have
	\begin{align}
	\mathbb{P}(B_1\cap A)\geq \mathbb{P}(B_2\cap A).\label{eq30}
	\end{align}
	By definition we have
	\begin{align*}
	\left|\mathbb{P}(B_2\cap A)-\mathbb{P}(B_2)\right|\leq \mathbb{P}(A^c)\lesssim \frac{1}{n^5}.
	\end{align*}
	By Eq.~\eqref{eq26} we have
	\begin{align*}
	\left|\mathbb{P}(B_2)- \mathbb{P}\left(\mathcal{N}(0,1)\leq x-\Lambda \Gamma\right)\right|\lesssim\frac{\kappa_1}{\sqrt{npL}}+\frac{1}{n^{10}}.
	\end{align*}
	On the other hand, we have
	\begin{align*}
	\left|\mathbb{P}\left(\mathcal{N}(0,1)\leq x-\Lambda \Gamma\right)-\mathbb{P}(\mathcal{N}(0,1)\leq x)\right|\leq \Lambda \Gamma.
	\end{align*}
	Therefore, we have
	\begin{align}
	\left|\mathbb{P}(B_2\cap A)-\mathbb{P}\left(\mathcal{N}(0,1)\leq x\right)\right|\lesssim \frac{\kappa_1}{\sqrt{npL}}+ \Gamma+\frac{1}{n^{5}}.\label{eq31}
	\end{align}
	Combine Eq.~\eqref{eq27}, Eq.~\eqref{eq28} and Eq.~\eqref{eq30} we know that
	\begin{align*}
	&\left|\mathbb{P}\left(\frac{\sqrt{L}\left(\bc^\top\tb_M-\bc^\top\tb^*\right)}{\sqrt{\bc^\top\left[\mathcal{P}\nabla^2\cL(\tb^*)\mathcal{P}\right]^{+}\bc}}\leq x\right)-\mathbb{P}\left(\mathcal{N}(0,1)\leq x\right)\right|\\ 
	\lesssim& \max\{\left|\mathbb{P}(B_3)-\mathbb{P}\left(\mathcal{N}(0,1)\leq x\right)\right|, \left|\mathbb{P}(B_2\cap A)-\mathbb{P}\left(\mathcal{N}(0,1)\leq x\right)\right|\}+\frac{1}{n^5}.
	\end{align*}
	Therefore, by Eq.~\eqref{eq29}, Eq.~\eqref{eq31} we have
	\begin{align*}
	&\left|\mathbb{P}\left(\frac{\sqrt{L}\left(\bc^\top\tb_M-\bc^\top\tb^*\right)}{\sqrt{\bc^\top\left[\mathcal{P}\nabla^2\cL(\tb^*)\mathcal{P}\right]^{+}\bc}}\leq x\right)-\mathbb{P}\left(\mathcal{N}(0,1)\leq x\right)\right| \\ 
	\lesssim &\frac{\kappa_1}{\sqrt{npL}}+\Gamma+\frac{1}{n^{5}}.
	\end{align*}
	Since the above inequality holds for every $x\in \mathbb{R}$, we prove the desired result.

	Thus, we finally conclude our proof of Theorem \ref{betainference}.

\end{proof}

\subsection{Proof of Corollary \ref{alphainference}}
\begin{proof}[Proof of Corollary \ref{alphainference}]
Let $\left\{\tilde{\boldsymbol{e}}_i\right\}_{i=1}^{n+d}$ be canonical basis vectors of $\mathbb{R}^{n+d}$. By taking $\boldsymbol{c} = \tilde{\boldsymbol{e}}_k$ for $k\in [n]$ in Theorem \ref{betainference}, we only have to show that $\frac{\Vert\boldsymbol{c}_{1:n}\Vert_1}{\Vert \bc\Vert_2}\lesssim 1$, and this is also equivalent to $\Vert \bc\Vert_2\gtrsim 1$. By the definition of $\mathcal{P}$ we have
\begin{align*}
    \Vert \bc\Vert_2^2 &= \left(1-\left(\boldsymbol{Z}\left(\boldsymbol{Z}^\top\boldsymbol{Z}\right)^{-1}\boldsymbol{Z}^\top\right)_{k,k}\right)^2 + \sum_{i\in [n],i\neq k} \left(\boldsymbol{Z}\left(\boldsymbol{Z}^\top\boldsymbol{Z}\right)^{-1}\boldsymbol{Z}^\top\right)_{i,k}^2 \\
    &\geq 1-2\left(\boldsymbol{Z}\left(\boldsymbol{Z}^\top\boldsymbol{Z}\right)^{-1}\boldsymbol{Z}^\top\right)_{k,k} \geq 1-2 \left\Vert \boldsymbol{Z}\left(\boldsymbol{Z}^\top\boldsymbol{Z}\right)^{-1}\boldsymbol{Z}^\top\right\Vert_{2,\infty} \\
    &\geq 1-2c_0\sqrt{\frac{d+1}{n}}\geq 0.1.
\end{align*}
As a result, the first part of Corollary \ref{alphainference} is proved by Theorem \ref{betainference}.

For simplicity we denote by $\displaystyle \Gamma = \kappa_1^6\frac{(d+1)\log n}{\sqrt{np L}}+\kappa_1^4\sqrt{\frac{(d+1)\log n}{np}}\left(\sqrt{d+1}+\frac{\log n}{\sqrt{np}}\right)$. To begin with, for fixed $x$ we have
\begin{align*}
    &\left|\mathbb{P}\left(\frac{\sqrt{L}\left(\oa_k-\alpha^*_k\right)}{\sqrt{\left(\left[\mathcal{P}\nabla^2\cL(\tb^*)\mathcal{P}\right]^{+}\right)_{k,k}}}\leq x\right)-\mathbb{P}\left(\mathcal{N}(0,1)\leq x\right)\right| \\
    =& \left|\mathbb{E}_\mathcal{G}\left[\mathbb{P}\left(\frac{\sqrt{L}\left(\oa_k-\alpha^*_k\right)}{\sqrt{\left(\left[\mathcal{P}\nabla^2\cL(\tb^*)\mathcal{P}\right]^{+}\right)_{k,k}}}\leq x\bigg|\mathcal{G}\right)\right]-\mathbb{P}\left(\mathcal{N}(0,1)\leq x\right)\right| \\
    \leq&\mathbb{E}_\mathcal{G}\left|\mathbb{P}\left(\frac{\sqrt{L}\left(\oa_k-\alpha^*_k\right)}{\sqrt{\left(\left[\mathcal{P}\nabla^2\cL(\tb^*)\mathcal{P}\right]^{+}\right)_{k,k}}}\leq x\bigg|\mathcal{G}\right)-\mathbb{P}\left(\mathcal{N}(0,1)\leq x\right)\right| \\
    \lesssim &\frac{\kappa_1}{\sqrt{npL}}+\frac{1}{n^{10}}.
\end{align*}
As a result, we have
\begin{align}
    \sup_{x\in\mathbb{R}}\left|\mathbb{P}\left(\frac{\sqrt{L}\left(\oa_k-\alpha^*_k\right)}{\sqrt{\left(\left[\mathcal{P}\nabla^2\cL(\tb^*)\mathcal{P}\right]^{+}\right)_{k,k}}}\leq x\right)-\mathbb{P}\left(\mathcal{N}(0,1)\leq x\right)\right|\lesssim\frac{\kappa_1}{\sqrt{npL}}+\frac{1}{n^{10}}.\label{eq19}
\end{align}
Consider event $\displaystyle A = \left\{\left|\frac{\sqrt{L}\left(\wh\alpha_{M,k}-\oa_k\right)}{\sqrt{\left(\left[\mathcal{P}\nabla^2\cL(\tb^*)\mathcal{P}\right]^{+}\right)_{k,k}}}\right|\leq \Lambda \Gamma\right\}$, where $\Lambda>0$ is some constant such that $\mathbb{P}(A^c) = O(n^{-5})$. Then we consider the following three events
\begin{align*}
    B_1 &= \left\{  \frac{\sqrt{L}\left(\wh\alpha_{M,k}-\alpha^*_k\right)}{\sqrt{\left(\left[\mathcal{P}\nabla^2\cL(\tb^*)\mathcal{P}\right]^{+}\right)_{k,k}}}\leq x\right\}, B_2 = \left\{\frac{\sqrt{L}\left(\oa_k-\alpha^*_k\right)}{\sqrt{\left(\left[\mathcal{P}\nabla^2\cL(\tb^*)\mathcal{P}\right]^{+}\right)_{k,k}}}\leq x-\Lambda \Gamma\right\}, \\
    B_3 &= \left\{\frac{\sqrt{L}\left(\oa_k-\alpha^*_k\right)}{\sqrt{\left(\left[\mathcal{P}\nabla^2\cL(\tb^*)\mathcal{P}\right]^{+}\right)_{k,k}}}\leq x+\Lambda \Gamma\right\}.
\end{align*}
Then we have
\begin{align}
    \left|\mathbb{P}\left(\frac{\sqrt{L}\left(\wh\alpha_{M,k}-\alpha^*_k\right)}{\sqrt{\left(\left[\mathcal{P}\nabla^2\cL(\tb^*)\mathcal{P}\right]^{+}\right)_{k,k}}}\leq x\right)-\mathbb{P}\left(\mathcal{N}(0,1)\leq x\right)\right|  = &\left|\mathbb{P}(B_1\cap A)+\mathbb{P}(B_1\cap A^c)-\mathbb{P}\left(\mathcal{N}(0,1)\leq x\right)\right| \\
    \lesssim &\left|\mathbb{P}(B_1\cap A)-\mathbb{P}\left(\mathcal{N}(0,1)\leq x\right)\right|+\frac{1}{n^5}.\label{eq21}
\end{align}
On the other hand, for $B_1\cap A$ we have
\begin{align*}
    B_2\cap A \subset B_1\cap A \subset B_3\cap A.
\end{align*}
As a result, we know that
\begin{align}
    \mathbb{P}(B_1\cap A)&\leq \mathbb{P}\left(B_3\cap A\right) \leq \mathbb{P}(B_3).\label{eq22}
\end{align}
By Eq.~\eqref{eq19} we have
\begin{align*}
    \left|\mathbb{P}(B_3)- \mathbb{P}\left(\mathcal{N}(0,1)\leq x+\Lambda \Gamma\right)\right|\lesssim\frac{\kappa_1}{\sqrt{npL}}+\frac{1}{n^{10}}.
\end{align*}
On the other hand, we have
\begin{align*}
    \left|\mathbb{P}\left(\mathcal{N}(0,1)\leq x+\Lambda \Gamma\right)-\mathbb{P}(\mathcal{N}(0,1)\leq x)\right|\leq \Lambda \Gamma.
\end{align*}
Therefore, we have
\begin{align}
    \left|\mathbb{P}(B_3)-\mathbb{P}\left(\mathcal{N}(0,1)\leq x\right)\right|\lesssim \frac{\kappa_1}{\sqrt{npL}}+\Gamma+\frac{1}{n^{10}}.\label{eq24}
\end{align}
For $B_1\cap A$ we also have
\begin{align}
    \mathbb{P}(B_1\cap A)\geq \mathbb{P}(B_2\cap A).\label{eq23}
\end{align}
By definition we have
\begin{align*}
    \left|\mathbb{P}(B_2\cap A)-\mathbb{P}(B_2)\right|\leq \mathbb{P}(A^c)\lesssim \frac{1}{n^5}.
\end{align*}
By Eq.~\eqref{eq19} we have
\begin{align*}
    \left|\mathbb{P}(B_2)- \mathbb{P}\left(\mathcal{N}(0,1)\leq x-\Lambda \Gamma\right)\right|\lesssim\frac{\kappa_1}{\sqrt{npL}}+\frac{1}{n^{10}}.
\end{align*}
On the other hand, we have
\begin{align*}
    \left|\mathbb{P}\left(\mathcal{N}(0,1)\leq x-\Lambda \Gamma\right)-\mathbb{P}(\mathcal{N}(0,1)\leq x)\right|\leq \Lambda \Gamma.
\end{align*}
Therefore, we have
\begin{align}
    \left|\mathbb{P}(B_2\cap A)-\mathbb{P}\left(\mathcal{N}(0,1)\leq x\right)\right|\lesssim \frac{\kappa_1}{\sqrt{npL}}+\Gamma+\frac{1}{n^{5}}.\label{eq25}
\end{align}
Combine Eq.~\eqref{eq21}, Eq.~\eqref{eq22} and Eq.~\eqref{eq23} we know that
\begin{align*}
    &\left|\mathbb{P}\left(\frac{\sqrt{L}\left(\wh\alpha_{M,k}-\alpha^*_k\right)}{\sqrt{\left(\left[\mathcal{P}\nabla^2\cL(\tb^*)\mathcal{P}\right]^{+}\right)_{k,k}}}\leq x\right)-\mathbb{P}\left(\mathcal{N}(0,1)\leq x\right)\right| \\ \lesssim& \max\{\left|\mathbb{P}(B_3)-\mathbb{P}\left(\mathcal{N}(0,1)\leq x\right)\right|, \left|\mathbb{P}(B_2\cap A)-\mathbb{P}\left(\mathcal{N}(0,1)\leq x\right)\right|\}+\frac{1}{n^5}.
\end{align*}
Therefore, by Eq.~\eqref{eq24}, Eq.~\eqref{eq25} we have
\begin{align*}
    &\left|\mathbb{P}\left(\frac{\sqrt{L}\left(\wh\alpha_{M,k}-\alpha^*_k\right)}{\sqrt{\left(\left[\mathcal{P}\nabla^2\cL(\tb^*)\mathcal{P}\right]^{+}\right)_{k,k}}}\leq x\right)-\mathbb{P}\left(\mathcal{N}(0,1)\leq x\right)\right| \\ \lesssim &\frac{\kappa_1}{\sqrt{npL}}+\Gamma+\frac{1}{n^{5}}.
\end{align*}
Since the above inequality holds for every $x\in \mathbb{R}$, we prove the desired result.

\end{proof}

\subsection{Proof of Corollary \ref{cor:inferenceb}}
The proof of Corollary \ref{cor:inferenceb} except the refined bound Eq.~\eqref{refinedbetabound} follows directly from the proof of Theorem \ref{betainference}. Therefore, we omit the details and only prove Eq.~\eqref{refinedbetabound} here. We prove the following lemma first.

\begin{lemma}\label{varianceconcentration}
	Consider some fixed constants $a_{i,j}^{(l)}$ for $ (i,j)\in \mathcal{E}, l\in [L]$, and random variable 
	\begin{align}
		X = \sum_{l=1}^L \sum_{(i,j)\in \mathcal{E}} a_{i,j}^{(l)}\left\{y_{j,i}^{(l)}-\frac{e^{\tx_i^\top \tb^*}}{e^{\tx_i^\top \tb^*}+e^{\tx_j^\top \tb^*}}\right\}. \label{standardX}
	\end{align}
	Conditioned on the comparison graph $\mathcal{G}$, with probability exceeding $1-O(n^{-11})$ we have
	\begin{align*}
		|X|\lesssim \sqrt{\kappa_1\text{Var}[X\mid \cG]\cdot\log n}.
	\end{align*}
\end{lemma}
\begin{proof}
	Let $X_{i,j}^{(l)} = a_{i,j}^{(l)}\left(y_{j,i}^{(l)}-e^{\tx_i^\top \tb^*}/\left(e^{\tx_i^\top \tb^*}+e^{\tx_j^\top \tb^*}\right)\right)$. Then we know that $\mathbb{E}X_{i,j}^{(l)} = 0$ and $|X_{i,j}^{(l)}|\leq |a_{i,j}^{(l)}|$. As a result, by Hoeffding inequality, with probability at least $1-O(n^{-11})$, we have
	\begin{align}
		\left|\sum_{l=1}^L\sum_{(i,j)\in\mathcal{E}, i>j}X_{i,j}^{(l)}\right|&\lesssim \sqrt{\log n\cdot\sum_{l=1}^L\sum_{(i,j)\in\mathcal{E}, i>j}\left(a_{i,j}^{(l)}\right)^2}.\label{eq36}
	\end{align}
	On the other hand, since $y_{i,j}^{(l)}$ are independent random variables, we know that
	\begin{align*}
		\text{Var}[X\mid \cG] &= \sum_{l=1}^L\sum_{(i,j)\in\mathcal{E}, i>j}\text{Var}[X_{i,j}^{(l)}\mid\cG]=\sum_{l=1}^L\sum_{(i,j)\in\mathcal{E}, i>j}\left(a_{i,j}^{(l)}\right)^2 \phi'(\tx_i^\top\tb^* - \tx_j^\top\tb^*) \\
		&\gtrsim \frac{1}{\kappa_1}\sum_{l=1}^L\sum_{(i,j)\in\mathcal{E}, i>j}\left(a_{i,j}^{(l)}\right)^2.
	\end{align*}
	As a result, we have $\sum_{l=1}^L\sum_{(i,j)\in\mathcal{E}, i>j}\left(a_{i,j}^{(l)}\right)^2\lesssim \kappa_1 \text{Var}[X\mid \cG]$. Therefore, by Eq.~\eqref{eq36} we know that
	\begin{align*}
		|X| = \left|\sum_{l=1}^L\sum_{(i,j)\in\mathcal{E}, i>j}X_{i,j}^{(l)}\right|&\lesssim \sqrt{\kappa_1\text{Var}[X\mid \cG]\log n}
	\end{align*}
	with probability exceeding $1-O(n^{-11})$.
\end{proof}
\begin{proof}[Proof of Eq.~\eqref{refinedbetabound}]
	Conditioned on the comparison graph $\mathcal{G}$, the entries of $\ob$ can be written as the form Eq.~\eqref{standardX}. By Lemma \ref{varianceconcentration} and union bound, conditioned on the comparison graph $\mathcal{G}$, we know that
	\begin{align}
		\left\|\ob_{n+1:n+d}-\bb^*\right\|_2 = \sqrt{\sum_{j=1}^d\left(\overline{\beta}_{n+j}-\beta_j^*\right)^2}&\lesssim \sqrt{\kappa_1\log n\cdot\sum_{j=1}^d\text{Var}\left[\overline{\beta}_{n+j}-\beta_j^*\mid \mathcal{G}\right]} \nonumber \\
		&=\sqrt{\kappa_1\log n\cdot\textbf{tr}\left[\text{Var}\left[\;\ob_{n+1:n+d}-\bb^*\mid \mathcal{G}\right]\right]}    \label{eq37}
	\end{align}
with probability at least $1-O(n^{-10})$. By Proposition \ref{pro2} we know that
\begin{align*}
	\text{Var}\left[\;\ob_{n+1:n+d}-\bb^*\mid \mathcal{G}\right] = \frac{1}{L}\left(\left[\mathcal{P}\nabla^2\cL(\tb^*)\mathcal{P}\right]^{+}\right)_{n+1:n+d,n+1:n+d}.
\end{align*}
By the definition of $\Theta$ we know that
\begin{align*}
	\lambda_{\text{max}}\left(\left(\left[\mathcal{P}\nabla^2\cL(\tb^*)\mathcal{P}\right]^{+}\right)_{n+1:n+d,n+1:n+d}\right)&\leq\lambda_{\text{max}}\left(\left[\mathcal{P}\nabla^2\cL(\tb^*)\mathcal{P}\right]^{+}\right)\\
	&\leq \frac{1}{\lambda_{\text{min},\perp}(\nabla^2\cL(\tb^*))}\lesssim\frac{\kappa_1}{np}
\end{align*}
under event $\mathcal{A}_2$. As a result, we have
\begin{align}
	\textbf{tr}\left[\text{Var}\left[\;\ob_{n+1:n+d}-\bb^*\mid \mathcal{G}\right]\right]&\leq \frac{1}{L}d	\lambda_{\text{max}}\left(\left(\left[\mathcal{P}\nabla^2\cL(\tb^*)\mathcal{P}\right]^{+}\right)_{n+1:n+d,n+1:n+d}\right) \\
	&\lesssim \frac{\kappa_1(d+1)}{npL}   \label{eq38}
\end{align}
with probability at least $1-O(n^{-10})$. Combine Eq.~\eqref{eq37} and Eq.~\eqref{eq38} we know that
\begin{align*}
		\left\|\ob_{n+1:n+d}-\bb^*\right\|_2 \lesssim \kappa_1\sqrt{\frac{(d+1)\log n}{npL}}
\end{align*}
with probability at least $1-O(n^{-10})$. On the other hand, by Theorem \ref{inferenceresidual}, we have
\begin{align*}
	\left\|\ob_{n+1:n+d}-\wh\bb_M\right\|_2\leq\left\|\Delta\tb\right\|_2\lesssim \kappa_1^4 \frac{(d+1)^{0.5}\log n}{\sqrt{n}pL}
\end{align*}
with probability at least $1-O(n^{-5})$. We know that
\begin{align*}
	\left\|\wh\bb_M-\bb^*\right\|_2\leq 	\left\|\ob_{n+1:n+d}-\bb^*\right\|_2+	\left\|\ob_{n+1:n+d}-\wh\bb_M\right\|_2\lesssim \sqrt{\frac{d+1}{n}}\max\left\{\kappa_1^4\frac{\log n}{pL},\kappa_1\sqrt{\frac{\log n}{pL}}\right\}
\end{align*}
with probability at least $1-O(n^{-5})$.
\end{proof}

\newpage 
\section{Proof of Auxiliary Lemmas in Section \ref{EstimationOutline}}\label{ProofDetails}
In this section, we prove detailed proof of aforementioned building blocks.

\subsection{Proof of Proposition \ref{prop_iden}}
\begin{proof}
	Assume that there are two parameter vectors $\tb = (\ba^\top, \bb^\top)^\top\in \Theta$ and $\tb' = (\ba'^{\top}, \bb'^{\top})^\top\in \Theta$ such that 
	\begin{align*}
		\frac{e^{\alpha_j^*+\bx_j^\top \bb^*}}{e^{\alpha_i^*+\bx_i^\top \bb^*}+e^{\alpha_j^*+\bx_j^\top \bb^*}}  = \frac{e^{\alpha_j'^{*}+\bx_j^\top \bb'^{*}}}{e^{\alpha_i'^*+\bx_i^\top \bb'^*}+e^{\alpha_j'^*+\bx_j^\top \bb'^*}} ,\quad \forall 1\leq i\neq j\leq n.
	\end{align*}
	Since $e^b/(e^a+e^b) = 1/(e^{a-b}+1)$, we know that 
	\begin{align*}
		\alpha_i^*+\bx_i^\top\bb^* - (\alpha_j^*+\bx_j^\top\bb^*) = \alpha_i'^*+\bx_i^\top\bb'^* - (\alpha_j'^*+\bx_j^\top\bb'^*),\quad \forall 1\leq i\neq j\leq n.
	\end{align*}
	Let $c =\alpha_1^*+\bx_1^\top\bb^* -(\alpha_1'^*+\bx_1^\top\bb'^*) $ and $\boldsymbol{1}_n$ be a $n$-dimensional vertoc whose entries are all one, then we know that 
	\begin{align*}
		\ba^* +\bX\bb^* = \ba'^*+\bX\bb'^*+c\boldsymbol{1}_n.
	\end{align*}
	Since $\bar{\bX}^\top\ba^* = \bar{\bX}^\top\ba'^*=\boldsymbol{0}$, we know that
	\begin{align*}
		\left\|\ba^*-\ba'^*\right\|_2^2 = (\ba^*-\ba'^*)^\top(\ba^*-\ba'^*) =  (\ba^*-\ba'^*)^\top(\bX\bb'^*+c\boldsymbol{1}_n-\bX\bb^*) = 0.
	\end{align*}
	This implies $\ba^* = \ba'^*$. Therefore, we have $\bX\bb^* = \bX\bb'^*+c\boldsymbol{1}_n$. This is equivalent to
	\begin{align*}
		\bar{\bX} \begin{bmatrix}
			c\\
			\bb^*-\bb'^*
		\end{bmatrix} = \boldsymbol{0}.
	\end{align*}
	Since $\bar{\bX}$ has full column rank, we know that $[c, (\bb^*-\bb'^*)^\top] = \boldsymbol{0}$. As a result, we must have $\tb = \tb'$.
	
\end{proof}

\subsection{Proof of Lemma \ref{gradient} and Its Corollary\label{proofgradient}}
In this subsection, we provide the proof of Lemma \ref{gradient} in the sense that we provide an upper bound for the gradient vector in $\ell_2$-norm.
\begin{proof}
The gradient is calculated as
\begin{align*}
    \nabla \mathcal{L}_\lambda(\tb^*) = \lambda \tb^*+\frac{1}{L}\sum_{(i,j)\in\mathcal{E}, i>j}\sum_{l=1}^L\underbrace{\left\{-y_{j,i}^{(l)}+\frac{e^{\tx_i^\top \tb^*}}{e^{\tx_i^\top \tb^*}+e^{\tx_j^\top \tb^*}}\right\}(\tx_i-\tx_j)}_{:=z_{i,j}^{(l)}}.
\end{align*}
Since $\bE [z_{i,j}^{(l)}] = 0, \Vert z_{i,j}^{(l)}\Vert_2\leq\Vert \tx_i-\tx_j\Vert_2\leq \sqrt{6}$, we have 
\begin{align*}
    \bE[z_{i,j}^{(l)}z_{i,j}^{(l)\top} ] =  \text{Var}[y_{j,i}^{(l)}]&(\tx_i-\tx_j)(\tx_i-\tx_j)^\top \prec(\tx_i-\tx_j)(\tx_i-\tx_j)^\top\\
    \text{and} &\qquad \bE[z_{i,j}^{(l)\top} z_{i,j}^{(l)} ]\leq 6.
\end{align*}
as long as $n/d$ is large enough. Thus, with high probability (with respect to the randomness of $\mathcal{G}$), we have
\begin{align*}
    \left\|\sum_{(i, j) \in \mathcal{E}, i>j} \sum_{l=1}^{L} \mathbb{E}\left[z_{i, j}^{(l)} z_{i, j}^{(l) \top}\right]\right\| \leq L\left\|\sum_{(i, j) \in \mathcal{E}, i>j}\left(\tx_{i}-\tx_{j}\right)\left(\tx_{i}-\tx_{j}\right)^{\top}\right\|=L\left\|\boldsymbol{L}_{\mathcal{G}}\right\| \lesssim L n p 
\end{align*}
and 
\begin{align*}
    \left|\sum_{(i, j) \in \mathcal{E}, i>j} \sum_{l=1}^{L} \mathbb{E}\left[\boldsymbol{z}_{i, j}^{(l) \top} \boldsymbol{z}_{i, j}^{(l)}\right]\right| \leq 6 L\left|\sum_{(i, j) \in \mathcal{E}, i>j} 1\right| \lesssim Ln^{2} p.
\end{align*}
Let $V:=\frac{1}{L^{2}} \max \left\{\left\|\sum_{(i, j) \in \mathcal{E}} \sum_{l=1}^{L} \mathbb{E}\left[\boldsymbol{z}_{i, j}^{(l)} \boldsymbol{z}_{i, j}^{(l) \top}\right]\right\|,\left|\sum_{(i, j) \in \mathcal{E}} \sum_{l=1}^{L} \mathbb{E}\left[\boldsymbol{z}_{i, j}^{(l) \top} \boldsymbol{z}_{i, j}^{(l)}\right]\right|\right\}$ and $B:=\max_{i,j,l}\Vert z_{i,j}^{(l)}\Vert/L$. By matrix Bernstein inequality \citep{tropp2015introduction}, we have
\begin{align*}
    \left\|\nabla \mathcal{L}_{\lambda}(\tb^*)-\mathbb{E}\left[\nabla \mathcal{L}_{\lambda}(\tb^*) \mid \mathcal{G}\right]\right\|_{2} \lesssim \sqrt{V \log (n+d+1)}+B \log (n+d+1) \lesssim \sqrt{\frac{n^{2} p \log n}{L}}+\frac{\log n}{L}
\end{align*}
with probability exceeding $1-O(n^{-11})$ as long as $d< n$ and $npL\gtrsim \log n$. On the other hand, we have % check in detail
\begin{align*}
    \left\|\mathbb{E}\left[\nabla \mathcal{L}_{\lambda}(\tb^*) \mid \mathcal{G}\right]\right\|_{2} = \lambda\left\|\tb^*\right\|_{2}\lesssim \frac{1}{\kappa_3\sqrt{d+1}}\sqrt{\frac{np\log n}{L}}\kappa_3\sqrt{(d+1)n}\lesssim \sqrt{\frac{n^2 p\log n}{L}}.
\end{align*}
To summarize, we have 
\begin{align*}
    \left\|\nabla \mathcal{L}_{\lambda}(\tb^*)\right\|_2\lesssim\sqrt{\frac{n^2 p\log n}{L}}
\end{align*}
with probability exceeding $1-O(n^{-11})$.
\end{proof}

Once Lemma \ref{gradient} is established, we have the following lemma which can be viewed as a direct corollary of Lemma \ref{gradient}. 
\begin{lemma}\label{gradientball}
   With $\lambda$ given by Theorem \ref{mainthm}, the following event 
\begin{align*}
    \mathcal{A}_3 = \left\{\left\|\nabla \mathcal{L}_{\lambda}\left(\tb\right)\right\|_{2} \leq C_0\sqrt{\frac{n^{2} p \log n}{L}}+ \left(\lambda+\frac{1}{2}c_1pn\right)r,\; \forall \tb \text{ s.t. }\left\Vert\tb-\tb^*\right\Vert_2\leq r \right\}
\end{align*}
is contained by the event $\mathcal{A}_1\cap\mathcal{A}_2$. As a result, $\mathcal{A}_3$ happens with probability exceeding $1-O(n^{-11})$.
\end{lemma}
\begin{proof}
By the fundamental theorem of calculus we know that
\begin{align*}
    \nabla\cL_\lambda(\tb) = \nabla\cL_\lambda(\tb^*)+\int_0^1\nabla^2\cL_\lambda(\tb(\tau))(\tb-\tb^*)d\tau,
\end{align*}
where $\tb(\tau) = \tb^*+\tau(\tb-\tb^*)$. So for all $\tb$ such that $\Vert\tb-\tb^*\Vert_2\leq r$, we have
\begin{align*}
     \left\Vert\nabla\cL_\lambda(\tb)\right\Vert_2 &\leq  \left\Vert\nabla\cL_\lambda(\tb^*)\right\Vert_2+\int_0^1\left\Vert\nabla^2\cL_\lambda(\tb(\tau))(\tb-\tb^*)\right\Vert_2 d\tau \\
     &\leq \left\Vert\nabla\cL_\lambda(\tb^*)\right\Vert_2+\left\Vert\tb-\tb^*\right\Vert_2\int_0^1\left\Vert\nabla^2\cL_\lambda(\tb(\tau))\right\Vert d\tau  \\
     &\leq \left\Vert\nabla\cL_\lambda(\tb^*)\right\Vert_2+r\int_0^1\left\Vert\nabla^2\cL_\lambda(\tb(\tau))\right\Vert d\tau .
\end{align*}
Under event $\mathcal{A}_1$, we have
\begin{align*}
    \left\|\nabla \mathcal{L}_{\lambda}\left(\tb^*\right)\right\|_{2} \leq C_0\sqrt{\frac{n^{2} p \log n}{L}}.
\end{align*}
And, under event $\mathcal{A}_2$, we have $\Vert \nabla^2\cL_\lambda(\tb(\tau))\Vert\leq \lambda+c_1pn/2$. As a result, $\mathcal{A}_3$ is contained by $\mathcal{A}_1\cap\mathcal{A}_2$.
\end{proof}

\subsection{Proof of Lemma \ref{eigenlem}}\label{proofeigenlemma}
In this subsection, we prove Lemma \ref{eigenlem} by demonstrating
\begin{align*}
\mathcal{A}_2 = \left\{\frac{1}{2}c_2pn\leq \lambda_{\text{min}, \perp}(\LG)\leq \Vert \LG\Vert\leq 2c_1pn \right\}
\end{align*}
holds with high probability.
\begin{proof}
Let $\bO$ be any $r \times (n+d)$ matrix with orthonormal rows such that the row space of $\bO$ is $\Theta$, where $r$ is the dimension of $\Theta$. Then, it holds that
\begin{align*}
    \Vert \LG\Vert &= \Vert \bO\LG\bO^\top\Vert;   \qquad 
    \lambda_{\text{min},\perp}(\LG) = \lambda_{\text{min}}(\bO\LG\bO^\top).
\end{align*}
Let $\bX_{i,j} = \bO(\tx_i-\tx_j)(\tx_i-\tx_j)^\top\bO^\top\textbf{1}((i,j)\in\mathcal{E})$ for $i>j$. Then we have $\displaystyle\bO\LG\bO^\top = \sum_{i>j}\bX_{i,j}$ and
\begin{align*}
    \bX_{i,j}\succeq \boldsymbol{0},\quad \Vert\bX_{i,j}\Vert\leq \Vert \tx_i-\tx_j\Vert_2^2\leq 6,
\end{align*}
as long as $n/d$ is large enough.  Furthermore, 
\begin{align*}
    \lambda_{\text{min}}\left(\mathbb{E}\sum_{i>j}\bX_{i,j}\right) &= \lambda_{\text{min}}\left(p\bO\boldsymbol{\Sigma}\bO^\top\right) = p\lambda_{\text{min}, \perp}(\boldsymbol{\Sigma}) \geq c_2pn; \\
    \lambda_{\text{max}}\left(\mathbb{E}\sum_{i>j}\bX_{i,j}\right) &= \lambda_{\text{max}}\left(p\bO\boldsymbol{\Sigma}\bO^\top\right) = p\lambda_{\text{max}}(\boldsymbol{\Sigma}) \leq c_1pn.
\end{align*}
By the matrix Chernoff inequality \citep{tropp2012user}, we have
\begin{align*}
    \mathbb{P}\left(\lambda_{\text{min}}\left(\sum_{i>j}\bX_{i,j}\right)\leq \frac{1}{2}\lambda_{\text{min}}\left(\mathbb{E}\sum_{i>j}\bX_{i,j}\right)\right)&\leq (n+d)\cdot 0.8^{c_2pn/6}; \\
    \mathbb{P}\left(\lambda_{\text{max}}\left(\sum_{i>j}\bX_{i,j}\right)\geq \frac{3}{2}\lambda_{\text{max}}\left(\mathbb{E}\sum_{i>j}\bX_{i,j}\right)\right)&\leq (n+d)\cdot 0.8^{c_2pn/6}. 
\end{align*}
As a result, if $pn >c_p\log n$ for some $c_p>0$, we have
\begin{align*}
    \mathbb{P}(\mathcal{A}_2)\geq\mathbb{P}\left(\frac{1}{2}c_2pn\leq \lambda_{\text{min}}\left(\sum_{i>j}\bX_{i,j}\right)\leq \lambda_{\text{max}}\left(\sum_{i>j}\bX_{i,j}\right)
    \leq \frac{3}{2}c_1pn \right)\geq 1-O(n^{-11}).
\end{align*}
This concludes the proof of Lemma \ref{eigenlem}.
\end{proof}

\subsection{Proof of Lemma \ref{lb} \label{proofLemmaA4}}

In this subsection, we provide the proof of Lemma \ref{lb} by providing a lower bound for $\lambda_{\text{min},\perp}(\nabla^2 \mathcal{L}_\lambda(\tb)).$
\begin{proof}
For pair $(i,j)$, without loss of generality we assume $\tx_i^\top\tb\leq \tx_j^\top\tb$,  we then obtain
\begin{align*}
    \frac{e^{\tx_i^\top \tb}e^{\tx_j^\top \tb}}{\left(e^{\tx_i^\top \tb}+e^{\tx_j^\top \tb}\right)^2} = \frac{e^{\tx_i^\top \tb-\tx_j^\top \tb}}{\left(1+e^{\tx_i^\top \tb-\tx_j^\top \tb}\right)^2} = \frac{e^{-|\tx_i^\top \tb-\tx_j^\top \tb|}}{\left(1+e^{-|\tx_i^\top \tb-\tx_j^\top \tb|}\right)^2}\geq \frac{1}{4}e^{-|\tx_i^\top \tb-\tx_j^\top \tb|}.
\end{align*}
One the other hand, it holds that
\begin{align*}
    |\tx_i^\top \tb-\tx_j^\top \tb|&\leq |\tx_i^\top \tb^*-\tx_j^\top \tb^*|+|\tx_i^\top \tb^*-\tx_i^\top \tb|+|\tx_j^\top \tb^*-\tx_j^\top \tb|. \\
    &\leq \log(\kappa_1) + |\alpha_i-\alpha_i^*|+|\bx_i^\top\bb^*-\bx_i^\top\bb|+|\alpha_j-\alpha_j^*|+|\bx_j^\top\bb^*-\bx_j^\top\bb| \\
    &\leq \log(\kappa_1) + 2\|\ba-\ba^*\|_{\infty}+2\sqrt{\frac{c_3(d+1)}{n}}\|\bb^*-\bb\|_2 \\
    &\leq \log(\kappa_1) + 2C_1+2\sqrt{\frac{c_3(d+1)}{n}}C_2.
\end{align*}
Therefore, we obtain
\begin{align*}
    \frac{e^{\tx_i^\top \tb}e^{\tx_j^\top \tb}}{\left(e^{\tx_i^\top \tb}+e^{\tx_j^\top \tb}\right)^2}\geq \frac{1}{4\kappa_1e^C},
\end{align*}
where $\displaystyle C = 2C_1+2\sqrt{\frac{c_3(d+1)}{n}}C_2$.
As a result, for the Hessian $\nabla^2 \mathcal{L}_\lambda(\tb)$ we have
\begin{align*}
    \lambda_{\text{min},\perp}(\nabla^2 \mathcal{L}_\lambda(\tb))\geq \lambda+ \frac{1}{4\kappa_1e^C} \lambda_{\text{min},\perp}(\LG)\geq \lambda+\frac{c_2pn}{8\kappa_1e^C}.
\end{align*}
This completes the proof of Lemma \ref{lb}.
\end{proof}

\subsection{Proof of Lemma \ref{lem2}}
\begin{proof}
Since $\mathcal{L}_\lambda(\cdot)$ is $\lambda$-strongly convex and $\displaystyle\lambda+\frac{1}{2}c_1np$-smooth on the event $\mathcal{A}_2$, we know that 
\begin{align*}
    \left\Vert \tb^t-\eta\nabla_\lambda(\tb_\lambda)-\tb^*\right\Vert_2\leq \rho \left\Vert \tb^t - \tb_\lambda\right\Vert_2.
\end{align*}
As a result, when event $\mathcal{A}_2$ happens, we have
\begin{align*}
   \left\Vert\tb^{t+1}-\tb^*\right\Vert_2 =  \left\Vert \mathcal{P}\left(\tb^t-\eta\nabla_\lambda(\tb_\lambda)-\tb^*\right)\right\Vert_2\leq \left\Vert \tb^t-\eta\nabla_\lambda(\tb_\lambda)-\tb^*\right\Vert_2\leq \rho \left\Vert \tb^t - \tb_\lambda\right\Vert_2.
\end{align*}
Therefore, under event $\mathcal{A}_2$, we have
\begin{align*}
    \left\Vert\tb^{t}-\tb^*\right\Vert_2\leq \rho^t\left\Vert\tb^{0}-\tb^*\right\Vert_2.
\end{align*}
\end{proof}

\subsection{Proof of Lemma \ref{lem3} \label{proofLemmaA6}}
In this section, we prove $\tb_{\lambda}$ is not far from the initial point $\tb^*$.
\begin{proof}
Since $\tb_\lambda$ is the minimizer, we have that $\mathcal{L}_\lambda(\tb^*)\geq \mathcal{L}_\lambda(\tb_\lambda)$. By the mean value theorem, for some $\tb'$ between $\tb^*$ and $\tb_\lambda$, we have
\begin{align*}
    \mathcal{L}_\lambda(\tb_\lambda) 
     = \mathcal{L}_\lambda(\tb^*)+ \nabla \mathcal{L}_\lambda(\tb^*)^\top (\tb_\lambda-\tb^*)+\frac{1}{2}(\tb_\lambda-\tb^*)^\top \nabla^2 \mathcal{L}_\lambda(\tb')(\tb_\lambda-\tb^*).
\end{align*}
As a result, we have
\begin{align*}
    \mathcal{L}_\lambda(\tb^*)&\geq \mathcal{L}_\lambda(\tb^*)+ \nabla \mathcal{L}_\lambda(\tb^*)^\top (\tb_\lambda-\tb^*)+\frac{1}{2}(\tb_\lambda-\tb^*)^\top \nabla^2 \mathcal{L}_\lambda(\tb^*)(\tb_\lambda-\tb^*) \\
    & \geq \mathcal{L}_\lambda(\tb^*)+ \nabla \mathcal{L}_\lambda(\tb^*)^\top (\tb_\lambda-\tb^*)+\frac{\lambda}{2}\left\Vert\tb_\lambda-\tb^*\right\Vert_2^2.
\end{align*}
Therefore, we get
\begin{align*}
    \frac{\lambda}{2}\left\Vert\tb_\lambda-\tb^*\right\Vert_2^2 & \leq -\nabla \mathcal{L}_\lambda(\tb^*)^\top (\tb_\lambda-\tb^*)\\
    &\leq \left\Vert \nabla \mathcal{L}_\lambda(\tb^*)\right\Vert_2\left\Vert \tb_\lambda-\tb^*\right\Vert_2.
\end{align*}
As a result, on event $\mathcal{A}_1$ we have
\begin{align*}
    \left\Vert \tb_\lambda-\tb^*\right\Vert_2\leq \frac{2\left\Vert \nabla \mathcal{L}_\lambda(\tb^*)\right\Vert_2}{\lambda}\leq \frac{2C_0}{c_\lambda}\max\left\{\frac{\kappa_2}{\kappa_1}, \kappa_3\sqrt{d+1}\right\}\sqrt{n}.
\end{align*}
We conclude the proof of Lemma \ref{lem3}.
\end{proof}

\subsection{Proof of Lemma \ref{lem5}}\label{prooflem5}
\begin{proof}
Combine Lemma \ref{lem2} and Lemma \ref{lem3} we have
\begin{align*}
\left\Vert \tb^T-\tb_\lambda\right\Vert_2&\leq \rho^T\left\Vert \tb^0-\tb_\lambda\right\Vert_2 \\
&\leq \left(1-\frac{2\lambda}{2\lambda+c_1 np}\right)^{n^5}\frac{2C_0}{c_\lambda}\max\left\{\frac{\kappa_2}{\kappa_1}, \kappa_3\sqrt{d+1}\right\}\sqrt{n}\\
&\leq \frac{2C_0}{c_\lambda}\max\left\{\frac{\kappa_2}{\kappa_1}, \kappa_3\sqrt{d+1}\right\}\sqrt{n}\exp\left(-\frac{2\lambda n^5}{2\lambda+c_1 np}\right) \\
&\leq C_7 \kappa_1\sqrt{\frac{(d+1)\log n}{npL}}
\end{align*}
for $L \leq c_4\cdot n^{c_5}$ and $n$ which is large enough.
\end{proof}

\subsection{Proof of Lemma \ref{induction1}} \label{rsc}
\begin{proof}
By definition we know that 
%\begin{align*}
%\left\Vert\tb^{t+1}-\tb^*\right\Vert_2 = \left\Vert\mathcal{P}\left(\tb^t-\eta\nabla\mathcal{L}_\lambda(\tb^t)-\tb^*\right)\right\Vert_2\leq \left\Vert \tb^t-\eta\nabla\mathcal{L}_\lambda(\tb^t)-\tb^*\right\Vert_2.
%\end{align*}
\begin{align*}
\tb^{t+1}-\tb^* = \mathcal{P}\left(\tb^t-\eta\nabla\mathcal{L}_\lambda(\tb^t)\right)-\tb^* = \mathcal{P}\left(\tb^t-\eta\nabla\mathcal{L}_\lambda(\tb^t)-\tb^*\right).
\end{align*}
Consider $\tb(\tau) = \tb^*+\tau\left(\tb^t-\tb^*\right)$. By the fundamental theorem of calculus we have
\begin{align*}
     \tb^t-\eta\nabla\mathcal{L}_\lambda(\tb^t)-\tb^* 
    & = \tb^t-\eta\nabla\mathcal{L}_\lambda(\tb^t)-\left[\tb^*-\eta\nabla\mathcal{L}_\lambda(\tb^*)\right]-\eta\nabla\mathcal{L}_\lambda(\tb^*) \\
    & = \left\{\boldsymbol{I}_{n+d}-\eta\int^1_0\nabla^2\cL_\lambda(\tb(\tau))d\tau\right\}\left(\tb^t-\tb^*\right)-\eta \nabla\cL_\lambda(\tb^*).
\end{align*}
Let $npL$ be large enough such that 
\begin{align*}
    2C_6\kappa_1^2\sqrt{\frac{(d+1)\log n}{npL}}\leq 0.1,\quad 2C_3\kappa_1\sqrt{c_3}\sqrt{\frac{(d+1)\log n}{npL}}\leq 0.1.
\end{align*}
In addition, by the assumption of induction, we have 
\begin{align*}
    \Vert \ba(\tau)-\ba^*\Vert_\infty\leq 0.05,\quad \Vert \bb(\tau)-\bb^*\Vert_2\leq 0.05\sqrt{\frac{n}{c_3(d+1)}}. 
\end{align*}
Then by Lemma \ref{lb}, we have
\begin{align*}
    \lambda_{\text{min}, \perp}\left(\nabla^2\cL_\lambda(\tb(\tau))\right)\geq \lambda+\frac{c_2pn}{8\kappa_1e^{0.2}}\geq \lambda+\frac{c_2pn}{10\kappa_1}, \quad \forall 0\leq \tau\leq 1.
\end{align*}
On the other hand, by Lemma \ref{ub}, we have
\begin{align*}
    \lambda_{\text{max}}\left(\nabla^2\cL_\lambda(\tb(\tau))\right)\leq \lambda+\frac{1}{2}c_1pn.
\end{align*}
Let $\displaystyle\boldsymbol{A} = \int^1_0\nabla^2\cL_\lambda(\tb(\tau))d\tau$, then we have
\begin{align*}
    \lambda+\frac{c_2pn}{10\kappa_1}\leq \lambda_{\text{min}, \perp}(A)\leq \lambda_{\text{max}}(A)\leq \lambda+\frac{1}{2}c_1pn.
\end{align*}
Since $\tb^t-\tb^*\in \Theta$, it holds that
\begin{align*}
    \left\Vert \mathcal{P}(\boldsymbol{I}_{n+d}-\eta\boldsymbol{A})(\tb^t-\tb^*)\right\Vert_2&\leq \max\{|1-\eta\lambda_{\text{min}, \perp}(A), 1-\eta\lambda_{\text{max}}(A)|\}\left\Vert \tb^t-\tb^*\right\Vert_2 \\
    &\leq \left(1-\frac{c_2}{20\kappa_1}\eta pn\right)\left\Vert\tb^t-\tb^*\right\Vert_2.
\end{align*}
Therefore, on the event $\mathcal{A}_1$, we have
\begin{align*}
    \left\Vert\tb^{t+1}-\tb^*\right\Vert_2&\leq \left\Vert \mathcal{P}\left(\boldsymbol{I}_{n+d}-\eta\boldsymbol{A}\right)(\tb^t-\tb^*)-\eta \mathcal{P}\nabla\cL_\lambda(\tb^*)\right\Vert_2\\
    &\leq \left\Vert \mathcal{P}\left(\boldsymbol{I}_{n+d}-\eta\boldsymbol{A}\right)(\tb^t-\tb^*)\right\Vert_2+\eta \left\Vert\nabla\cL_\lambda(\tb^*)\right\Vert_2 \\
    &\leq \left(1-\frac{c_2}{20\kappa_1}\eta pn\right)\left\Vert\tb^t-\tb^*\right\Vert_2 +C_0\eta \sqrt{\frac{n^2p\log n}{L}} \\
    &\leq \left(1-\frac{c_2}{20\kappa_1}\eta pn\right)C_3 \kappa_1\sqrt{\frac{\log n}{pL}} +C_0\eta \sqrt{\frac{n^2p\log n}{L}} \\
    &\leq C_3 \kappa_1\sqrt{\frac{\log n}{pL}},
\end{align*}
as long as $\displaystyle C_3\geq \frac{20C_0}{c_2}$.
\end{proof}

\subsection{Proof of Lemma \ref{induction3}}\label{prooflem10}
\begin{proof}
For any $m\in[n]$, by definition we have
%\begin{align*}
%    \left\Vert\tb^{t+1}-\tb^{t+1,(m)}\right\Vert_2 &= \left\Vert \mathcal{P}\left(\tb^t-\eta\nabla\cL_\lambda(\tb^t)\right)-\mathcal{P}\left(\tb^{t,(m)}-\eta\nabla\cL_\lambda^{(m)}(\tb^{t,(m)})\right)\right\Vert_2 \\
%    &\leq \left\Vert \tb^t-\eta\nabla\mathcal{L}_\lambda(\tb^t)-\left[\tb^{t,(m)}-\eta\nabla\mathcal{L}_\lambda^{(m)}(\tb^{t,(m)})\right]\right\Vert_2.
%\end{align*}
\begin{align*}
    \tb^{t+1}-\tb^{t+1,(m)} &=  \mathcal{P}\left(\tb^t-\eta\nabla\mathcal{L}_\lambda(\tb^t)-\left[\tb^{t,(m)}-\eta\nabla\mathcal{L}_\lambda^{(m)}(\tb^{t,(m)})\right]\right).
\end{align*}
We consider $\tb(\tau) = \tb^{t,(m)}+\tau\left(\tb^t-\tb^{t,(m)}\right)$. By the fundamental theorem of calculus we have
\begin{align}
    & \tb^t-\eta\nabla\mathcal{L}_\lambda(\tb^t)-\left[\tb^{t,(m)}-\eta\nabla\mathcal{L}_\lambda^{(m)}(\tb^{t,(m)})\right] \\
    =& \tb^t-\eta\nabla\mathcal{L}_\lambda(\tb^t)-\left[\tb^{t,(m)}-\eta\nabla\mathcal{L}_\lambda(\tb^{t,(m)})\right]-\eta \left(\nabla\mathcal{L}_\lambda(\tb^{t,(m)})-\nabla\mathcal{L}_\lambda^{(m)}(\tb^{t,(m)})\right) \\
    =& \left\{\boldsymbol{I}_{n+d}-\eta\int^1_0\nabla^2\cL_\lambda(\tb(\tau))d\tau\right\}\left(\tb^{t}-\tb^{t,(m)}\right)-\eta \left(\nabla\mathcal{L}_\lambda(\tb^{t,(m)})-\nabla\mathcal{L}_\lambda^{(m)}(\tb^{t,(m)})\right).\label{induction3decompose}
\end{align}
From \eqref{inductionA}$\sim$ \eqref{inductionD} we know that
\begin{align*}
    \Vert \ba^{t,(m)}-\ba^*\Vert_{\infty}&\leq \Vert \ba^{t}-\ba^*\Vert_{\infty}+\max_{1\leq m\leq n}\left\Vert \tb^{t,(m)}-\tb^t\right\Vert_2 \leq (C_4+C_6)\kappa_1^2\sqrt{\frac{(d+1)\log n}{npL}}; \\
    \Vert \bb^{t,(m)}-\bb^*\Vert_{2}&\leq \left\Vert \tb^{t}-\tb^*\right\Vert_{2}+\max_{1\leq m\leq n}\left\Vert \tb^{t,(m)}-\tb^t\right\Vert_2 \leq  (C_3+C_4)\kappa_1\sqrt{\frac{\log n}{pL}};\\
    \Vert \ba^{t}-\ba^*\Vert_{\infty}&\leq C_6\kappa_1^2\sqrt{\frac{(d+1)\log n}{npL}}; \\
    \Vert \bb^{t}-\bb^*\Vert_{2}&\leq \left\Vert \tb^{t}-\tb^*\right\Vert_{2}\leq C_3\kappa_1\sqrt{\frac{\log n}{pL}}.
\end{align*}
Consider $npL$ which is large enough such that
\begin{align*}
     2(C_4+C_6)\kappa_1^2\sqrt{\frac{(d+1)\log n}{npL}},\; 2(C_3+C_4)\kappa_1\sqrt{c_3}\sqrt{\frac{(d+1)\log n}{npL}}\leq 0.1.
\end{align*}
Then we also have
\begin{align*}
    2C_6\kappa_1^2\sqrt{\frac{(d+1)\log n}{npL}},\;2C_3\kappa_1\sqrt{c_3}\sqrt{\frac{(d+1)\log n}{npL}}\leq 0.1.
\end{align*}
Use the same approach derived in \S\ref{rsc}, we have
\begin{align}
    \left\|\mathcal{P}\left\{\boldsymbol{I}_{n+d}-\eta\int^1_0\nabla^2\cL_\lambda(\tb(\tau))d\tau\right\}\left(\tb^{t}-\tb^{t,(m)}\right)\right\|_2\leq \left(1-\frac{c_2}{20\kappa_1}\eta pn\right)\left\Vert\tb^t-\tb^{t,(m)}\right\Vert_2,\label{induction3sc}
\end{align}
as long as $\displaystyle 0<\eta\leq \frac{2}{2\lambda+c_1np}$.

On the other hand, since $\left\|\mathcal{P}(  \nabla\mathcal{L}_\lambda(\tb^{t,(m)})-\nabla\mathcal{L}_\lambda^{(m)}(\tb^{t,(m)}))\right\|_2\leq \left\|  \nabla\mathcal{L}_\lambda(\tb^{t,(m)})-\nabla\mathcal{L}_\lambda^{(m)}(\tb^{t,(m)})\right\|_2$, it remains to bound $\left\|  \nabla\mathcal{L}_\lambda(\tb^{t,(m)})-\nabla\mathcal{L}_\lambda^{(m)}(\tb^{t,(m)})\right\|_2.$ By definition, we have
\begin{align*}
    &\nabla\mathcal{L}_\lambda(\tb^{t,(m)})-\nabla\mathcal{L}_\lambda^{(m)}(\tb^{t,(m)}) \\
    =& \sum_{i\neq m}\left\{\left(-y_{m,i}+\frac{e^{  \tx_i^\top\tb^{t,(m)}}}{e^{\tx_i^\top\tb^{t,(m)}}+e^{\tx_m^\top\tb^{t,(m)}}} \right)\textbf{1}((i,m)\in\mathcal{E})-p\left(-y_{m,i}^*+\frac{e^{  \tx_i^\top\tb^{t,(m)}}}{e^{\tx_i^\top\tb^{t,(m)}}+e^{\tx_m^\top\tb^{t,(m)}}}\right)\right\}(\tx_i-\tx_m)\\
    =&\underbrace{\sum_{i\neq m}\left\{\left(-\frac{e^{  \tx_i^\top\tb^{*}}}{e^{\tx_i^\top\tb^*}+e^{\tx_m^\top\tb^*}}+\frac{e^{  \tx_i^\top\tb^{t,(m)}}}{e^{\tx_i^\top\tb^{t,(m)}}+e^{\tx_m^\top\tb^{t,(m)}}} \right)\left(\textbf{1}((i,m)\in\mathcal{E})-p\right)\right\}(\tx_i-\tx_m)}_{:=\boldsymbol{u}^m} \\
    &+\underbrace{\frac{1}{L}\sum_{(i,m)\in\mathcal{E}}\sum^L_{l=1}\left(-y^{(l)}_{m,i}+\frac{e^{  \tx_i^\top\tb^*}}{e^{\tx_i^\top\tb^*}+e^{\tx_m^\top\tb^*}}\right)(\tx_i-\tx_m)}_{:=\boldsymbol{v}^m}.
\end{align*}

By definition, we also have
\begin{align*}
    v_{j}^{m}= \begin{cases}\frac{1}{L} \sum_{l=1}^{L}\left(-y_{m, j}^{(l)}+\frac{e^{  \tx_j^\top\tb^*}}{e^{\tx_j^\top\tb^*}+e^{\tx_m^\top\tb^*}}\right), & \text { if }(j, m) \in \mathcal{E} \\ \frac{1}{L} \sum_{i:(i, m) \in \mathcal{E}} \sum_{l=1}^{L}\left(y_{m, i}^{(l)}-\frac{e^{  \tx_i^\top\tb^*}}{e^{\tx_i^\top\tb^*}+e^{\tx_m^\top\tb^*}}\right), & \text { if } j=m ; \\ 
    \frac{1}{L} \sum_{i:(i, m) \in \mathcal{E}} \sum_{l=1}^{L}\left(-y_{m, i}^{(l)}+\frac{e^{  \tx_i^\top\tb^*}}{e^{\tx_i^\top\tb^*}+e^{\tx_m^\top\tb^*}}\right)((\tx_i)_j-(\tx_m)_j), & \text { if } j>n ; \\ 
    0, & \text { else. }\end{cases}
\end{align*}
Consider random variable $M = |\left\{i:(i,m)\in \mathcal{E}\right\}|$. By Chernoff bound \citep{tropp2012user}, we know that
\begin{align*}
\mathbb{P}(M\geq 2pn)\leq (e/4)^{pn}\leq O(n^{-11}),
\end{align*}
as long as $np>c_p\log n$ for some $c_p>0$. As long as $\Vert \bx_i-\bx_m\Vert_2\leq 2\sqrt{c_3(d+1)/n}\leq 1$, we have $|(\tx_i)_j-(\tx_m)_j)|\leq 1$ for $j>n$. Since $\left|-y^{(l)}_{m,i}+\frac{e^{  \tx_i^\top\tb^*}}{e^{\tx_i^\top\tb^*}+e^{\tx_m^\top\tb^*}}\right|\leq 1$, by Hoeffding's inequality and union bound, we get
\begin{align*}
    |v_j^m|&\lesssim \sqrt{\frac{M\log n}{L}}, \text{ if } j=m \text{ or } j>n; \\
    |v_j^m|&\lesssim \sqrt{\frac{\log n}{L}}, \text{ if } (j,m)\in\mathcal{E}.
\end{align*}
with probability exceeding $1-O(n^{-11})$ conditioning on $\mathcal{E}$ as long as $d< n$. On the other hand, since $M\leq 2pn$ with probability exceeding $1-O(n^{-11})$, we have
\begin{align*}
    \Vert\boldsymbol{v}^m\Vert_2^2\lesssim (d+1)\frac{2pn\log n}{L} +2pn\frac{\log n}{L}\lesssim \frac{pn(d+1)\log n}{L}
\end{align*}
with probability exceeding $1-O(n^{-11})$. %That is to say,
%\begin{align*}
%    \Vert\boldsymbol{v}^m\Vert_2\lesssim \sqrt{ \frac{pn(d+1)\log n}{L}}.
%\end{align*}

On the other hand, for $\boldsymbol{u}^m$ we have
\begin{align*}
    u_{j}^{m}= \begin{cases}\xi_j(1-p), & \text { if }(j, m) \in \mathcal{E} \\ 
    -\sum_{i:(i, m) \in \mathcal{E}} \xi_i\left(\textbf{1}((i,m)\in\mathcal{E})-p\right), & \text { if } j=m ; \\ 
     \sum_{i:(i, m) \in \mathcal{E}} \xi_i\left(\textbf{1}((i,m)\in\mathcal{E})-p\right)((\tx_i)_j-(\tx_m)_j), & \text { if } j>n ; \\ 
    -\xi_j p, & \text { else, }\end{cases}
\end{align*}
where $$\displaystyle \xi_j = -\frac{e^{  \tx_j^\top\tb^{*}}}{e^{\tx_j^\top\tb^*}+e^{\tx_m^\top\tb^*}}+\frac{e^{  \tx_j^\top\tb^{t,(m)}}}{e^{\tx_j^\top\tb^{t,(m)}}+e^{\tx_m^\top\tb^{t,(m)}}} = -\frac{1}{1+e^{\tx_m^\top\tb^*-\tx_j^\top\tb^*}}+\frac{1}{1+e^{\tx_m^\top\tb^{t,(m)}-\tx_j^\top\tb^{t,(m)}}}.$$
Consider $\displaystyle g(x) = \frac{1}{1+e^x}$. Since $\displaystyle\left|g'(x)\right|\leq 1$, we have that
\begin{align*}
    |\xi_j| &= \left|g(\tx_m^\top\tb^{t,(m)}-\tx_j^\top\tb^{t,(m)})-g(\tx_m^\top\tb^*-\tx_j^\top\tb^*)\right| \\
    &\leq \left|(\tx_m^\top\tb^{t,(m)}-\tx_j^\top\tb^{t,(m)})-(\tx_m^\top\tb^*-\tx_j^\top\tb^*)\right|\\
    &\leq \left|\tx_m^\top\tb^{t,(m)}-\tx_m^\top\tb^*\right|+\left|\tx_j^\top\tb^{t,(m)}-\tx_j^\top\tb^*\right| \\
    &\leq \left|\alpha_m^{t,(m)}-\alpha_m^*\right|+\left|\bx_m^\top\bb^{t,(m)}-\bx_m^\top\bb^*\right|+\left|\alpha_j^{t,(m)}-\alpha_j^*\right|+\left|\bx_j^\top\bb^{t,(m)}-\bx_j^\top\bb^*\right| \\
    &\leq 2 \left\|\ba^{t,(m)}-\ba^*\right\|_\infty + 2\sqrt{c_3(d+1)/n}\left\|\bb^{t,(m)}-\bb^*\right\|_2 \\
    &\leq \left[2(C_4+C_6)\kappa_1^2+2(C_3+C_4)\kappa_1\sqrt{c_3}\right]\sqrt{\frac{(d+1)\log n}{npL}} := \tC_1\sqrt{\frac{(d+1)\log n}{npL}}.
\end{align*}
By Bernstein inequality we know that
\begin{align*}
    |u_j^m|&\lesssim \sqrt{\left(p\sum^n_{i=1}\xi_i^2\right)\log n}+\max_{1\leq i\leq n}|\xi_i|\log n\\
    &\leq \left(\sqrt{np\log n}+\log n\right)\tC_1\sqrt{\frac{(d+1)\log n}{npL}},\text{ if } j=m \text{ or } j>n.
\end{align*}
As a result, for $\boldsymbol{u}^m$ we have
\begin{align*}
    \Vert \boldsymbol{u}^m\Vert_2^2 &= (u^m_m)^2 +\sum_{j>n}(u^m_j)^2+\sum_{j:(j,m)\in \mathcal{E}}(u^m_j)^2 +\sum_{j:(j,m)\notin \mathcal{E},j\neq m, j\leq n}(u^m_j)^2 \\
    &\lesssim (d+1)\left(\sqrt{np\log n}+\log n\right)^2\tC_1^2\frac{(d+1)\log n}{npL}+np\tC_1^2\frac{(d+1)\log n}{npL}+p^2n\tC_1^2\frac{(d+1)\log n}{npL} \\
    &\lesssim pn(d+1)\log n \tC_1^2\frac{(d+1)\log n}{npL}.
\end{align*}
In summary,  there exists constants $D_1, D_2$ which are independent of $C_i, i\geq 0$ such that
\begin{align}
     \Vert\boldsymbol{v}^m\Vert_2\leq D_1\sqrt{\frac{pn(d+1)\log n}{L}},\quad \Vert\boldsymbol{u}^m\Vert_2\leq D_2\tC_1 (d+1)\log n\sqrt{\frac{1}{L}}\label{induction3uv}
\end{align}
with probability exceeding $1-O(n^{-11})$. Combining  Eq.~\eqref{induction3decompose}, Eq.~\eqref{induction3sc} and Eq.~\eqref{induction3uv} we have
\begin{align*}
    \left\Vert \tb^{t+1}-\tb^{t+1,(m)}\right\Vert_2\leq& \left(1-\frac{c_2}{20\kappa_1}\eta pn\right)\left\Vert\tb^t-\tb^{t,(m)}\right\Vert_2   \\
    &+ \eta \left(D_1\sqrt{\frac{pn(d+1)\log n}{L}}+D_2\tC_1 (d+1)\log n\sqrt{\frac{1}{L}}\right) \\
    \leq& \left(1-\frac{c_2}{20\kappa_1}\eta pn\right)C_4\kappa_1\sqrt{\frac{(d+1)\log n}{npL}} \\
    &+ \eta \left(D_1\sqrt{\frac{pn(d+1)\log n}{L}}+D_2\tC_1 (d+1)\log n\sqrt{\frac{1}{L}}\right) \\
    \leq& C_4\kappa_1\sqrt{\frac{(d+1)\log n}{npL}},
\end{align*}
as long as $\displaystyle C_4\geq \frac{40D_1}{c_2}$ and $n$ is large enough such that $\displaystyle C_4\geq \frac{40D_2}{c_2}\tC_1\sqrt{\frac{(d+1)\log n}{np}}$.

\end{proof}

\subsection{Proof of Lemma \ref{induction2}}\label{prooflem9}
\begin{proof}
For $m\in[n]$, we have
\begin{align*}
     \alpha_m^{t+1,(m)}-\alpha_m^{*} &= \left[\mathcal{P}\left(\tb^{t,(m)}-\eta\nabla\cL_\lambda^{(m)}\left(\tb^{t,(m)}\right)-\tb^*\right)\right]_m \\
     & = \alpha_m^{t,(m)}-\eta\left[\mathcal{P}\nabla\cL_\lambda^{(m)}\left(\tb^{t,(m)}\right)\right]_m - \alpha_m^* \\
     & = \underbrace{\alpha_m^{t,(m)}-\eta\left[\nabla\cL_\lambda^{(m)}\left(\tb^{t,(m)}\right)\right]_m-\alpha_m^*}_{\mu_1} + \underbrace{\eta \left[(I-\mathcal{P})\nabla\cL_\lambda^{(m)}\left(\tb^{t,(m)}\right)\right]_m}_{\mu_2}
\end{align*}
For $\mu_2$, we have
\begin{align*}
    \left|\mu_2\right|&\leq \eta\left\Vert(I-\mathcal{P})\nabla\cL_\lambda^{(m)}\left(\tb^{t,(m)}\right)\right\Vert_\infty \\
    &\leq \eta \left\Vert I-\mathcal{P}\right\Vert_{2,\infty}\left\Vert\nabla\cL_\lambda^{(m)}\left(\tb^{t,(m)}\right)\right\Vert_2 \\
    &\leq c_0\eta\sqrt{\frac{d+1}{n}}\left\Vert\nabla\cL_\lambda^{(m)}\left(\tb^{t,(m)}\right)\right\Vert_2.
\end{align*}
By Lemma \ref{gradientball}, with probability at least $1-O(n^{-11})$ we have
\begin{align*}
    \left\Vert\nabla\cL_\lambda^{(m)}\left(\tb^{t,(m)}\right)\right\Vert_2\leq (C_0+c_1C_3+c_1C_4)\kappa_1\sqrt{\frac{n^2p\log n}{L}},
\end{align*}
as long as 
\begin{align*}
    npL\geq \frac{4c_\lambda^2}{c_1^2}\min\left\{\frac{\kappa_1^2}{\kappa_2^2},\frac{1}{(d+1)\kappa_3^2}\right\}\log n.
\end{align*}
In this case, for $\mu_2$ we have
\begin{align}
    \left|\mu_2\right|\leq  c_0(C_0+c_1C_3+c_1C_4)\eta\kappa_1\sqrt{\frac{(d+1)np\log n}{L}}:=\tC_2\eta\kappa_1\sqrt{\frac{(d+1)np\log n}{L}}.\label{mu2}
\end{align}
On the other hand, for $\mu_1$ we have
\begin{align}
    \mu_1 &= \alpha_m^{t,(m)}-\eta\left[\nabla\cL_\lambda^{(m)}\left(\tb^{t,(m)}\right)\right]_m-\alpha_m^* \\
    &= \alpha_m^{t,(m)}-\alpha_m^*-\eta p\sum_{i\neq m}\left\{\frac{e^{\tx_i^\top\tb^*}}{e^{\tx_i^\top\tb^*}+e^{\tx_m^\top\tb^*}}-\frac{e^{\tx_i^\top\tb^{t, (m)}}}{e^{\tx_i^\top\tb^{t, (m)}}+e^{\tx_m^\top\tb^{t, (m)}}}\right\}-\eta\lambda\alpha_m^{t,(m)}. \label{eqa}
\end{align}
%Since
%\begin{align*}
%    \frac{e^{\tx_i^\top\tb^*}}{e^{\tx_i^\top\tb^*}+e^{\tx_m^\top\tb^*}}-\frac{e^{\tx_i^\top\tb^{t, (m)}}}{e^{\tx_i^\top\tb^{t, (m)}}+e^{\tx_m^\top\tb^{t, %(m)}}} = \frac{1}{1+e^{\tx_m^\top\tb^* - \tx_i^\top\tb^*}}-\frac{1}{1+e^{\tx_m^\top\tb^{t, (m)} - \tx_i^\top\tb^{t, (m)}}},
%\end{align*}
Normalizing the numerators below to 1 and by the mean value theorem,  there exists some $c_i$ between $\tx_m^\top\tb^* - \tx_i^\top\tb^*$ and $\tx_m^\top\tb^{t, (m)} - \tx_i^\top\tb^{t, (m)}$ such that
\begin{align}
    \frac{e^{\tx_i^\top\tb^*}}{e^{\tx_i^\top\tb^*}+e^{\tx_m^\top\tb^*}}-\frac{e^{\tx_i^\top\tb^{t, (m)}}}{e^{\tx_i^\top\tb^{t, (m)}}+e^{\tx_m^\top\tb^{t, (m)}}} =& -\frac{e^{c_i}}{(1+e^{c_i})^2}\left[\tx_m^\top\tb^* - \tx_i^\top\tb^*-\tx_m^\top\tb^{t, (m)} + \tx_i^\top\tb^{t, (m)}\right] \nonumber \\
    =&-\frac{e^{c_i}}{(1+e^{c_i})^2}\left[\alpha_m^* - \alpha_i^*-\alpha_m^{t,(m)}+\alpha_i^{t,(m)}\right]  \nonumber \\ &\;-\frac{e^{c_i}}{(1+e^{c_i})^2}\left[\bx_m^\top\bb^* - \bx_i^\top\bb^*-\bx_m^\top\bb^{t, (m)} + \bx_i^\top\bb^{t, (m)}\right].
    \label{eqb}
\end{align}
Combining Eq.~\eqref{eqa} and Eq.~\eqref{eqb}, we have
\begin{align*}
    \mu_1 =& \left(1-\eta p\sum_{i\neq m}\frac{e^{c_i}}{(1+e^{c_i})^2}\right)\left(\alpha_m^{t,(m)}-\alpha_m^{*}\right)\\
    &+\eta p\sum_{i\neq m}\frac{e^{c_i}}{(1+e^{c_i})^2}\left[\alpha_i^{t,(m)}-\alpha_i^{*}+(\bx_m-\bx_i)^\top\left(\bb^*-\bb^{t,(m)}\right)\right]-\eta\lambda\alpha_m^{t,(m)} \\
    =& \left(1-\eta\lambda-\eta p\sum_{i\neq m}\frac{e^{c_i}}{(1+e^{c_i})^2}\right)\left(\alpha_m^{t,(m)}-\alpha_m^{*}\right)\\
    &+\eta p\sum_{i\neq m}\frac{e^{c_i}}{(1+e^{c_i})^2}\left[\alpha_i^{t,(m)}-\alpha_i^{*}+(\bx_m-\bx_i)^\top\left(\bb^*-\bb^{t,(m)}\right)\right]-\eta\lambda\alpha_m^{*}.
\end{align*}
By taking absolute value on both side, we get
\begin{align*}
    |\mu_1|\leq& \left| 1-\eta\lambda-\eta p\sum_{i\neq m}\frac{e^{c_i}}{(1+e^{c_i})^2}\right||\alpha_m^{t,(m)}-\alpha_m^{*}| \\
    &+ \frac{\eta p}{4} \sum_{i\neq m}\left[|\alpha_i^{t,(m)}-\alpha_i^{*}|+\Vert\bx_m-\bx_i\Vert_2\Vert\bb^*-\bb^{t,(m)}\Vert_2\right]+\eta\lambda |\alpha_m^*| \\
    \leq& \left| 1-\eta\lambda-\eta p\sum_{i\neq m}\frac{e^{c_i}}{(1+e^{c_i})^2}\right||\alpha_m^{t,(m)}-\alpha_m^{*}| \\
    &+ \frac{\eta p}{4}\left[\sqrt{n}\|\ba^{t,(m)}-\ba^{*}\|_2+n\cdot 2\sqrt{\frac{c_3(d+1)}{n}}\Vert\bb^*-\bb^{t,(m)}\Vert_2\right]+\eta\lambda \|\ba^*\|_{\infty}\\
    \leq& \left| 1-\eta\lambda-\eta p\sum_{i\neq m}\frac{e^{c_i}}{(1+e^{c_i})^2}\right||\alpha_m^{t,(m)}-\alpha_m^{*}| \\
    &+ \frac{\eta p}{4}\sqrt{n}\left(1+2\sqrt{c_3(d+1)}\right)\left\Vert\tb^*-\tb^{t,(m)}\right\Vert_2+\eta\lambda \|\ba^*\|_{\infty}.
\end{align*}
Since $\displaystyle 1-\eta\lambda-\eta p\sum_{i\neq m}\frac{e^{c_1}}{(1+e^{c_i})^2}\geq 1-\eta\lambda-\eta p\frac{n}{4}\geq 0$, we have
\begin{align*}
    \left| 1-\eta\lambda-\eta p\sum_{i\neq m}\frac{e^{c_i}}{(1+e^{c_i})^2}\right| &= 1-\eta\lambda-\eta p\sum_{i\neq m}\frac{e^{c_i}}{(1+e^{c_i})^2} \\
    &\leq 1-\eta p(n-1)\min_{i\neq m}\frac{e^{c_i}}{(1+e^{c_i})^2}.
\end{align*}
By the defintion of $c_i,$ we have
\begin{align*}
    \max_{i\neq m}|c_i|\leq& \max_{i\neq m} |\tx_m^\top\tb^* - \tx_i^\top\tb^*|+\max_{i\neq m}\left|\tx_m^\top\tb^* - \tx_i^\top\tb^*-\tx_m^\top\tb^{t, (m)} + \tx_i^\top\tb^{t, (m)}\right| \\
    \leq& \log \kappa_1+\max_{i\neq m}|\alpha^*_m-\alpha_i^*-\alpha_m^{t,(m)}+\alpha_i^{t,(m)}|+\max_{i\neq m}\left|(\bx_m-\bx_i)^\top\left(\bb^*-\bb^{t, (m)}\right)\right| \\
    \leq& \log \kappa_1+2\Vert \ba^{t,(m)}-\ba^*\Vert_{\infty}+2\sqrt{\frac{c_3(d+1)}{n}}\Vert \bb^*-\bb^{t, (m)}\Vert_2.
\end{align*}
Consider $npL$ which is large enough such that
\begin{align*}
    2(C_4+C_6)\kappa_1^2\sqrt{\frac{(d+1)\log n}{npL}}\leq 0.1,\quad 2(C_3+C_4)\kappa_1\sqrt{c_3}\sqrt{\frac{(d+1)\log n}{npL}}\leq 0.1.
\end{align*}
Then we have
\begin{align*}
    \max_{i\neq m}|c_i|\leq& \log \kappa_1+2\Vert \ba^{t,(m)}-\ba^*\Vert_{\infty}+2\sqrt{\frac{c_3(d+1)}{n}}\Vert \bb^*-\bb^{t, (m)}\Vert_2 \\
    \leq& \log \kappa_1 + 2\left(\Vert \ba^{t}-\ba^*\Vert_{\infty}+\max_{1\leq m\leq n}\left\Vert \tb^t-\tb^{t,(m)}\right\Vert_2\right) \\
    &+2\sqrt{\frac{c_3(d+1)}{n}}\left(\left\Vert \tb^t-\tb^{*}\right\Vert_2+\max_{1\leq m\leq n}\left\Vert \tb^t-\tb^{t,(m)}\right\Vert_2\right) \\
    \leq &\log \kappa_1+0.2.
\end{align*}
Then it holds that,
\begin{align*}
    \min_{i\neq m}\frac{e^{c_i}}{(1+e^{c_i})^2} = \min_{i\neq m}\frac{e^{-|c_i|}}{(1+e^{-|c_i|})^2}\geq \min_{i\neq m}\frac{e^{-|c_i|}}{4} = \frac{e^{-\max_{i\neq m}|c_i|}}{4}\geq\frac{1}{4\kappa_1e^{0.2}} \geq \frac{1}{5\kappa_1}.
\end{align*}
Using $\displaystyle n-1\geq \frac{n}{2}$ for $n\geq 2$, we have
\begin{align*}
    |\mu_1|\leq& \left(1-\frac{1}{10\kappa_1}\eta pn\right)|\alpha_m^{t,(m)}-\alpha_m^{*}| \\ &+\frac{1+2\sqrt{c_3(d+1)}}{4}\eta p\sqrt{n}\left(\left\Vert \tb^t-\tb^{*}\right\Vert_2+\max_{1\leq m\leq n}\left\Vert \tb^t-\tb^{t,(m)}\right\Vert_2\right) + \eta\lambda \kappa_2 \\
    \leq& \left(1-\frac{1}{10\kappa_1}\eta pn\right)C_5\kappa_1^2\sqrt{\frac{(d+1)\log n}{npL}} \\ &+\frac{1+2\sqrt{c_3(d+1)}}{4}\eta p\sqrt{n}(C_3+C_4)\kappa_1\sqrt{\frac{\log n}{pL}} + \eta\lambda \kappa_2.
\end{align*} % add eqref
Combine this result with Eq.~\eqref{mu2}, we get
\begin{align*}
    \left| \alpha_m^{t+1,(m)} - \alpha_m^*\right| \leq& \left|\mu_1\right|+\left|\mu_2\right|  \\
    \leq& \tC_2\eta\kappa_1\sqrt{\frac{(d+1)np\log n}{L}}+\left(1-\frac{1}{10\kappa_1}\eta pn\right)C_5\kappa_1^2\sqrt{\frac{(d+1)\log n}{npL}} \\
    &+\frac{1+2\sqrt{c_3(d+1)}}{4}\eta p\sqrt{n}(C_3+C_4)\kappa_1\sqrt{\frac{\log n}{pL}} + \eta\lambda \kappa_2 \\
    \leq & C_5\kappa_1^2\sqrt{\frac{(d+1)\log n}{npL}}.
\end{align*}
as long as $C_5\geq 30\tC_2$, $C_5\geq 7.5(1+2\sqrt{c_3})(C_3+C_4)$ and $C_5\geq 30c_\lambda/\sqrt{d+1}$. This concludes our proof for Lemma \ref{induction2}.
\end{proof}

\subsection{Proof of Lemma \ref{induction4}}\label{prooflem11}
\begin{proof}
For any $m\in[n]$, we have
\begin{align*}
    |\alpha_m^{t+1}-\alpha_m^*|&\leq |\alpha_m^{t+1}-\alpha_m^{t+1,(m)}|+|\alpha_m^{t+1,(m)}-\alpha_m^{*}| \\
    &\leq \left\|\tb_m^{t+1}-\tb_m^{t+1,(m)}\right\|_2+|\alpha_m^{t+1,(m)}-\alpha_m^{*}| \\
    &\leq C_4\kappa_1\sqrt{\frac{(d+1)\log n}{npL}}+C_5\kappa_1^2\sqrt{\frac{(d+1)\log n}{npL}} \\
    &\leq (C_4+C_5)\kappa_1^2\sqrt{\frac{(d+1)\log n}{npL}}.
\end{align*}
As a result, we have
\begin{align*}
    \left\|\ba^{t+1}-\ba^*\right\|_\infty\leq C_6\kappa_1^2\sqrt{\frac{(d+1)\log n}{npL}},
\end{align*}
as long as $C_6\geq C_4+C_5$.
\end{proof}

\newpage
\section{Proof of Auxiliary Lemmas in Section \ref{InferenceOutline}}
\subsection{Proof of Lemma \ref{Lilemma} and Two Propositions (Propositions \ref{pro2} and \ref{pro3})}\label{prooflemb1}
\begin{proof}[Proof of Lemma \ref{Lilemma}]
(1) By definition for $i\in [n]$ we have
\begin{align*}
    \left(\nabla\cL(\tb^*)\right)_i &= \sum_{j\neq i, (i,j)\in\mathcal{E}}\left\{-y_{j,i}+\phi(\tx_i^\top\tb^*-\tx_j^\top\tb^*)\right\} \\
    &=\frac{1}{L}\sum_{j\neq i, (i,j)\in\mathcal{E}}\sum_{l=1}^L\left\{-y^{(l)}_{j,i}+\phi(\tx_i^\top\tb^*-\tx_j^\top\tb^*)\right\}.
\end{align*}
Since $\left|-y^{(l)}_{j,i}+\phi(\tx_i^\top\tb^*-\tx_j^\top\tb^*)\right|\leq 1$, by Bernstein inequality we have
\begin{align*}
    \left|\left(\nabla\cL(\tb^*)\right)_i-\mathbb{E}\left[\left(\nabla\cL(\tb^*)\right)_i\bigg|\mathcal{G}\right]\right|&\lesssim\frac{1}{L}\left(\sqrt{\log n \left(\sum_{j\neq i, (i,j)\in\mathcal{E}}1\right)L}+\log n\right) \\
    &\lesssim\sqrt{\frac{np\log n}{L}}
\end{align*}
with probability exceeding $1-O(n^{-10})$, as long as $npL\gtrsim \log n$. On the other hand, since $\mathbb{E}\left[-y^{(l)}_{j,i}+\phi(\tx_i^\top\tb^*-\tx_j^\top\tb^*)\right]=0$, we know that $\mathbb{E}\left[\left(\nabla\cL(\tb^*)\right)_i\bigg|\mathcal{G}\right]=0$. As a result, we have
\begin{align*}
    \left|\left(\nabla\cL(\tb^*)\right)_i\right|\lesssim\sqrt{\frac{np\log n}{L}}.
\end{align*}

(2) By definition we have 
\begin{align*}
    \sum_{j\neq i}\left(\nabla^2\cL(\tb^*)\right)_{i,j}^2 &= \left\Vert\left(\sum_{j\neq i, (i,j)\in \mathcal{E}}\phi'(\tx_i^\top\tb^*-\tx_j^\top\tb^*)\left(\tx_i-\tx_j\right)\right)_{-i} \right\Vert^2_2 \\
    &=\sum_{j\neq i, (i,j)\in \mathcal{E}} 1 +\left\Vert\sum_{j\neq i, (i,j)\in \mathcal{E}}\phi'(\tx_i^\top\tb^*-\tx_j^\top\tb^*)\left(\bx_i-\bx_j\right) \right\Vert^2_2 \\
    &\leq \sum_{j\neq i, (i,j)\in \mathcal{E}} 1 +\left(\sum_{j\neq i, (i,j)\in \mathcal{E}} 1\right)\sum_{j\neq i, (i,j)\in \mathcal{E}} \Vert \bx_i-\bx_j\Vert_2^2 \\
    &\lesssim np+np\cdot dp
\end{align*}
with probability at least $1-O(n^{-10})$. Similarly, we have
\begin{align*}
    \sum_{k>n}\left(\nabla^2\cL(\tb^*)\right)_{i,k}^2 &= \left\Vert\sum_{j\neq i, (i,j)\in \mathcal{E}}\phi'(\tx_i^\top\tb^*-\tx_j^\top\tb^*)\left(\bx_i-\bx_j\right) \right\Vert^2_2 \\
    &\leq \left(\sum_{j\neq i, (i,j)\in \mathcal{E}} 1\right)\sum_{j\neq i, (i,j)\in \mathcal{E}} \Vert \bx_i-\bx_j\Vert_2^2 \\
    &\lesssim np\cdot dp
\end{align*}
and 
\begin{align*}
    \sum_{j\in [n],j\neq i}\left|\left(\nabla^2\cL(\tb^*)\right)_{i,j}\right| = \sum_{j\in [n],j\neq i}\left|\phi'(\tx_i^\top\tb^*-\tx_j^\top\tb^*)\right|\textbf{1}((i,j)\in \mathcal{E})\leq  \sum_{j\in [n],j\neq i}\textbf{1}((i,j)\in \mathcal{E})\lesssim np
\end{align*}
with probability at least $1-O(n^{-10})$.

(3) For $i,j\in [n], i\neq j$, by definition we know that $\displaystyle y_{j,i} = \frac{1}{L}\sum_{l=1}^L y_{j,i}^{(l)}$ is the average of $L$ independent Bernoulli random variables. By Hoeffding's inequality we know that 
\begin{align*}
    \left|y_{j,i}-\mathbb{E}y_{j,i}\right|\lesssim \sqrt{\frac{\log n}{L}}
\end{align*}
with probability at least $1-O(n^{-12})$. As a result, by union bound we know that 
\begin{align*}
    \left|y_{j,i}-\mathbb{E}y_{j,i}\right|\lesssim \sqrt{\frac{\log n}{L}}
\end{align*}
holds for all $i,j\in [n], i\neq j$ with probability at least $1-O(n^{-10})$.
\end{proof}

We also include here two propositions which are also involved in the later proofs.

\begin{proposition}\label{pro2}
	$\tb$ is the solution of the following linear equations
	\begin{align*}
		\left\{ \begin{array}{ll}
			&\mathcal{P}\nabla\cL(\tb^*)+\mathcal{P}\nabla^2\cL(\tb^*)\left(\ob-\tb^*\right) = \boldsymbol{0};  \\
			&\mathcal{P}\ob = \ob.
		\end{array}
		\right.
	\end{align*}
\end{proposition}
Proposition \ref{pro2} follows from Eq.~\eqref{eq11}, which gives the definition of $\ob$.
\begin{proposition}\label{pro3}
	Under event $\mathcal{A}_2$, we have
	\begin{align*}
		\text{Var}\left[\;\ob\mid\mathcal{G}\right] = \frac{1}{L}\left[\mathcal{P}\nabla^2\cL(\tb^*)\mathcal{P}\right]^{+}.
	\end{align*}
\end{proposition}
 We next provide the proof of Proposition \ref{pro3} here.
\begin{proof}[Proof of Proposition \ref{pro3}]
	Since $\mathcal{P}\nabla\cL(\tb^*)+\mathcal{P}\nabla^2\cL(\tb^*)(\ob-\tb^*) = \boldsymbol{0}$, by taking variance (conditioned on $\mathcal{G}$) on the both sides we have
	\begin{align}
		\mathcal{P}\text{Var}\left[\nabla\cL(\tb^*)\mid \mathcal{G}\right]\mathcal{P}+\mathcal{P}\nabla^2\cL(\tb^*)\text{Var}\left[\;\ob\mid\mathcal{G}\right]\nabla^2\cL(\tb^*)\mathcal{P}+2\mathcal{P}\text{Cov}\left(\nabla\cL(\tb^*),\ob\mid \mathcal{G}\right)\nabla^2\cL(\tb^*)\mathcal{P} = \boldsymbol{0}.\label{var0}
	\end{align}
	On the other hand, by considering the covariance (conditioned on $\mathcal{G}$) of $\nabla\cL(\tb^*)$ and $\mathcal{P}\nabla\cL(\tb^*)+\mathcal{P}\nabla^2\cL(\tb^*)(\ob-\tb^*)$ we get
	\begin{align*}
		\text{Var}\left[\nabla\cL(\tb^*)\mid \mathcal{G}\right]\mathcal{P}+\text{Cov}\left(\nabla\cL(\tb^*),\ob\mid \mathcal{G}\right)\nabla^2\cL(\tb^*)\mathcal{P} = \boldsymbol{0}.
	\end{align*}
	As a result, we know that
	\begin{align}
		\mathcal{P}\text{Cov}\left(\nabla\cL(\tb^*),\ob\mid \mathcal{G}\right)\nabla^2\cL(\tb^*)\mathcal{P} = -\mathcal{P}\text{Var}\left[\nabla\cL(\tb^*)\mid \mathcal{G}\right]\mathcal{P}.\label{covnablaLob}
	\end{align}
	Combine Eq.~\eqref{var0} and Eq.~\eqref{covnablaLob} we get
	\begin{align}
		\mathcal{P}\nabla^2\cL(\tb^*)\text{Var}\left[\;\ob\mid\mathcal{G}\right]\nabla^2\cL(\tb^*)\mathcal{P} = \mathcal{P}\text{Var}\left[\nabla\cL(\tb^*)\mid \mathcal{G}\right]\mathcal{P} =\frac{1}{L}\mathcal{P}\nabla^2\cL(\tb^*)\mathcal{P} .\label{varob1}
	\end{align}
	By taking variance on the both sides of $\mathcal{P}\ob = \ob$, we have $\mathcal{P}\text{Var}\left[\;\ob\mid\mathcal{G}\right]\mathcal{P} = \text{Var}\left[\;\ob\mid\mathcal{G}\right]$. This also implies $(\boldsymbol{I}-\mathcal{P})\text{Var}\left[\;\ob\mid\mathcal{G}\right] = 0$. As a result, Eq.~\eqref{varob1} can be also written as
	\begin{align*}
		\mathcal{P}\nabla^2\cL(\tb^*)\mathcal{P}\text{Var}\left[\;\ob\mid\mathcal{G}\right]\mathcal{P}\nabla^2\cL(\tb^*)\mathcal{P} = \frac{1}{L}\mathcal{P}\nabla^2\cL(\tb^*)\mathcal{P}.
	\end{align*}
	This immediately implies $\mathcal{P}\nabla^2\cL(\tb^*)\mathcal{P}(\boldsymbol{I} - \text{Var}\left[\;\ob\mid\mathcal{G}\right]\mathcal{P}\nabla^2\cL(\tb^*)\mathcal{P}) = \boldsymbol{0}$. Under event $\mathcal{A}_2$, we have $\lambda_{\text{min},\perp}(\nabla^2\cL(\tb^*))>0$. As a result, for any $\tb\in\Theta$, we have
	\begin{align*}
		&\lambda_{\text{min},\perp}(\nabla^2\cL(\tb^*))\left\Vert \left(\boldsymbol{I} - \text{Var}\left[\;\ob\mid\mathcal{G}\right]\mathcal{P}\nabla^2\cL(\tb^*)\mathcal{P}\right)\tb\right\Vert_2^2\\ 
		\leq& \tb^\top \left(\boldsymbol{I} - \text{Var}\left[\;\ob\mid\mathcal{G}\right]\mathcal{P}\nabla^2\cL(\tb^*)^\top\mathcal{P}\right)^\top\nabla^2\cL(\tb^*)\left(\boldsymbol{I} - \text{Var}\left[\;\ob\mid\mathcal{G}\right]\mathcal{P}\nabla^2\cL(\tb^*)\mathcal{P}\right)\tb \\
		= & \tb^\top \left(\boldsymbol{I} - \text{Var}\left[\;\ob\mid\mathcal{G}\right]\mathcal{P}\nabla^2\cL(\tb^*)^\top\mathcal{P}\right)^\top\mathcal{P}\nabla^2\cL(\tb^*)\mathcal{P}\left(\boldsymbol{I} - \text{Var}\left[\;\ob\mid\mathcal{G}\right]\mathcal{P}\nabla^2\cL(\tb^*)\mathcal{P}\right)\tb= 0.
	\end{align*}
	As a result, we have $(\boldsymbol{I} - \text{Var}\left[\;\ob\mid\mathcal{G}\right]\mathcal{P}\nabla^2\cL(\tb^*)\mathcal{P})\tb = \boldsymbol{0}$ for all $\tb\in\Theta$. Combine this fact with $(\boldsymbol{I}-\mathcal{P})\text{Var}\left[\;\ob\mid\mathcal{G}\right] = \boldsymbol{0}$, we know that 
	\begin{align*}
		\text{Var}\left[\;\ob\mid\mathcal{G}\right] = \left[\mathcal{P}\nabla^2\cL(\tb^*)\mathcal{P}\right]^{+}.
	\end{align*}
\end{proof}

\subsection{Proof of Theorem \ref{inferencemainthm}}\label{proveinferencemainthm}
\begin{proof}
We know that
\begin{align*}
    \boldsymbol{0} = \mathcal{P}\nabla \cL (\tb_M) &= \mathcal{P}\nabla \cL (\tb^*) +\mathcal{P}\int_{0}^1\nabla^2\cL(\tb^*+t(\tb_M-\tb^*))\left(\tb_M-\tb^*\right)dt \\
    &= \mathcal{P}\nabla \cL (\tb^*) +\mathcal{P}\left\{\int_{0}^1\nabla^2\cL(\tb^*+t(\tb_M-\tb^*))dt\right\}\left(\tb_M-\tb^*\right).
\end{align*}
Let 
\begin{align*}
    \boldsymbol{R} = \left\{\int_{0}^1\nabla^2\cL(\tb^*+t(\tb_M-\tb^*))dt\right\}\left(\tb_M-\tb^*\right)-\nabla^2\cL(\tb^*)\left(\tb_M-\tb^*\right),
\end{align*}
we have
\begin{align}
    \boldsymbol{0} = \mathcal{P}\nabla \cL (\tb^*) +\mathcal{P}\nabla^2 \cL (\tb^*)\left(\tb_M-\tb^*\right)+\mathcal{P}\boldsymbol{R}.\label{eq15}
\end{align}
On the other hand, we know that
\begin{align}
    \boldsymbol{0} = \mathcal{P}\nabla \ocL (\ob) &= \mathcal{P}\nabla \cL (\tb^*) +\mathcal{P}\nabla^2\cL(\tb^*)\left(\ob-\tb^*\right).\label{eq16}
\end{align}
Combine Eq.~\eqref{eq15} and Eq.~\eqref{eq16} we have
\begin{align}
    \mathcal{P}\nabla^2 \cL (\tb^*)\left(\ob-\tb_M\right) = \mathcal{P}\boldsymbol{R}.\label{eq17}
\end{align}
For $\boldsymbol{R}$ we have
\begin{align*}
    \Vert \boldsymbol{R}\Vert_2 &\leq \left\Vert \int_{0}^1\nabla^2\cL(\tb^*+t(\tb_M-\tb^*))-\nabla^2\cL(\tb^*)dt\right\Vert\left\Vert\tb_M-\tb^* \right\Vert_2 \\
    &\leq  \int_{0}^1\left\Vert\nabla^2\cL(\tb^*+t(\tb_M-\tb^*))-\nabla^2\cL(\tb^*)\right\Vert dt\left\Vert\tb_M-\tb^* \right\Vert_2.
\end{align*}
By definition we have
\begin{align*}
    &\left\Vert\nabla^2\cL(\tb^*+t(\tb_M-\tb^*))-\nabla^2\cL(\tb^*)\right\Vert  \\
    =&\left\Vert\sum_{(i,j)\in\mathcal{E},i>j}\left(\phi'(t(\tx_i^\top\tb_M-\tx_j^\top\tb_M)+(1-t)(\tx_i^\top\tb^*-\tx_j^\top\tb^*))-\phi'(\tx_i^\top\tb^*-\tx_j^\top\tb^*)\right)(\tx_i-\tx_j)(\tx_i-\tx_j)^\top\right\Vert \\
    \leq&\left\Vert\sum_{(i,j)\in\mathcal{E},i>j}\left|\phi'(t(\tx_i^\top\tb_M-\tx_j^\top\tb_M)+(1-t)(\tx_i^\top\tb^*-\tx_j^\top\tb^*))-\phi'(\tx_i^\top\tb^*-\tx_j^\top\tb^*)\right|(\tx_i-\tx_j)(\tx_i-\tx_j)^\top\right\Vert \\
    \lesssim&\left\Vert\sum_{(i,j)\in\mathcal{E},i>j}\left|t(\tx_i^\top\tb_M-\tx_j^\top\tb_M)-t(\tx_i^\top\tb^*-\tx_j^\top\tb^*)\right|(\tx_i-\tx_j)(\tx_i-\tx_j)^\top\right\Vert \\
    \lesssim &\left\Vert\sum_{(i,j)\in\mathcal{E},i>j}(\tx_i-\tx_j)(\tx_i-\tx_j)^\top\right\Vert \left\Vert\tb_M-\tb^* \right\Vert_c \lesssim \left\Vert \boldsymbol{L}_\mathcal{G}\right\Vert\left\Vert\tb_M-\tb^* \right\Vert_c,\quad \forall t\in [0,1].
\end{align*}
By Lemma \ref{eigenlem}  we know that 
\begin{align*}
\left\Vert\nabla^2\cL(\tb^*+t(\tb_M-\tb^*))-\nabla^2\cL(\tb^*)\right\Vert	\lesssim np\left\Vert\tb_M-\tb^* \right\Vert_c
\end{align*}
holds uniformly for all $t\in [0,1]$ with probability at least $1-O(n^{-11})$. As a result, we get
\begin{align}
    \Vert \boldsymbol{R}\Vert_2 &\lesssim  \int_{0}^1 np\left\Vert\tb_M-\tb^* \right\Vert_c dt\left\Vert\tb_M-\tb^* \right\Vert_2\lesssim np\left\Vert\tb_M-\tb^* \right\Vert_c\left\Vert\tb_M-\tb^* \right\Vert_2\label{eq18}
\end{align}
with probability exceeding $1-O(n^{-11})$. Since $\ob, \tb_M\in \Theta$, we know that
%\begin{align*}
%    &\lambda_{\text{min}, \perp}\left(\nabla^2\cL(\tb^*)\right)\left\Vert\ob-\tb_M \right\Vert_2^2 = \lambda_{\text{min}, \perp}\left(\nabla^2\cL(\tb^*)\right)\left\Vert\mathcal{P}\left(\ob-\tb_M\right) \right\Vert_2^2 \\
%    \leq & \left(\ob-\tb_M\right)^\top\mathcal{P} \nabla^2\cL(\tb^*)\mathcal{P}\left(\ob-\tb_M\right) = \left(\ob-\tb_M\right)^\top\mathcal{P} \nabla^2\cL(\tb^*)\left(\ob-\tb_M\right) = \left(\ob-\tb_M\right)^\top\mathcal{P}\boldsymbol{R} \\
%    \leq & \left\Vert\ob-\tb_M \right\Vert_2 \left\Vert\mathcal{P}\boldsymbol{R} \right\Vert_2 \leq \left\Vert\ob-\tb_M \right\Vert_2 \left\Vert\boldsymbol{R} \right\Vert_2.
%\end{align*}
\begin{align*}
    &\lambda_{\text{min}, \perp}\left(\nabla^2\cL(\tb^*)\right)\left\Vert\ob-\tb_M \right\Vert_2^2 \leq  \left(\ob-\tb_M\right)^\top \nabla^2\cL(\tb^*)\left(\ob-\tb_M\right) \\
    =& \left(\ob-\tb_M\right)^\top\mathcal{P} \nabla^2\cL(\tb^*)\left(\ob-\tb_M\right)  = \left(\ob-\tb_M\right)^\top\mathcal{P} \boldsymbol{R}\\
    \leq &\left\Vert\ob-\tb_M \right\Vert_2 \left\Vert\mathcal{P}\boldsymbol{R} \right\Vert_2\leq \left\Vert\ob-\tb_M \right\Vert_2 \left\Vert\boldsymbol{R} \right\Vert_2.
\end{align*}
As a result, we get 
\begin{align}
     \left\Vert \ob-\tb_M\right\Vert_2\leq \frac{\Vert \boldsymbol{R}\Vert_2}{\lambda_{\text{min}, \perp}\left(\nabla^2\cL(\tb^*)\right)}.\label{newton}
\end{align}
By Lemma \ref{lb} we know that
\begin{align}
    \lambda_{\text{min}, \perp}\left(\nabla^2\cL(\tb^*)\right)\gtrsim\frac{pn}{\kappa_1}.\label{eigmin}
\end{align}
with probability exceeding $1-O(n^{-11})$. Therefore, combine Eq.~\eqref{eq18}, Eq.~\eqref{newton} and Eq.~\eqref{eigmin} we get
\begin{align*}
    \left\Vert \ob-\tb_M\right\Vert_2\leq \frac{\Vert \boldsymbol{R}\Vert_2}{\lambda_{\text{min}, \perp}\left(\nabla^2\cL(\tb^*)\right)}\lesssim\kappa_1^4\frac{(d+1)^{0.5}\log n}{\sqrt{n}pL}
\end{align*}
with probability exceeding $1-O(n^{-6})$.
\end{proof}

\subsection{Proof of Proposition \ref{oaproposition} and Proposition \ref{alphaproposition}}\label{proofpro1pro2}
We denote by $\Psi = \{\tx_i-\tx_j:i,j\in [n]\}$. By the definition of $\cL(\cdot)$ and $\ocL(\cdot)$ we know that for any $\tb\in \mathbb{R}^{n+d}$ and $\boldsymbol{v}\in \Psi^\perp$, we have
\begin{align*}
    \cL(\tb)  = \cL(\tb+\boldsymbol{v}) \text{ and }\ocL(\tb)  = \ocL(\tb+\boldsymbol{v}).
\end{align*}
On the other hand, under Assumption \ref{ass1} we have the following lemma.
\begin{lemma}
	 For any $\tb\in \mathbb{R}^{n+d}$, there exists a $\boldsymbol{v}\in\Psi^{\perp}$ such that $\tb+\boldsymbol{v}\in \Theta$.
\end{lemma}
\begin{proof}
We only have to show that for any $\tb\in\mathbb{R}^{n+d}$, $\tb+\Psi^\perp\bigcap\Theta\neq\emptyset$. Assume there exists a $\tb\in \mathbb{R}^{n+d}$ such that $\tb+\Psi^\perp\bigcap\Theta = \emptyset$. First of all, we must have $\boldsymbol{Z}^\top\tb\notin \boldsymbol{Z}^\top\Psi^\perp$. Since $\boldsymbol{Z}^\top\Psi^\perp\subset \mathbb{R}^{d+1}$, we must have $\text{dim}(\boldsymbol{Z}^\top\Psi^\perp)\leq d$ as $\boldsymbol{Z}^\top\Psi^\perp\neq \mathbb{R}^{d+1}$. Since $\text{dim}(\Psi^\perp) = d+1$, we know that there exists a non-zero vector $\boldsymbol{v}\in\Psi^\perp$ such that $\boldsymbol{Z}^\top\boldsymbol{v}=\boldsymbol{0}$. By the definition of $\Theta$, we know that $\boldsymbol{v}\in \Theta$. Recall the definition of $\boldsymbol{\Sigma}$ and Assumption \ref{ass1} in \S \ref{Consistency}. Since $\boldsymbol{v}\in \Psi^\perp$, we know that $\boldsymbol{\Sigma}\boldsymbol{v} = 0$, so we must have $\lambda_{\text{min},\perp}(\boldsymbol{\Sigma}) = 0$. As a result, this contradicts to Assumption \ref{ass1} since $c_2 = 0$.
\end{proof}
With this lemma, we then turn to prove Proposition \ref{oaproposition} and \ref{alphaproposition}.
\begin{proof}[Proof of Proposition \ref{oaproposition} and \ref{alphaproposition}]
Assume there exists a $z$ such that $\ocL|_{\ob_{-i}}(z)<\ocL|_{\ob_{-i}}(\oa_i)$. Then we let $\boldsymbol{w}\in \mathbb{R}^{n+d}$ be the vector such that $\boldsymbol{w}_{-i} = \ob_{-i}$ and $w_i = z$. And, let $\boldsymbol{v}$ be the vector in $\Psi^\perp$ such that $\boldsymbol{w}+\boldsymbol{v}\in\Theta$. Then we have
\begin{align*}
    \ocL(\boldsymbol{w}+\boldsymbol{v}) = \ocL(\boldsymbol{w}) = \ocL|_{\ob_{-i}}(z)<\ocL|_{\ob_{-i}}(\oa_i) = \ocL(\ob).
\end{align*}
This contradicts to the definition of $\ob$ Eq.~\eqref{eq11}.

Similarly, if we assume that there exists a $z$ such that $\cL|_{\tb_{M,-i}}(z)<\cL|_{\tb_{M,-i}}(x)(\wh\alpha_{M,i})$. Then we let $\boldsymbol{w}\in \mathbb{R}^{n+d}$ be the vector such that $\boldsymbol{w}_{-i} = \tb_{M,-i}$ and $w_i = z$. And, let $\boldsymbol{v}$ be the vector in $\Psi^\perp$ such that $\boldsymbol{w}+\boldsymbol{v}\in\Theta$. Then we have
\begin{align*}
    \cL(\boldsymbol{w}+\boldsymbol{v}) = \cL(\boldsymbol{w}) = \cL|_{\tb_{M,-i}}(z)<\cL|_{\tb_{M,-i}}(\wh\alpha_{M,i}) = \cL(\tb_M).
\end{align*}
This contradicts to the definition of $\tb_M$ Eq.~\eqref{MLEun}.
\end{proof}

\subsection{Auxiliary Results for Proving Lemma \ref{alphainferencelem1}}

In this section we include two results which are helpful to the proof of Lemma \ref{alphainferencelem1} in \S \ref{proofalphainferencelem1}. These two results are analogies of Theorem \ref{noregularization} and Theorem \ref{inferencemainthm} which we have proven before. The main difference is we replace $\cL(\cdot)$ with $\cL^{(i)}(\cdot)$. As a result, the following two results can be viewed as the leave-one-out version of Theorem \ref{noregularization} and Theorem \ref{inferencemainthm}. To be more specific, we define $\tb_M^{(i)}$ as
\begin{align*}
    \tb_M^{(i)} &= \argmin_{\tb\in \Theta}\cL^{(i)}(\tb)
\end{align*}
and let $\ob^{(i)}$ be the solution of the following equations
\begin{align*}
\left\{ \begin{array}{ll}
&\mathcal{P}\nabla\cL^{(i)}(\tb^*)+\mathcal{P}\nabla^2\cL^{(i)}(\tb^*)\left(\ob^{(i)}-\tb^*\right) = \boldsymbol{0};  \\
&\mathcal{P}\ob^{(i)} = \ob^{(i)}.
\end{array}
\right.
\end{align*}

\begin{lemma}\label{auxlemma1}
Suppose $np > c_p\log n$ for some $c_p>0$ and $d+1< n, (d+1)\log n\lesssim np$. We consider $L \leq  c_4\cdot n^{c_5}$ for any absolute constants $c_4,c_5>0$. Then for every $i\in [n]$ and $\tb_M^{(i)} = (\wh\ba_M^{(i)\top},\wh\bb_M^{(i)\top})^\top$, with probability at least $1-O(n^{-6})$ we have
\begin{align*}
\Vert \wh\ba_M^{(i)}-\ba^*\Vert_\infty\lesssim \kappa_1^2\sqrt{\frac{(d+1)\log n}{npL}}, \left\Vert \wh\bb_M^{(i)}-\bb^*\right\Vert_2\lesssim\kappa_1\sqrt{\frac{\log n}{pL}}\text{ and }\left\Vert \tb_M^{(i)}-\tb_M\right\Vert_2\lesssim \kappa_1\sqrt{\frac{(d+1)\log n}{npL}}.
\end{align*}

\end{lemma}

\begin{lemma}\label{auxlemma2}
Under the assumptions of Theorem \ref{auxlemma1}, for every $i\in [n]$, with probability at least $1-O(n^{-6})$ we have
\begin{align*}
    \left\Vert \tb_M^{(i)}-\ob^{(i)}\right\Vert_2\lesssim\kappa_1^4\frac{(d+1)^{0.5}\log n}{\sqrt{n}pL}.
\end{align*}
\end{lemma}
 %The core of the proof is the nonzero eigenvalues of the Hessian $\nabla^2 \cL^{(i)}(\tb^*)$ enjoys the same upper bound and lower bound as $\nabla^2 \cL(\tb^*)$, as we stated below.

The proof of Lemma \ref{auxlemma1} and Lemma \ref{auxlemma2} are almost the same as the previous results so we omit the proof details here. One can show Lemma \ref{auxlemma1} by mimicing the proof in \S \ref{regularizedtoriginal}, \S \ref{prooflem5} and the results in Lemma \ref{induction3}. In addition, Lemma \ref{auxlemma2} can be proved by mimicing the proof in \S \ref{proveinferencemainthm}.

\subsection{Proof of Lemma \ref{alphainferencelem1}} \label{proofalphainferencelem1}
\begin{proof}
Since $\oa_i$ can be expressed as
\begin{align*}
    \oa_i = \alpha_i^*-\frac{\left(\nabla\cL(\tb^*)\right)_i+\sum\limits_{j\neq i}\left(\ob_j-\tb^*_j\right)\left(\nabla^2\cL(\tb^*)\right)_{i,j}}{\left(\nabla^2\cL(\tb^*)\right)_{i,i}},
\end{align*}
we have
\begin{align}
    \oa_i'-\oa_i = \frac{\sum\limits_{j\neq i}\left(\tb_{M,j}-\ob_j\right)\left(\nabla^2\cL(\tb^*)\right)_{i,j}}{\left(\nabla^2\cL(\tb^*)\right)_{i,i}}.\label{alphaprimealpha}
\end{align}
We decompose Eq.~\eqref{alphaprimealpha} as
\begin{align*}
    \oa_i'-\oa_i &= \underbrace{\frac{\sum\limits_{j\neq i}\left(\tb_{M,j}-\tb_{M,j}^{(i)}\right)\left(\nabla^2\cL(\tb^*)\right)_{i,j}}{\left(\nabla^2\cL(\tb^*)\right)_{i,i}}}_{A_1}+\underbrace{\frac{\sum\limits_{j\neq i}\left(\tb_{M,j}^{(i)}-\ob_j^{(i)}\right)\left(\nabla^2\cL(\tb^*)\right)_{i,j}}{\left(\nabla^2\cL(\tb^*)\right)_{i,i}} }_{A_2} \\
    &+\underbrace{\frac{\sum\limits_{j\neq i}\left(\ob_j^{(i)}-\ob_j\right)\left(\nabla^2\cL(\tb^*)\right)_{i,j}}{\left(\nabla^2\cL(\tb^*)\right)_{i,i}}}_{A_3}.
\end{align*}

Next, we bound $A_1$-$A_3$ one by one. Before proceeding, the denominator $\left(\nabla^2\cL(\tb^*)\right)_{i,i}$ can be bounded as 
\begin{align*}
    \left(\nabla^2\cL(\tb^*)\right)_{i,i} &= \sum_{j\neq i}\phi'(\tx_i^\top\tb^*-\tx_j^\top\tb^*)\textbf{1}((i,j)\in\mathcal{E}) \\
    &\gtrsim \frac{1}{\kappa_1}\sum_{j\neq i}\textbf{1}((i,j)\in\mathcal{E})\gtrsim \frac{np}{\kappa_1}
\end{align*}
with probability at least $1-O(n^{-11})$. 

For $A_1$, by Lemma \ref{auxlemma1} we have
\begin{align*}
    \left\Vert \tb_{M}-\tb_{M}^{(i)}\right\Vert_2 \lesssim \kappa_1\sqrt{\frac{(d+1)\log n}{npL}}.
\end{align*}
So the numerator of $A_1$ can be bounded as
\begin{align*}
    \left|\sum\limits_{j\neq i}\left(\tb_{M,j}-\tb_{M,j}^{(i)}\right)\left(\nabla^2\cL(\tb^*)\right)_{i,j}\right|\leq & \left\Vert\tb_{M}-\tb_{M}^{(i)}\right\Vert_2 \sqrt{\sum_{j\neq i}\left(\nabla^2\cL(\tb^*)\right)_{i,j}^2} \\
    \lesssim & \kappa_1\sqrt{\frac{(d+1)\log n}{npL}}\sqrt{(d+1)np}\lesssim\kappa_1 (d+1)\sqrt{\frac{\log n}{L}}
\end{align*}
with probability at least $1-O(n^{-10})$. As a result, $A_1$ can be bounded as
\begin{align}
    \left|A_1\right|\lesssim \frac{\kappa_1 (d+1)\sqrt{\frac{\log n}{L}}}{np/\kappa_1}\leq \kappa_1^2\frac{d+1}{np}\sqrt{\frac{\log n}{L}}\label{A1}
\end{align}
with probability at least $1-O(n^{-10})$.

When it comes to $A_2$, by definition we know that  $\tb_{M,j}^{(i)}-\ob_j^{(i)}$ is independent with $\left(\nabla^2\cL(\tb^*)\right)_{i,j}$ for all $j\in [n+d]$. As a result, by Bernstein's inequality and Lemma \ref{auxlemma2} we know that
\begin{align}
    &\left|\sum\limits_{j\in [n],j\neq i}\left(\tb_{M,j}^{(i)}-\ob_j^{(i)}\right)\left(\nabla^2\cL(\tb^*)\right)_{i,j} - \mathbb{E}\left[\sum\limits_{j\in [n],j\neq i}\left(\tb_{M,j}^{(i)}-\ob_j^{(i)}\right)\left(\nabla^2\cL(\tb^*)\right)_{i,j}\bigg| \tb^{(i)}_M,\ob^{(i)} \right]\right| \\
    \lesssim & \sqrt{(\log n) \sum\limits_{j\in [n],j\neq i}\left(\tb_{M,j}^{(i)}-\ob_j^{(i)}\right)^2\bE\left(\nabla^2\cL(\tb^*)\right)_{i,j}^2}+ (\log n)\left\Vert\wh\ba_M^{(i)}-\oba^{(i)}\right\Vert_\infty \\
    \lesssim & \sqrt{p\log n}\left\Vert\wh\ba_M^{(i)}-\oba^{(i)}\right\Vert_2+ (\log n)\left\Vert\wh\ba_M^{(i)}-\oba^{(i)}\right\Vert_\infty \leq \sqrt{p\log n}\left\Vert\tb_M^{(i)}-\ob^{(i)}\right\Vert_2+ (\log n)\left\Vert\wh\ba_M^{(i)}-\oba^{(i)}\right\Vert_\infty\\
    \lesssim & \kappa_1^4\frac{(d+1)^{0.5}(\log n)^{1.5}}{\sqrt{np}L}+ (\log n)\left\Vert\wh\ba_M^{(i)}-\oba^{(i)}\right\Vert_\infty \label{A21}
\end{align}
with probability at least $1-O(n^{-6})$. On the other hand, by Lemma \ref{auxlemma2} we have
\begin{align}
    &\left|\mathbb{E}\left[\sum\limits_{j\in [n],j\neq i}\left(\tb_{M,j}^{(i)}-\ob_j^{(i)}\right)\left(\nabla^2\cL(\tb^*)\right)_{i,j}\bigg| \tb^{(i)}_M,\ob^{(i)} \right]\right| \\
    \leq& \sqrt{\sum\limits_{j\in [n],j\neq i}\left(\tb_{M,j}^{(i)}-\ob_j^{(i)}\right)^2}\sqrt{\sum\limits_{j\in [n],j\neq i}\left(\mathbb{E}\left(\nabla^2\cL(\tb^*)\right)_{i,j}\right)^2}  \\
    \lesssim& p\sqrt{n}\left\Vert\tb_M^{(i)}-\ob^{(i)}\right\Vert_2\lesssim \kappa_1^4\frac{(d+1)^{0.5}\log n}{L}\label{A22}
\end{align}
and
\begin{align}
    \left|\sum\limits_{k>i}\left(\tb_{M,k}^{(i)}-\ob_k^{(i)}\right)\left(\nabla^2\cL(\tb^*)\right)_{i,k}\right|\leq \left\Vert\tb_M^{(i)}-\ob^{(i)}\right\Vert_2 \sqrt{\sum_{k>i}\left(\nabla^2\cL(\tb^*)\right)_{i,k}^2}\lesssim \kappa_1^4\frac{(d+1)\log n}{L}\label{A23}
\end{align}
with probability exceeding $1-O(n^{-6})$. Combine Eq.~\eqref{A21}, Eq.~\eqref{A22} and Eq.~\eqref{A23} we finally bound the numerator of $A_2$ as
\begin{align*}
    \left|\sum\limits_{j\neq i}\left(\tb_{M,j}^{(i)}-\ob_j^{(i)}\right)\left(\nabla^2\cL(\tb^*)\right)_{i,j}\right|\lesssim \kappa_1^4 \frac{(d+1)\log n}{L}+(\log n)\left\Vert\wh\ba_M^{(i)}-\oba^{(i)}\right\Vert_\infty
\end{align*}
 with probability at least $1-O(n^{-6})$. As a result, we have
 \begin{align}
     |A_2|\lesssim \kappa_1^5 \frac{(d+1)\log n}{npL}+\kappa_1\frac{\log n}{np}\left\Vert\wh\ba_M^{(i)}-\oba^{(i)}\right\Vert_\infty \label{A2tmp}
 \end{align}
 with probability at least $1-O(n^{-6})$.
 
We finally consider bounding $A_3$. By definition, we know that
\begin{align*}
\left\{ \begin{array}{ll}
&\mathcal{P}\nabla\cL(\tb^*)+\mathcal{P}\nabla^2\cL(\tb^*)\left(\ob-\tb^*\right) = \boldsymbol{0};  \\
&\mathcal{P}\nabla\cL^{(i)}(\tb^*)+\mathcal{P}\nabla^2\cL^{(i)}(\tb^*)\left(\ob^{(i)}-\tb^*\right) = \boldsymbol{0}.
\end{array}
\right.
\end{align*}
Combine the two equations we get
\begin{align}
    \underbrace{\mathcal{P}\left(\nabla\cL(\tb^*)-\nabla\cL^{(i)}(\tb^*)\right)}_{\boldsymbol{w}_1}+\underbrace{\mathcal{P}\left(\nabla^2\cL(\tb^*) - \nabla^2\cL^{(i)}(\tb^*)\right)\left(\ob-\tb^*\right)}_{\boldsymbol{w}_2}+\mathcal{P}\nabla^2\cL^{(i)}(\tb^*)\left(\ob-\ob^{(i)}\right) = \boldsymbol{0}.\label{obobidiff}
\end{align}

For $\boldsymbol{w}_1$, it is easy to see $\Vert\mathcal{P}(\nabla\cL(\tb^*)-\nabla\cL^{(i)}(\tb^*))\Vert_2\leq \Vert\nabla\cL(\tb^*)-\nabla\cL^{(i)}(\tb^*)\Vert_2 $. By definition we have
\begin{align*}
	 \nabla \cL(\tb^*) - \nabla \cL^{(i)}(\tb^*) &= \sum_{j\neq i}(-y_{j,i}+\phi(\tx_i^\top\tb^*-\tx_j^\top\tb^*))\textbf{1}((i,j)\in\mathcal{E})(\tx_i-\tx_j).
\end{align*}
For  $(\nabla \cL(\tb^*) - \nabla \cL^{(i)}(\tb^*))_i$, by Hoeffding inequality, we have
\begin{align*}
	&\left| (\nabla \cL(\tb^*) - \nabla \cL^{(i)}(\tb^*))_i \right| = \frac{1}{L}\left|\sum_{(i,j)\in \mathcal{E}}\sum_{l=1}^L(-y_{j,i}^{(l)}+\phi(\tx_i^\top\tb^*-\tx_j^\top\tb^*))\textbf{1}((i,j)\in\mathcal{E})\right| \\
	\lesssim& \frac{1}{L}\sqrt{L\sum_{j\neq i}\textbf{1}((i,j)\in\mathcal{E})\log n}
\end{align*} 
conditioned on $\mathcal{G}$ with probability at least $1-O(n^{-10})$. On the other hand, since $\sum_{j\neq i, j\in [n]}\textbf{1}((i,j)\in\mathcal{E})\lesssim np$ with probability at least $1-O(n^{-10})$, we have 
\begin{align*}
	\left| (\nabla \cL(\tb^*) - \nabla \cL^{(i)}(\tb^*))_i \right| 
	\lesssim \sqrt{\frac{np\log n}{L}}
\end{align*} 
with probability at least $1-O(n^{-10})$. Furthermore, by Lemma \ref{Lilemma} we have
\begin{align*}
	\sum_{j\neq i, j\in [n]}\left(\nabla \cL(\tb^*) - \nabla \cL^{(i)}(\tb^*)\right)_j^2  &= \sum_{j\neq i, j\in [n]}\left((-y_{j,i}+\phi(\tx_i^\top\tb^*-\tx_j^\top\tb^*))\textbf{1}((i,j)\in\mathcal{E})\right)^2 \\
	&\lesssim \frac{\log n}{L} \sum_{j\neq i}\textbf{1}((i,j)\in\mathcal{E})\lesssim \frac{np\log n}{L}
\end{align*}
with probability at least $1-O(n^{-10})$. For $( \nabla \cL(\tb^*) - \nabla \cL^{(i)}(\tb^*))_{n+1:n+d}$, we have
\begin{align*}
	\left\Vert\left( \nabla \cL(\tb^*) - \nabla \cL^{(i)}(\tb^*)\right)_{n+1:n+d}\right\Vert_2 &= \left\Vert\sum_{j\neq i}(-y_{j,i}+\phi(\tx_i^\top\tb^*-\tx_j^\top\tb^*))\textbf{1}((i,j)\in\mathcal{E})(\bx_i-\bx_j)\right\Vert_2  \\
	&\leq  \sum_{j\neq i}\textbf{1}((i,j)\in\mathcal{E})\left|-y_{j,i}+\phi(\tx_i^\top\tb^*-\tx_j^\top\tb^*)\right|\left\Vert\bx_i-\bx_j\right\Vert_2 \\
	&\lesssim np\sqrt{\frac{\log n}{L}}\sqrt{\frac{d+1}{n}}\leq \sqrt{\frac{(d+1)np\log n}{L}}
\end{align*}
with probability at least $1-O(n^{-10})$. Therefore, for $\boldsymbol{w}_1$ we have
\begin{align}
	\left\Vert \boldsymbol{w}_1\right\Vert_2^2\leq \left\Vert \nabla \cL(\tb^*) - \nabla \cL^{(i)}(\tb^*)\right\Vert_2^2\leq&  \left(\nabla \cL(\tb^*) - \nabla \cL^{(i)}(\tb^*)\right)_i ^2  +\sum_{j\neq i, j\in [n]}\left(\nabla \cL(\tb^*) - \nabla \cL^{(i)}(\tb^*)\right)_j^2\\
	&+	\left\Vert\left( \nabla \cL(\tb^*) - \nabla \cL^{(i)}(\tb^*)\right)_{n+1:n+d}\right\Vert_2^2\lesssim \frac{(d+1)np\log n}{L}\label{w1}
\end{align}
with probability exceeding $1-O(n^{-10})$.

For $\boldsymbol{w}_2$, since $\nabla^2 \cL(\tb^*) - \nabla^2 \cL^{(i)}(\tb^*) = \sum_{j\neq i}(\mathbf{1}((i,j)\in \mathcal{E})-p)\phi'(\tx_i^\top\tb^*-\tx_j^\top\tb^*)(\tx_i-\tx_j)(\tx_i-\tx_j)^\top$, it holds that
\begin{align*}
    &\left(\nabla^2\cL(\tb^*) - \nabla^2\cL^{(i)}(\tb^*)\right)\left(\ob-\tb^*\right) \\
    =& \sum_{j\neq i}(\mathbf{1}((i,j)\in \mathcal{E})-p)\phi'(\tx_i^\top\tb^*-\tx_j^\top\tb^*)\left((\tx_i-\tx_j)^\top\left(\ob-\tb^*\right)  \right)(\tx_i-\tx_j).
\end{align*}
As a result, it follows that
\begin{align*}
    &\left\Vert\mathcal{P}\left(\nabla^2\cL(\tb^*) - \nabla^2\cL^{(i)}(\tb^*)\right)\left(\ob-\tb^*\right)\right\Vert_2 \leq\left\Vert \left(\nabla^2\cL(\tb^*) - \nabla^2\cL^{(i)}(\tb^*)\right)\left(\ob-\tb^*\right)\right\Vert_2 \\
    \leq &\sum_{j\neq i}\left\Vert (p-\mathbf{1}((i,j)\in \mathcal{E}))\phi'(\tx_i^\top\tb^*-\tx_j^\top\tb^*)\left((\tx_i-\tx_j)^\top\left(\ob-\tb^*\right)  \right)(\tx_i-\tx_j)\right\Vert_2 \\
    \leq & \sum_{j\neq i} \left|p-\mathbf{1}((i,j)\in \mathcal{E})\right|\left|(\tx_i-\tx_j)^\top\left(\ob-\tb^*\right)  \right|\left\Vert\tx_i-\tx_j\right\Vert_2 \\
    \lesssim & \sum_{j\neq i} \left|p-\mathbf{1}((i,j)\in \mathcal{E})\right|\left\Vert\ob-\tb^*\right\Vert_c \lesssim np \left\Vert\ob-\tb^*\right\Vert_c
\end{align*}
with probability exceeding $1-O(n^{-10})$. On the other hand, we obtain
\begin{align*}
     \left\Vert\tb^*-\ob\right\Vert_c&\leq \left\Vert\tb^*-\tb_M\right\Vert_c+ \left\Vert\tb_M-\ob\right\Vert_c \\
    &\lesssim \kappa_1^2\sqrt{\frac{(d+1)\log n}{npL}} + \left\Vert\wh\ba_M-\oba\right\Vert_\infty +\sqrt{\frac{d+1}{n}}\kappa_1^4\frac{(d+1)^{0.5}\log n}{\sqrt{n}pL} \\
    &\lesssim \kappa_1^2\sqrt{\frac{(d+1)\log n}{npL}} + \left\Vert\wh\ba_M-\oba\right\Vert_\infty
\end{align*}
with probability at least $1-O(n^{-6})$. That is to say, we have 
\begin{align}
    \left\Vert \boldsymbol{w}_2\right\Vert_2\lesssim \kappa_1^2\sqrt{\frac{(d+1)np\log n}{L}}+np\left\Vert\wh\ba_M-\oba\right\Vert_\infty\label{w2}
\end{align}
with probability at least $1-O(n^{-6})$. 

Combine Eq.~\eqref{w1} and Eq.~\eqref{w2} with Eq.~\eqref{obobidiff}, we know that
\begin{align*}
    \left\Vert \mathcal{P}\nabla^2\cL^{(i)}(\tb^*)\left(\ob-\ob^{(i)}\right)\right\Vert_2 = \left\Vert \boldsymbol{w}_1+\boldsymbol{w}_2\right\Vert_2\leq \left\Vert \boldsymbol{w}_1\right\Vert_2+ \left\Vert \boldsymbol{w}_2\right\Vert_2\lesssim\kappa_1^2\sqrt{\frac{(d+1)np\log n}{L}}+np\left\Vert\wh\ba_M-\oba\right\Vert_\infty
\end{align*}
with probability at least $1-O(n^{-6})$. Since $\ob-\ob^{(i)}\in \Theta$, we have
%\begin{align*}
%    &\lambda_{\text{min}, \perp}\left(\nabla^2\cL^{(i)}(\tb^*)\right)\left\Vert \ob-\ob^{(i)}\right\Vert_2^2\leq \lambda_{\text{min}, \perp}\left(\nabla^2\cL^{(i)}(\tb^*)\right)\left\Vert\mathcal{P}\left( \ob-\ob^{(i)}\right)\right\Vert_2^2 \\
%    \leq & \left( \ob-\ob^{(i)}\right)^\top \mathcal{P}\nabla^2\cL^{(i)}(\tb^*)\mathcal{P}\left( \ob-\ob^{(i)}\right) = \left( \ob-\ob^{(i)}\right)^\top \mathcal{P}\nabla^2\cL^{(i)}(\tb^*)\left( \ob-\ob^{(i)}\right) \\
%    \leq & \left\Vert \ob-\ob^{(i)}\right\Vert_2 \left\Vert \mathcal{P}\nabla^2\cL^{(i)}(\tb^*)\left(\ob-\ob^{(i)}\right)\right\Vert_2.
%\end{align*}
\begin{align*}
    &\lambda_{\text{min}, \perp}\left(\nabla^2\cL^{(i)}(\tb^*)\right)\left\Vert \ob-\ob^{(i)}\right\Vert_2^2\leq \left( \ob-\ob^{(i)}\right)^\top \nabla^2\cL^{(i)}(\tb^*)\left( \ob-\ob^{(i)}\right) \\
    = & \left( \ob-\ob^{(i)}\right)^\top \mathcal{P}\nabla^2\cL^{(i)}(\tb^*)\left( \ob-\ob^{(i)}\right) \leq  \left\Vert \ob-\ob^{(i)}\right\Vert_2 \left\Vert \mathcal{P}\nabla^2\cL^{(i)}(\tb^*)\left(\ob-\ob^{(i)}\right)\right\Vert_2.
\end{align*}
As a result, we know that
\begin{align}
    \left\Vert \ob-\ob^{(i)}\right\Vert_2\leq\frac{\left\Vert \mathcal{P}\nabla^2\cL^{(i)}(\tb^*)\left(\ob-\ob^{(i)}\right)\right\Vert_2}{\lambda_{\text{min}, \perp}\left(\nabla^2\cL^{(i)}(\tb^*)\right)}\lesssim\kappa_1^3\sqrt{\frac{(d+1)\log n}{npL}} +\kappa_1 \left\Vert\wh\ba_M-\oba\right\Vert_\infty\label{obleaveoneout}
\end{align}
with probability at least $1-O(n^{-6})$. Therefore, for $A_3$ we finally achieve
\begin{align}
    |A_3| &\lesssim \left|\frac{\sum\limits_{j\neq i}\left(\ob^{(i)}_j-\ob_j\right)\left(\nabla^2\cL(\tb^*)\right)_{i,j}}{np/\kappa_1}\right| \leq \frac{\left\Vert\ob^{(i)}-\ob\right\Vert_2 \sqrt{\sum_{j\neq i}\left(\nabla^2\cL(\tb^*)\right)_{i,j}^2}}{np/\kappa_1} \\
    &\lesssim  \kappa_1^4\frac{d+1}{np}\sqrt{\frac{\log n}{L}}+\kappa_1^2\sqrt{\frac{d+1}{np}}\left\Vert\wh\ba_M-\oba\right\Vert_\infty\label{A3}
\end{align}
with probability at least $1-O(n^{-6})$.

And, by Eq.~\eqref{obleaveoneout} we know that
\begin{align*}
    \left\Vert\wh\ba_M^{(i)}-\oba^{(i)}\right\Vert_\infty&\leq \left\Vert\tb_M^{(i)}-\tb_M\right\Vert_2+\left\Vert\wh\ba_M-\oba\right\Vert_\infty+\left\Vert\ob-\ob^{(i)}\right\Vert_2 \\
    &\lesssim \kappa_1^3\sqrt{\frac{(d+1)\log n}{npL}} +\kappa_1 \left\Vert\wh\ba_M-\oba\right\Vert_\infty
\end{align*}
with probability exceeding $1-O(n^{-6})$. Combine with Eq.~\eqref{A2tmp} we have
\begin{align}
    |A_2|\lesssim \kappa_1^5 \frac{(d+1)\log n}{npL} +\kappa_1^4\frac{\log n}{np}\sqrt{\frac{(d+1)\log n}{npL}}+\kappa_1^2\frac{\log n}{np} \left\Vert\wh\ba_M-\oba\right\Vert_\infty\label{A2}
\end{align}
with probability exceeding $1-O(n^{-6})$. Combine Eq.~\eqref{A1}, Eq.~\eqref{A2} and Eq.~\eqref{A3} we know that
\begin{align*}
    |\oa_i'-\oa_i|\lesssim \kappa_1^5 \frac{(d+1)\log n}{npL} +\frac{\kappa_1^4}{np}\sqrt{\frac{(d+1)\log n}{L}}\left(\sqrt{d+1}+\frac{\log n}{\sqrt{np}}\right)+\kappa_1^2\left(\sqrt{\frac{d+1}{np}}+\frac{\log n}{np}\right)\left\Vert\wh\ba_M-\oba\right\Vert_\infty
\end{align*}
with probability at least $1-O(n^{-6})$.

\end{proof}

\subsection{Proof of Lemma \ref{alphainferencelem2}}\label{proof_lemb3}
\begin{proof}
Since $\wh\alpha_{M,i}$ is the minimizer of $\cL\bigg|_{\tb_{M,-i}}(\cdot)$, we know that $\displaystyle\left(\cL\bigg|_{\tb_{M,-i}}\right)'(\wh\alpha_{M,i})=0$. By the mean value theorem we know that
\begin{align*}
    \left(\cL\bigg|_{\tb_{M,-i}}\right)'(\wh\alpha_{M,i}) = \left(\cL\bigg|_{\tb_{M,-i}}\right)'(\alpha^*_{i})+\left(\cL\bigg|_{\tb_{M,-i}}\right)''(b_1)(\wh\alpha_{M,i}-\alpha^*_{i}),
\end{align*}
where $b_1$ is some real number between $\alpha^*_{i}$ and $\wh\alpha_{M,i}$. As a result, we have
\begin{align*}
    \wh\alpha_{M,i} = \alpha^*_{i}-\frac{\left(\cL\bigg|_{\tb_{M,-i}}\right)'(\alpha^*_{i})}{\left(\cL\bigg|_{\tb_{M,-i}}\right)''(b_1)}.
\end{align*}
By the definition Eq.~\eqref{nll} and Eq.~\eqref{Lconditioned}, we have
\begin{align*}
    \left(\cL\bigg|_{\tb_{M,-i}}\right)'(x) &= \sum_{j\neq i,(i,j)\in \mathcal{E}}\left\{-y_{j,i}+\phi(\tx_i^\top\tb_M-\tx_j^\top\tb_M+x-\wh\alpha_{M,i})\right\} \\
    \left(\cL\bigg|_{\tb_{M,-i}}\right)''(x) &= \sum_{j\neq i,(i,j)\in \mathcal{E}}\phi'(\tx_i^\top\tb_M-\tx_j^\top\tb_M+x-\wh\alpha_{M,i}).
\end{align*}
We first estimate the difference $\left(\cL\bigg|_{\tb_{M,-i}}\right)''(b_1)-\left(\nabla^2\cL(\tb^*)\right)_{i,i}$. We have
\begin{align*}
    &\left|\left(\cL\bigg|_{\tb_{M,-i}}\right)''(b_1)-\left(\nabla^2\cL(\tb^*)\right)_{i,i}\right| \\
    =& \left|\sum_{j\neq i,(i,j)\in \mathcal{E}}\phi'(\tx_i^\top\tb_M-\tx_j^\top\tb_M+b_1-\wh\alpha_{M,i})-\phi'(\tx_i^\top\tb^*-\tx_j^\top\tb^*)\right| \\
    \leq &\sum_{j\neq i,(i,j)\in \mathcal{E}}\left|(\tx_i^\top\tb_M-\tx_j^\top\tb_M+b_1-\wh\alpha_{M,i})-(\tx_i^\top\tb^*-\tx_j^\top\tb^*)\right| \\
    \lesssim & \sum_{j\neq i,(i,j)\in \mathcal{E}}\left(\left\Vert \tb_M-\tb^*\right\Vert_c+\left| b_1-\wh\alpha_{M,i} \right|\right)\\
    \lesssim & np\left\Vert \tb_M-\tb^*\right\Vert_c
\end{align*}
with probability at least $1-O(n^{-11})$. On the other hand, we have
\begin{align}
    &\left(\cL\bigg|_{\tb_{M,-i}}\right)'(\tb^*_{i})-\left(\left(\nabla\cL(\tb^*)\right)_i+\sum\limits_{j\neq i}\left(\tb_{M,j}-\tb^*_j\right)\left(\nabla^2\cL(\tb^*)\right)_{i,j}\right) \\
=&\sum_{j\neq i,(i,j)\in \mathcal{E}}\left\{-y_{j,i}+\phi(\tx_i^\top\tb_M-\tx_j^\top\tb_M+\alpha^*_i-\wh\alpha_{M,i})\right\}- \sum_{j\neq i,(i,j)\in \mathcal{E}}\left\{-y_{j,i}+\phi(\tx_i^\top\tb^*-\tx_j^\top\tb^*)\right\} \\
&-\sum\limits_{j\neq i,j\in[n+d]}\left(\tb_{M,j}-\tb^*_j\right)\left(\nabla^2\cL(\tb^*)\right)_{i,j}\\
=& \sum_{j\neq i,(i,j)\in \mathcal{E}}\left\{\phi(\tx_i^\top\tb_M-\tx_j^\top\tb_M+\alpha^*_i-\wh\alpha_{M,i})-\phi(\tx_i^\top\tb^*-\tx_j^\top\tb^*)-\left(\alpha^*_j-\wh\alpha_{M,j}\right)\phi'(\tx_i^\top\tb^*-\tx_j^\top\tb^*)\right\} \\
&-\sum_{k\in[d]}\left(\tb_{M,n+k}-\tb^*_{n+k}\right)\left(\sum_{j\neq i, (i,j)\in\mathcal{E}}\phi'(\tx_i^\top\tb^*-\tx_j^\top\tb^*)\left(\bx_i-\bx_j\right)_k\right) \\
=&\sum_{j\neq i, (i,j)\in \mathcal{E}} r_j,\label{eq13}
\end{align}
where 
\begin{align*}
    r_j =& \phi(\tx_i^\top\tb_M-\tx_j^\top\tb_M+\alpha^*_i-\wh\alpha_{M,i})-\phi(\tx_i^\top\tb^*-\tx_j^\top\tb^*)-\left(\alpha^*_j-\wh\alpha_{M,j}\right)\phi'(\tx_i^\top\tb^*-\tx_j^\top\tb^*) \\
    &- (\bx_i-\bx_j)^\top(\wh\bb_M-\bb^*)\phi'(\tx_i^\top\tb^*-\tx_j^\top\tb^*).
\end{align*}
By Taylor expansion we know that
\begin{align*}
    \phi(\tx_i^\top\tb_M-\tx_j^\top\tb_M+\alpha^*_i-\wh\alpha_{M,i}) =& \phi(\tx_i^\top\tb^*-\tx_j^\top\tb^*) \\
    &+ \phi'(\tx_i^\top\tb^*-\tx_j^\top\tb^*) \left((\bx_i-\bx_j)^\top(\wh\bb_M-\bb^*)+\alpha^*_j-\wh\alpha_{M,j}\right) \\
    &+ \phi''(b_2)\left((\bx_i-\bx_j)^\top(\wh\bb_M-\bb^*)+\alpha^*_j-\wh\alpha_{M,j}\right)^2,
\end{align*}
where $b_2$ is some real number between $\tx_i^\top\tb^*-\tx_j^\top\tb^*$ and $\tx_i^\top\tb_M-\tx_j^\top\tb_M+\alpha^*_i-\wh\alpha_{M,i}$. As a result, we have
\begin{align}
    |r_j|\leq |\phi''(b_2)|\left((\bx_i-\bx_j)^\top(\wh\bb_M-\bb^*)+\alpha^*_j-\wh\alpha_{M,j}\right)^2\lesssim \left\Vert \tb_M-\tb^*\right\Vert_c^2.\label{eq14}
\end{align}
Combine Eq.~\eqref{eq13} and Eq.~\eqref{eq14} we have
\begin{align*}
    \left|\left(\cL\bigg|_{\tb_{M,-i}}\right)'(\tb^*_{i})-\left(\left(\nabla\cL(\tb^*)\right)_i+\sum\limits_{j\neq i}\left(\tb_{M,j}-\tb^*_j\right)\left(\nabla^2\cL(\tb^*)\right)_{i,j}\right)\right|\lesssim np\left\Vert \tb_M-\tb^*\right\Vert_c^2
\end{align*}
with probability exceeding $1-O(n^{-11})$. As a result, we have
\begin{align*}
    \wh\alpha_{M,i}-\oa_i' =& \frac{\left(\nabla\cL(\tb^*)\right)_i+\sum\limits_{j\neq i}\left(\tb_{M,j}-\tb^*_j\right)\left(\nabla^2\cL(\tb^*)\right)_{i,j}}{\left(\nabla^2\cL(\tb^*)\right)_{i,i}}-\frac{\left(\cL\bigg|_{\tb_{M,-i}}\right)'(\alpha^*_{i})}{\left(\cL\bigg|_{\tb_{M,-i}}\right)''(b_1)} \\
     =& \underbrace{\frac{\left(\cL\bigg|_{\tb_{M,-i}}\right)''(b_1)-\left(\nabla^2\cL(\tb^*)\right)_{i,i}}{\left(\cL\bigg|_{\tb_{M,-i}}\right)''(b_1)\cdot \left(\nabla^2\cL(\tb^*)\right)_{i,i}}\left(\left(\nabla\cL(\tb^*)\right)_i+\sum\limits_{j\neq i}\left(\tb_{M,j}-\tb^*_j\right)\left(\nabla^2\cL(\tb^*)\right)_{i,j}\right)}_{A_1} \\
    &+\underbrace{\frac{\left(\cL\bigg|_{\tb_{M,-i}}\right)'(\alpha^*_{i})-\left(\left(\nabla\cL(\tb^*)\right)_i+\sum\limits_{j\neq i}\left(\tb_{M,j}-\tb^*_j\right)\left(\nabla^2\cL(\tb^*)\right)_{i,j}\right)}{\left(\cL\bigg|_{\tb_{M,-i}}\right)''(b_1)}}_{A_2},
\end{align*}
where 
\begin{align*}
    |A_1|&\leq \frac{np\left\Vert \tb_M-\tb^*\right\Vert_c}{\frac{np}{\kappa_1}\left(\frac{np}{\kappa_1}-\left\Vert \tb_M-\tb^*\right\Vert_c\right)}\left|\left(\nabla\cL(\tb^*)\right)_i+\sum\limits_{j\neq i}\left(\tb_{M,j}-\tb^*_j\right)\left(\nabla^2\cL(\tb^*)\right)_{i,j}\right|\\
    &\lesssim\frac{\kappa_1^4}{np}\sqrt{\frac{(d+1)\log n}{npL}}\left|\left(\nabla\cL(\tb^*)\right)_i+\sum\limits_{j\neq i}\left(\tb_{M,j}-\tb^*_j\right)\left(\nabla^2\cL(\tb^*)\right)_{i,j}\right|
\end{align*}
and 
\begin{align*}
    |A_2|\lesssim \frac{1}{\frac{np}{\kappa_1}-\left\Vert \tb_M-\tb^*\right\Vert_c}np\left\Vert \tb_M-\tb^*\right\Vert_c^2\lesssim \kappa_1^5\frac{(d+1)\log n}{npL}
\end{align*}
with probability exceeding $1-O(n^{-6})$. On the other hand, we know that
\begin{align*}
    &\left|\left(\nabla\cL(\tb^*)\right)_i+\sum\limits_{j\neq i}\left(\tb_{M,j}-\tb^*_j\right)\left(\nabla^2\cL(\tb^*)\right)_{i,j}\right| \\
    \leq& \left|\left(\nabla\cL(\tb^*)\right)_i\right|+\left|\sum\limits_{j\neq i}\left(\tb_{M,j}-\tb^*_j\right)\left(\nabla^2\cL(\tb^*)\right)_{i,j}\right| \\
    \leq& \left|\left(\nabla\cL(\tb^*)\right)_i\right|+\Vert\wh\ba_M-\ba^*\Vert_\infty\sum_{j\in [n],j\neq i}\left|\left(\nabla^2\cL(\tb^*)\right)_{i,j}\right|+\Vert\wh\bb_M-\bb^*\Vert_2 \sqrt{\sum_{k>n}\left(\nabla^2\cL(\tb^*)\right)_{i,k}^2}\\
    \lesssim& \sqrt{\frac{np\log n}{L}}+\kappa_1^2\sqrt{\frac{(d+1)\log n}{npL}}np+\kappa_1\sqrt{\frac{\log n}{pL}}\sqrt{dnp^2}\lesssim \kappa_1^2\sqrt{\frac{(d+1)np\log n}{L}}
\end{align*}
with probability at least $1-O(n^{-6})$. To sum up, we have
\begin{align*}
    \left|\alpha_{M,i}-\oa_i'\right|\leq |A_1|+|A_2|&\lesssim \frac{\kappa_1^4}{np}\sqrt{\frac{(d+1)\log n}{npL}}\kappa_1^2\sqrt{\frac{(d+1)np\log n}{L}}+\kappa_1^5\frac{(d+1)\log n}{npL} \\
    &\lesssim\kappa_1^6\frac{(d+1)\log n}{npL}
\end{align*}
with probability at least $1-O(n^{-6})$.
\end{proof}

\section{Details of Real Data Experiments}\label{detail_real_data}
In this section, we provide the details in computing the maximum likelihood estimator $(\wh\ba_M, \wh\bb_M)$. We first generated the variables $\bX$ and comparisons as described in \S \ref{real_data_pokeman}. We standardized each column to make sure they have mean $0$ and standard deviation $1$ and then multiplied by $2/27$ (scale to the order of $\sqrt{d+1/n}$ as mentioned in the main text). In real-world data the numbers of comparisons between each compared pair are not the same, so we denote by $L_{i,j}$ the number of comparisons between pair $(i,j)$. Let $\{ y_{i,j}^{(l)}:(i,j)\in \mathcal{E}, l\in [L_{i,j}]\}$ be all the comparisons we have, then the negative log-likelihood can be written as 
\begin{align*}
		\mathcal{L}(\tb): 
	&=\sum_{(i, j) \in \mathcal{E}, i>j}\sum_{l=1}^{L_{i,j}}\left\{-y_{j, i}^{(l)}\left(\tx_i^\top\tb-\tx_j^\top\tb\right)+\log \left(1+e^{\tx_i^\top\tb-\tx_j^\top\tb}\right)\right\}.
\end{align*}
\iffalse
We then scaled this log-likelihood by the total number of portfolios selected at this time point.  In addition, we added a small $\ell_2$ regularization term to ensure the algorithm converges. For the same reason we state in \S \ref{real_data_pokeman}, that is, we expect the residual scores $\ba$ to be close to $0$, we only imposed the  $\ell_2$ regularization term to the $\ba$ part. As a result, the objective that we minimized is 
\begin{align*}
	\frac{1}{|\text{number of selected portfolios}|}\mathcal{L}(\tb)+\frac{\lambda}{2}\Vert \ba\Vert_2^2.
\end{align*}
with $\lambda$ being $0.1$. We ran projected gradient descent with step size $3e-3$ and stopped the algorithm when the difference before and after updating has $\ell_2$ norm not greater than $1e-2$.
\fi

% Note: in this sample, the section number is hard-coded in. Following
% proper LaTeX conventions, it should properly be coded as a reference:

%In this appendix we prove the following theorem from
%Section~\ref{sec:textree-generalization}:

\vskip 0.2in
\bibliography{sample}

\end{document}